\documentclass[11pt]{article}

\newcommand{\version}{long}

\usepackage{algorithm,algpseudocode,amsmath,amssymb,amsthm,complexity,etoolbox,fullpage,graphicx,ifthen,mathtools,url,thmtools,thm-restate}
\usepackage[inline,shortlabels]{enumitem}

\usepackage{complexity}

\newclass{\ioPSPACE}{i.o.\text{-}PSPACE}
\newlang{\Halt}{Halt}
\ifundef{\soda}{\usepackage{etoolbox,hyperref,ifthen,tabulary}

\newcommand{\ifft}{\ifundef{\shortiff}{if and only if }{iff }}
\newcommand{\spc}{{ }}

\newcommand{\indraft}[1]{\ifthenelse{\equal{\version}{draft}}{#1}{}}
\newcommand{\infinal}[1]{\ifthenelse{\equal{\version}{final}}{#1}{Only shown in the final version}}
\newcommand{\inshort}[1]{\ifthenelse{\equal{\version}{short}}{#1}{}}
\newcommand{\inlong}[1]{\ifthenelse{\equal{\version}{long}}{#1}{}}
\newcommand{\inshortorlong}[2]{\inshort{#1}\inlong{#2}}
\newcommand{\inlongaln}{\inlong{\\ {} &}}

\newcommand{\insodaornot}[2]{\ifdef{\soda}{#1}{#2}}
\newcommand{\insoda}[1]{\insodaornot{#1}{}}
\newcommand{\notinsoda}[1]{\insodaornot{}{#1}}

\newcommand{\appref}[1]{Appendix~\ref{app:#1}}
\newcommand{\secref}[1]{Section~\ref{sec:#1}}
\newcommand{\ssecref}[1]{Subsection~\ref{ssec:#1}}

\newcommand{\renameenv}[2]{
  \expandafter\let\csname #1#2\expandafter\endcsname
  \csname #1\endcsname
  \expandafter\let\csname end#1#2\expandafter\endcsname
  \csname end#1\endcsname
  \expandafter\let\csname #2\endcsname\relax
  \expandafter\let\csname end#2\endcsname\relax}

\ifundef{\defaultlists}{
  \usepackage[inline,shortlabels]{enumitem}
  \setenumerate[1]{(a),itemsep=0pt,topsep=3pt,parsep=0pt,partopsep=0pt}
  \setenumerate[2]{(i),noitemsep,topsep=3pt,parsep=0pt,partopsep=0pt}
  \setenumerate[3]{(A),noitemsep,topsep=3pt,parsep=0pt,partopsep=0pt}
  \setenumerate[4]{(I),noitemsep,topsep=3pt,parsep=0pt,partopsep=0pt}
  \setitemize{noitemsep,topsep=3pt,parsep=0pt,partopsep=0pt}
  \setdescription{noitemsep,topsep=3pt,parsep=0pt,partopsep=0pt}
  \setlist{noitemsep,topsep=3pt,parsep=0pt,partopsep=0pt}}{}

\newcolumntype{x}[1]{>{\centering\arraybackslash}m{#1}}
}{} 
\usepackage{algpseudocode,amsmath,amssymb,etoolbox,fixltx2e,mathtools}

\let\eqref\relax

\DeclareFontFamily{U}{mathx}{\hyphenchar\font45}
\DeclareFontShape{U}{mathx}{m}{n}{
      <5> <6> <7> <8> <9> <10>
      <10.95> <12> <14.4> <17.28> <20.74> <24.88>
      mathx10
      }{}
\DeclareSymbolFont{mathx}{U}{mathx}{m}{n}
\DeclareMathSymbol{\bigtimes}{1}{mathx}{"91}

\algnewcommand{\Input}{\item[\textbf{Input:}]}
\algnewcommand{\Output}{\item[\textbf{Output:}]}
\MakeRobust{\Call} 

\newcommand\restr[2]{{
  \left.\kern-\nulldelimiterspace
  #1
  \vphantom{\big|}
  \right|_{#2}
  }}

\newcommand{\indphi}[1]{\restr{\phi}{#1}}
\newcommand{\abs}[1]{\left\vert#1\right\vert}
\newcommand{\can}{\mathrm{Can}}
\newcommand{\invar}{\mathrm{Inv}}
\newcommand{\cay}{\mathrm{Cay}}

\newcommand{\coleq}{\coloneqq}

\newcommand{\aut}{\mathrm{Aut}}

\newcommand{\iso}{\mathrm{Iso}}

\newcommand{\soc}{\mathrm{soc}}

\newcommand{\ch}{\;\mathrm{char}\;}

\newcommand{\norm}[1]{\left\Vert#1\right\Vert}

\newcommand{\genb}[2]{\left\langle#1 \;\middle|\; #2\right\rangle}
\newcommand{\setb}[2]{\left\{#1 \;\middle|\; #2\right\}}

\newcommand{\nth}[1]{\ensuremath{{#1}^{\mathrm{th}}}}

\renewcommand{\algref}[1]{Algorithm~\ref{alg:#1}}
\newcommand{\figref}[1]{Figure~\ref{fig:#1}}
\newcommand{\linref}[1]{line~\ref{line:#1}}
\newcommand{\itmref}[1]{\ref{itm:#1}}
\newcommand{\thmref}[1]{Theorem~\ref{thm:#1}}
\newcommand{\lemref}[1]{Lemma~\ref{lem:#1}}

\newcommand{\defref}[1]{Definition~\ref{defn:#1}}

\newcommand{\propref}[1]{Proposition~\ref{prop:#1}}
\newcommand{\corref}[1]{Corollary~\ref{cor:#1}}

\newcommand{\eqref}[1]{(\ref{eq:#1})}

\newcommand{\eps}{\epsilon}

\newcommand{\bmg}{\mathbf{g}}
\newcommand{\bmh}{\mathbf{h}}

\newcommand{\bmu}{\mathbf{u}}
\newcommand{\bmv}{\mathbf{v}}

\newcommand{\bmx}{\mathbf{x}}
\newcommand{\bmy}{\mathbf{y}}

\newcommand{\bmA}{\mathbf{A}}
\newcommand{\bmB}{\mathbf{B}}

\newcommand{\bbR}{\mathbb{R}}

\newcommand{\bbZ}{\mathbb{Z}}

\newcommand{\cA}{\mathcal{A}}
\newcommand{\cB}{\mathcal{B}}

\newcommand{\cG}{\mathcal{G}}
\newcommand{\cH}{\mathcal{H}}
\newcommand{\cI}{\mathcal{I}}

\newcommand{\cP}{\mathcal{P}}
\newcommand{\cQ}{\mathcal{Q}}

\newcommand{\cS}{\mathcal{S}}

\newcommand{\la}{\leftarrow}

\newcommand{\ra}{\rightarrow}

\newcommand{\tril}{\triangleleft}

\newcommand{\trile}{\trianglelefteq}

\inshort{\usepackage{amsthm}
\makeatletter
\def\thm@space@setup{\thm@preskip=3pt \thm@postskip=3pt}
\makeatother

\makeatletter
\renewenvironment{proof}[1][\proofname]{\par
  \pushQED{\qed}
  \normalfont
  \topsep3pt \partopsep0pt 
  \trivlist
  \item[\hskip\labelsep
        \itshape
    #1\@addpunct{.}]\ignorespaces
  }{
    \popQED\endtrivlist\@endpefalse
    \addvspace{0pt plus 0pt} 
  }
\makeatother
} 

\usepackage{etoolbox,ifthen,thmtools,thm-restate}

\ifundef{\dontnumberwithin}{\declaretheorem[numberwithin=section]{dummy}}{\declaretheorem{dummy}} 
\declaretheorem[sibling=dummy]{theorem}

\declaretheorem[sibling=dummy]{lemma}

\declaretheorem[sibling=dummy]{definition}

\declaretheorem[sibling=dummy]{proposition}
\declaretheorem[sibling=dummy]{corollary}

\ifundef{\defaultthmcontinues}{\renewcommand{\thmcontinues}[1]{}}{}

\inshort{\newcommand{\shortiff}{}}
\inshort{\usepackage[compact]{titlesec}}

\newcommand{\hphi}{\hat \phi}

\newcommand{\comps}{\inshortorlong{CS}{COMPOSITION-SERIES}}
\newcommand{\mns}{\inshortorlong{MNS}{MINIMAL-NORMAL-SUBGROUPS}}
\newcommand{\sims}{\inshortorlong{SMNS}{SIMPLE-SUBGROUPS}}

\begin{document}

\title{Breaking the $n^{\log n}$ Barrier for Solvable-Group Isomorphism}
\author{David J. Rosenbaum \\ {\small University of Washington} \\ {\small Department of Computer Science \& Engineering} \\ {\small Email: djr@cs.washington.edu}}
\date{\notinsoda{December 11, 2013}}

\maketitle
\notinsoda{\thispagestyle{empty}}

\begin{abstract}
  We consider the \emph{group isomorphism problem}: given two finite groups $G$ and $H$ specified by their multiplication tables, decide if $G \cong H$.  The $n^{\log n}$ barrier for group isomorphism has withstood all attacks --- even for the special cases of $p$-groups and solvable groups --- ever since the $n^{\log n + O(1)}$ generator-enumeration algorithm.  Following a framework due to Wagner, we present the first significant improvement over $n^{\log n}$ by reducing group isomorphism to composition-series isomorphism which is then reduced to low-degree graph isomorphism.  We show that group isomorphism is $n^{(1 / 2) \log_p n + O(1)}$ Turing reducible to composition-series isomorphism where $p$ is the smallest prime dividing the order of the group.  Combining our reduction with an $n^{O(p)}$ algorithm for $p$-group composition-series isomorphism, we obtain an $n^{(1 / 2) \log n + O(1)}$ algorithm for $p$-group isomorphism.  We then generalize our techniques from $p$-groups using Sylow bases to derive an $n^{(1 / 2) \log n + O(\log n / \log \log n)}$ algorithm for solvable-group isomorphism.  Finally, we relate group isomorphism to the collision problem which allows us replace the $1 / 2$ in the exponents with $1 / 4$ using randomized algorithms and $1 / 6$ using quantum algorithms.
\end{abstract}

\inlong{
  \newpage
  \setcounter{page}{1}}

\section{Introduction}
\label{sec:intro}
In the \emph{group isomorphism problem}, we are given two finite groups $G$ and $H$ of order $n$ in terms of their multiplication tables and must decide if $G \cong H$.  It is not known if group isomorphism is in \P\spc or if it is \NP-complete; the strongest complexity theory result is that solvable-group isomorphism is in $\NP \cap \coNP$ assuming that $\EXP \not\subseteq \ioPSPACE$~\cite{arvind2003a}.  Group isomorphism is Karp reducible to graph isomorphism~\cite{miller1979a} but is not known to be \GI-complete and may be easier due to the group structure~\cite{chattopadhyay2010a}.  A consequence of this reduction is that if group isomorphism is \NP-complete then the polynomial hierarchy collapses to the second level~\cite{boppana1987a}.  Moreover, the natural algebraic analogue of graph isomorphism is testing isomorphism of semigroups~\cite{booth1978a} which are much less structured than groups.

The \emph{generator-enumeration algorithm} for group isomorphism has been known for several decades (cf.~\cite{felsch1970a,miller1978a}); the analogous result that group isomorphism can be decided in $O(\log^2 n)$ space was independently obtained by Lipton, Snyder and Zalcstein~\cite{lipton1977a}.  The algorithm works by finding a minimal ordered generating set $\bmg$ for $G$ by brute force; any isomorphism from $G$ to $H$ can be defined by its action on an ordered generating set so we simply consider all $n^{\abs{S}}$ maps from $\bmg$ into $H$ and check if any of these extend to an isomorphism.  This algorithm runs in $n^{\abs{\bmg} + O(1)}$ time which is upper bounded by $n^{\log_p n + O(1)}$ since there is always a generating set of size $\log_p n$ where $p$ is the smallest prime dividing the order of the group.  These are still the best results for general groups.

Faster algorithms have been obtained for various special cases.  Lipton, Snyder and Zalcstein~\cite{lipton1977a}; Savage~\cite{savage1980a}; Vikas~\cite{vikas1996a}; and Kavitha~\cite{kavitha2007a} all showed polynomial-time algorithms for Abelian groups.  Le Gall~\cite{legall2008a} gave a polynomial-time algorithm for groups consisting of a semidirect product of an Abelian group with a cyclic group of coprime order; this was extended to a class of groups with a normal Hall subgroup by Qiao, Sarma and Tang~\cite{qiao2011a}.  Babai and Qiao~\cite{babai2012a} showed that testing isomorphism of groups with Abelian Sylow towers is in polynomial time.  Babai, Codenotti, Grochow and Qiao~\cite{babai2011a} showed an $n^{O(\log \log n)}$ time algorithm for the class of groups with no normal Abelian subgroups; the runtime was later improved to polynomial by Babai, Codenotti and Qiao~\cite{codenotti2011a,babai2012b}.

Our main result is an algorithm that is faster than the generator-enumeration algorithm for the class of solvable groups.  The solvable groups contain the nilpotent groups and hence the nilpotent groups of class $2$ which are ``almost Abelian'' in a precise sense but are considered to be a formidable obstacle to efficient algorithms for general group isomorphism~\cite{babai2011a,codenotti2011a,babai2012a}.  Before our work, it was a longstanding open problem~\cite{lipton2011a} to obtain an $n^{(1 - \eps) \log n + O(1)}$ algorithm where $\eps > 0$ even for $p$-groups.

\begin{theorem}
  \label{thm:sol-group-iso}
  Solvable-group isomorphism is in $n^{(1 / 2) \log_p n + O(\log n / \log \log n)}$ deterministic time where $p$ is the smallest prime dividing the order of the group.
\end{theorem}

In the case of randomized and quantum algorithms, we can replace the $1 / 2$ by $1 / 4$ and $1 / 6$ respectively.  We note that, although not all groups are solvable, solvable groups are quite general.  Every nilpotent group (and hence every $p$-group) is solvable.  More generally, groups of odd order~\cite{feit1963a} or of order $p^a q^b$ where $p$ and $q$ are prime (cf. ~\cite{robinson1996a}) are solvable.  Empirically, the proportion of $2$-groups among the finite groups of order at most $N$ tends to $1$ as $N$ approaches infinity~\cite{besche2002a}; therefore, the density of the solvable groups in the finite groups also tends to $1$.

We start by reducing\inlong{ group isomorphism} to composition-series isomorphism (which we define shortly).

\begin{restatable}{theorem}{groupredcomp}
  \label{thm:group-red-comp}
  Group isomorphism is $n^{(1 / 2) \log_p n + O(1)}$ time deterministic Turing reducible to composition-series isomorphism where $p$ is the smallest prime dividing the order of the group.
\end{restatable}

A \emph{composition series} $S$ is a tower of subgroups $G_0 = 1 \tril \cdots \tril G_m = G$ such that no additional subgroups can be inserted.  The factor groups $G_{i + 1} / G_i$ are called the \emph{composition factors} of $S$.  We say that two composition series $G_0 = 1 \tril \cdots \tril G_m = G$ and $H_0 = 1 \tril \cdots \tril H_m = H$ are isomorphic if there is an isomorphism $\phi : G \ra H$ such that $\phi[G_i] = H_i$ for all $i$.  In the \emph{composition-series isomorphism problem}, we are given composition series $S$ and $S'$ where the subgroups are specified by their corresponding subsets and must decide whether $S \cong S'$.

Wagner~\cite{wagner2011b} showed a Karp reduction from composition-series isomorphism to low-degree graph isomorphism~\cite{luks1982a,babai1983a,babai1983b}.  Using this reduction, he gave an $n^{O(p)}$ algorithm for $p$-group composition series isomorphism using an $n^{O(d)}$ algorithm\inlong{\footnote{There is actually an $n^{O(d / \log d)}$ algorithm~\cite{babai1983b} for isomorphism of graphs of degree at most $d$ which yields an $n^{O(p / \log p)}$ algorithm.}} for isomorphism of graphs of degree at most $d$.  Wagner's reduction also applies to general groups, but produces graphs of large degree when there are large composition factors.  Wagner reduced the degrees by using a modification of the generator-enumeration algorithm to guess the action of the isomorphism on large composition factors; the graphs are then modified to enforce this mapping.  This would have yielded an $n^{O(\log n / \log \log n)}$ algorithm for isomorphism of composition series for arbitrary groups.  Unfortunately, the resulting graphs depend on which coset representatives are chosen so the algorithm is flawed.  Wagner attempted to fix this error\inlong{ in two subsequent revisions}~\cite{wagner2012a,wagner2012b}\inshortorlong{ but ultimately retracted his claim}{; unfortunately, these contain a more subtle manifestation of the same flaw.  Wagner retracted his attempted fix} in the latest revision of his paper~\cite{wagner2012c} following further discussions with the author\inshort{. } (\inshortorlong{S}{s}ee \inshortorlong{the full version of our paper~\cite{rosenbaum2012b} for a detailed discussion}{\appref{flaw}}.)  However, we will show that Wagner's generator-fixing trick can still be salvaged in a special case which we utilize in our algorithm for solvable groups.

Another result of Wagner~\cite{wagner2011b} is a reduction from group isomorphism to composition-series isomorphism.  This reduction is based on guessing generators for the composition series and is no more efficient than the generator-enumeration algorithm.  Wagner combined this with the reduction from $p$-group composition-series isomorphism to graph isomorphism to obtain an algorithm for $p$-group isomorphism that is no faster than the generator-enumeration algorithm.

By instead using our more efficient $n^{(1 / 2) \log_p n + O(1)}$ time reduction from group isomorphism to composition-series isomorphism, we obtain an algorithm for $p$-groups that is faster than the generator-enumeration algorithm when $p$ is small.  When $p$ is large, the generator-enumeration algorithm runs in $n^{\log_p n + O(1)}$ time since every $p$-group has a generating set of size $\log_p n$.  Choosing whichever algorithm is faster based on $p$ and $n$ yields the following result.

\begin{theorem}
  \label{thm:p-group-iso}
  $p$-group isomorphism is in $n^{(1 / 2) \log n + O(1)}$ deterministic time\inlong{\footnote{We remark that we have $\log n$ here rather than $\log_p n$ as in the algorithm for solvable groups; this is the price we pay for reducing the lower-order term to $O(1)$ in the case of $p$-groups.}}.
\end{theorem}

The first technical contribution of this work is our idea of guessing the intermediate subgroups of the composition series directly rather than guessing their generators.  Our improved $n^{(1 / 2) \log_p n + O(1)}$ time reduction from group isomorphism to composition-series isomorphism is a direct consequence of this innovation.  The $n^{(1 / 2) \log_p n + O(1)}$ bound is obtained by building up a composition series for the \emph{socle} of $G$ (denoted $\soc(G)$) whose subgroups are direct products.  The number of choices for the first subgroup in the series is bounded by the order of $G$.  For each subsequent subgroup, the number of choices is reduced by a factor of $p$ where $p$ is the smallest prime that divides the order of $G$.  Taking the product of the number of choices at each step, we obtain a product of the first $\log_p \abs{\soc(G)}$ powers of $p$.  Summing the exponents, the $1 / 2$ then falls out of the formula $\sum_{i = 1}^k i = k (k + 1) / 2$.  We obtain the rest of the composition series by recursively computing a composition series for $G / \soc(G)$ and lifting back to subgroups of $G$.

It is important to note that when the intermediate subgroups in the composition series are determined by guessing generators as in Wagner's algorithm, the number of possible choices for each subsequent subgroup does not reduce by a factor of $p$ at each step.  This yields an upper bound of $n^{\log_p n + O(1)}$; to obtain the factor of $1 / 2$ in the exponent of our result, it is crucial that we consider all possible intermediate subgroups at each step instead of guessing generators.

Our second technical innovation allows us to generalize \thmref{p-group-iso} to the class of solvable groups.  By extending the techniques in the proof of \thmref{group-red-comp} using the theory of Sylow bases, we derive a Turing reduction from solvable-group isomorphism to Hall composition-series isomorphism (a variant of composition-series isomorphism which we define later).

\begin{theorem}
  \label{thm:sol-red-hcomp}
  Solvable-group isomorphism is $n^{(1 / 2) \log_p n + O(1)}$ time deterministic Turing reducible to Hall composition-series isomorphism where $p$ is the smallest prime dividing the order of the group.
\end{theorem}

By generalizing the reduction from composition-series isomorphism to low-degree graph isomorphism, we show that Hall composition-series isomorphism is polynomial-time Karp reducible to low-degree graph isomorphism.  \thmref{sol-group-iso} then follows from \thmref{sol-red-hcomp} and the $n^{O(d / \log d)}$ algorithm~\cite{babai1983b} for testing isomorphism of graphs of degree at most $d$.

We also show a reformulation of group isomorphism as a collision problem which further reduces the constants in the exponents by a factor of $1 / 2$ using randomized algorithms or a factor of $1 / 3$ using quantum algorithms.  To apply this trick to reduce the $1 / 2$ in the exponents of Theorems~\ref{thm:sol-group-iso} -- \ref{thm:sol-red-hcomp} to $1 / 4$ and $1 / 6$ for randomized and quantum algorithms, it is necessary to use the $n^{O(d)}$ algorithm~\cite{babai1983a,babai1983b} for computing canonical forms of graphs of degree at most $d$ instead of using the $n^{O(d / \log d)}$ graph-isomorphism testing algorithm.  We start by fixing an isomorphism $\phi : G \ra H$.  This isomorphism can be used to define a bijection from the set $\cS$ of all composition series for $G$ and the set $\cS'$ of all composition series for $H$ that we consider.  Thus, each composition series $S \in \cS$ for $G$ is paired with some composition series $S' \in \cS'$ for $H$.  We guess $n^{(1 / 4) \log_p n + O(1)}$ composition series uniformly at random from both $\cS$ and $\cS'$; a collision detection argument implies that if $G \cong H$ then with high probability we guessed some pair of composition series $S \in \cS$ and $S' \in \cS'$ which are also isomorphic.  We then determine if an isomorphic pair exists by constructing a graph from each composition series guessed, computing the canonical forms and sorting the resulting graphs.  From this, we obtain $1 / 4$ for randomized algorithms.  The $1 / 6$ for quantum algorithms follows from using a quantum collision detection~\cite{bassard1997b} instead of selecting $S$ and $S'$ uniformly at random.

By modifying the generator-enumeration algorithm to perform group canonization, we can also apply our randomized and quantum speedups to the generator-enumeration algorithm using the same tricks (see \inshortorlong{the full version~\cite{rosenbaum2012b}}{\appref{gen-rq-can}} for the details).  This yields $n^{(1 / 2) \log_p + O(1)}$ and $n^{(1 / 3) \log_p n + O(1)}$ randomized and quantum algorithms for testing isomorphism of arbitrary groups.

The deterministic variants of our algorithms also apply to the \emph{group canonization problem}.  Here, the goal is to compute a multiplication table of a group with the underlying set $[n]$ which is isomorphic to the original group and uniquely identifies its isomorphism class.  We accomplish this by computing the canonical form of the graph for each composition series using the $n^{O(d)}$ algorithm for canonization of graphs of degree at most $d$.  We then select the canonical form that comes first lexicographically.  Since a multiplication table can be extracted from this graph in deterministic polynomial time, we obtain an algorithm for group canonization.  The details are presented in \inshortorlong{the full version~\cite{rosenbaum2012b}}{\appref{group-can}}.

In \secref{group-back}, we provide a review of standard group theory facts that are used in the rest of our paper.  In \secref{graph-red}, we present a formalized variant of Wagner's reduction and prove its correctness.  In \secref{comp-series}, we derive results on lifting composition series from factor groups and bounds on the number of composition series that must be considered.  In \secref{comp-red}, we prove the correctness of our Turing reductions from group isomorphism to composition-series isomorphism.  In \secref{p-algorithms}, we derive our algorithms for $p$-group isomorphism.  In \secref{sylow-bases}, we review the theory of Sylow bases which we utilize in our construction.\inlong{  Before presenting the full algorithm for solvable groups, we provide a sketch in \secref{sol-sketch}.}  In \secref{hcomp-red}, we show our reduction from solvable-group isomorphism to Hall composition-series isomorphism.  In \secref{hgraph-red}, we reduce Hall composition-series isomorphism to graph isomorphism.  In \secref{sol-algorithms}, we leverage our machinery to obtain our algorithms for solvable-group isomorphism.  \inshort{Due to length constraints, we omit some proofs and sketch others.  We direct the reader to the full version~\cite{rosenbaum2012b} for complete proofs of all results.}

\section{Group theory background}
\label{sec:group-back}
In this section, we review some standard results from group theory; we omit the proofs.  \inshortorlong{These}{Some of these are given for convenience in \appref{group-res}; the rest} may be found in group theory and algebra texts~\cite{rotman1995a,robinson1996a,lang2002a,holt2005a,rose2009a,artin2010a,wilson2010a,roman2011a}.  Readers familiar with group theory may wish to skip this section.

Let $G$ and $H$ be groups.  Throughout this paper, we assume all groups are finite and let $n = \abs{G}$.  We let $\iso(G, H)$ denote the set of all isomorphisms from $G$ to $H$.  The \emph{conjugation} of $x$ by $g$ is written as $x^g = g x g^{-1}$; the bijections $\iota_g : G \ra G : x \mapsto x^g$ are the inner automorphisms of $G$.  A \emph{normal} subgroup $N$ of $G$ (denoted $N \trile G$) is a subgroup that is closed under conjugation by elements of $G$.  The \emph{normal closure} of subset $S$ is the smallest normal subgroup of $G$ containing $S$ and is equal to $\langle S^G \rangle$.  The \emph{coset} of an element $g \in G$ for a subgroup $H \leq G$ is the set $g H = \setb{g h}{h \in H}$.  The \emph{canonical map} $\varphi : G \ra G / N : g \mapsto g N$ sends each element of $G$ to its coset.  We say that $H$ is a \emph{characteristic} subgroup of $G$ (denoted $H \ch G$) if every automorphism $\phi \in \aut(G)$ satisfies $\phi[H] = H$.  The \emph{center} of $G$ (denoted $Z(G)$) is the subgroup of the elements that commute with all of $G$.  For a prime $p$, a \emph{$p$-group} is a group of order a power of $p$.\inlong{  This is equivalent to the condition that every element of the group has order a power of $p$.}  A \emph{Sylow $p$-subgroup} of $G$ is a $p$-subgroup of order $p^e$ where $p^e$ is the largest power of $p$ dividing $n$.  A \emph{Sylow} subgroup of $G$ is a Sylow $p$-subgroup for some prime $p$.

A group is called \emph{simple} if it has no nontrivial proper normal subgroups.  A group is \emph{characteristically simple} if it has no proper nontrivial characteristic subgroups.  A \emph{minimal normal subgroup} $N$ of $G$ is a nontrivial normal subgroup of $G$ that does not properly contain any nontrivial normal subgroup of $G$.  The \emph{socle} of a group $G$ is denoted $\soc(G)$ and is the (characteristic) subgroup generated by the minimal normal subgroups of $G$.  A \emph{subnormal series} $S$ of a group $G$ is a tower of subgroups $G_0 = 1 \tril \cdots \tril G_m = G$.  The groups $G_{i + 1} / G_i$ are called the \emph{factors} of $S$.  A \emph{composition series} of a group $G$ is a maximal subnormal series for $G$.  The factors of a composition series are called \emph{composition factors}; by the Jordan-H{\"{o}}lder Theorem (cf.~\cite{robinson1996a,rose2009a}), the multiset of composition factors is uniquely determined by $G$ so they are also called the composition factors of $G$.  A group is \emph{solvable} if it has a subnormal series in which each factor is Abelian.  A \emph{normal series} is a subnormal series in which each $G_i \trile G$.  A \emph{central series} is a normal series where each $G_{i + 1} / G_i \leq Z(G / G_i)$.  A group is \emph{nilpotent} if it has a central series.  Thus, every nilpotent group is solvable.  A group is \emph{$k$-nilpotent} if it has a central series of length $k$.

We shall make use of the following results.

\begin{restatable}[cf.~\cite{rotman1995a,holt2005a,wilson2010a}]{proposition}{mnscharsim}
  Every minimal normal subgroup is characteristically simple.
\end{restatable}

\begin{restatable}[cf.~\cite{rotman1995a,robinson1996a,holt2005a,wilson2010a}]{proposition}{charsimprod}
  A group is characteristically simple \ifft it is the direct product of isomorphic simple groups.
\end{restatable}

\begin{restatable}[cf.~\cite{robinson1996a,holt2005a}]{proposition}{socprod}
  The socle of a group is a direct product of minimal normal subgroups.
\end{restatable}

Together, the last three propositions imply that the socle can be written as a direct product of some subset of its simple minimal normal subgroups.  We shall make use of this fact in our reduction from group isomorphism to composition-series isomorphism.  Note that the direct product decomposition of the socle is \emph{not} unique.  The next proposition allows us to compute this decomposition.

\begin{restatable}[cf.~\cite{seress2003a}]{proposition}{soccomp}
  The socle and minimal normal subgroups can be computed in polynomial time.
\end{restatable}



\newcommand{\firstsylowenv}{
  \begin{restatable}[First Sylow Theorem, cf.~\cite{rotman1995a,robinson1996a,lang2002a,rose2009a,artin2010a,roman2011a}]{theorem}{firstsylow}
    Let $G$ be a group and let $p$ be a prime that divides its order.  There exists a Sylow $p$-subgroup of $G$.
  \end{restatable}}

\newcommand{\secsylowenv}{
  \begin{restatable}[Second Sylow Theorem, cf.~\cite{rotman1995a,robinson1996a,lang2002a,rose2009a,artin2010a,roman2011a}]{theorem}{secsylow}
    All Sylow $p$-subgroups of $G$ are conjugate.
  \end{restatable}}

\newcommand{\compsylowenv}{
  \begin{restatable}[Kantor~\cite{kantor1985a}]{proposition}{compsylow}
    \label{prop:comp-sylow}
    A Sylow $p$-subgroup of a group $G$ can be computed in polynomial time.
  \end{restatable}}

\inlong{
  \firstsylowenv
  \secsylowenv}

\compsylowenv

\section{A Karp reduction from\insodaornot{\\}{ }composition-series isomorphism to\insodaornot{\\}{ }low-degree graph isomorphism}
\label{sec:graph-red}
In this section, we present a variant of Wagner's Karp reduction from composition-series isomorphism to graph isomorphism and supply the correctness proof claimed in the introduction.  The main ideas are essentially the same; however, we fill in several gaps in Wagner's arguments and define the construction precisely so that rigorous proofs are possible.  Finally, we adapt the reduction to perform composition-series canonization instead of isomorphism testing.  First, we define isomorphisms for composition series.

\begin{definition}
  \label{defn:iso-resp}
  Let $S$ be the subnormal series $G_0 = 1 \tril \cdots \tril G_m = G$ and let $S'$ be the subnormal series $H_0 = 1 \tril \cdots \tril H_m = H$.  An isomorphism $\phi : G \ra H$ respects $S$ and $S'$ if $\phi(G_i) = H_i$ for all $i$.  We say that an isomorphism $\phi : G \ra H$ which respects $S$ and $S'$ is an isomorphism from $S$ to $S'$.
\end{definition}

We shall make use of the following result.

\begin{theorem}[Babai, Kantor and Luks~\cite{babai1983b}]
  \label{thm:const-deg-iso}
  Isomorphism of colored graphs of degree at most $d$ is in $n^{O(d / \log d)}$ deterministic time.
\end{theorem}

A closely related problem is \emph{graph canonization}.  Here, we are given a graph and must compute its canonical form.  This is a graph which is isomorphic to the original graph and has the property that two graphs are isomorphic \ifft their canonical forms are equal.

In a \emph{colored graph}, each vertex has an associated color.  Let $X_1$ and $X_2$ be colored graphs.  An isomorphism $\phi : X_1 \ra X_2$ must also respect the colors so that a node $x \in X_1$ has color $k$ \ifft $\phi(x) \in X_2$ has color $k$.  We will use the following algorithm for colored graph canonization.

\begin{theorem}[Babai, Kantor and Luks~\cite{babai1983a,babai1983b}]
  \label{thm:const-deg-can}
  Canonization of colored graphs of degree at most $d$ is in $n^{O(d)}$ deterministic time.
\end{theorem}

We now construct a tree by starting with the whole group $G$ and decomposing it into its cosets $G / G_{m - 1}$; we then further decompose each coset in $G / G_{m - 1}$ into the cosets $G / G_{m - 2}$ that it contains.  This process is repeated until we reach the trivial group $G_0 = 1$.  We make this precise with the following definition.

\begin{definition}
  \label{defn:TS}
  Let $G$ be a group and consider the composition series $S$ given by the subgroups $G_0 = 1 \tril \cdots \tril G_m = G$.  Then $T(S)$ is defined to be the tree whose nodes are $\bigcup_i G / G_i$.  The root node is $G = G_m$ which we identify with $\{G\}$.  The leaf nodes are $\{x\} \in G / 1$ which we identify with $x \in G$.  For each node $x G_{i + 1} \in G / G_{i + 1}$, there is an edge to each $y G_i$ such that $y G_i \subseteq x G_{i + 1}$.
\end{definition}

We now use this tree to define a graph that encodes the multiplication table of $G$.  The idea is to attach a multiplication gadget to the nodes $x, y, z \in G$ for each entry $x y = z$ in the multiplication table.  Naively, each node $x \in G$ would have degree $\Omega(n)$.  We avoid this problem by creating a new tree which contains a copy of the tree $T(S)$ for each element of $G$; the root nodes of these trees are then identified with the leaves of another copy of $T(S)$.  We then show how to construct multiplication gadgets so that each of the $n^2$ leaf nodes is involved in only a constant number of multiplication gadgets.  This causes the resulting graph to have degree $p + O(1)$.  The details of this construction are described in the following definition.

\begin{definition}
  \label{defn:X-S}
  Let $G$ be a group and consider the composition series $S$ given by the subgroups $G_0 = 1 \tril \cdots \tril G_m = G$.  To construct $X(S)$, we start with a copy $T_S$ of $T(S)$.  For each $x \in G$, we create an additional copy $T_S^{(x)}$ of $T(S)$.  We then combine these graphs by identifying the root of each $T_S^{(x)}$ with the leaf $x$ of $T_S$.  The node for $y G_i$ in $T_S^{(x)}$ is denoted $(y G_i)^{(x)}$; in particular, the leaf for $y \in G$ in $T_S^{(x)}$ is denoted $y^{(x)}$.  For each $x, y \in G$, we connect three new leaves $y_\la^{(x)}$, $y_\ra^{(x)}$ and $y_=^{(x)}$ to the leaf $y^{(x)}$ in $T_S^{(x)}$ and add them to $X(S)$.  For all $x, y, z \in G$ such that $x y = z$, we draw an edge from $y_\la^{(x)}$ to $x_\ra^{(y)}$ and from $x_\ra^{(y)}$ to $y_=^{(z)}$ and color $y_\la^{(x)}$ ``left'', $x_\ra^{(y)}$ ``right'' and $y_=^{(z)}$ ``equal''.  We color the root node of $T_S$ ``root''; the remaining nodes are colored ``internal''.
\end{definition}

The graph $X(S)$ is a \emph{cone graph}; that is, a rooted tree with edges between nodes at the same level.  We call the edges that form the tree in a cone graph \emph{tree edges} and the edges between nodes at the same level \emph{cross edges}.  We now show a bijection from $\iso(S, S')$ to $\iso(X(S), X(S'))$.  Our main contribution to the construction described in this section is a rigorous proof of this result.

\begin{theorem}
  \label{thm:p-graph-red-bij}
  Let $S$ and $S'$ be composition series for the groups $G$ and $H$.  Then there is a bijection between $\iso(S, S')$ and $\iso(X(S), X(S'))$.
\end{theorem}

We provide a sketch and defer the complete proof to \inshortorlong{the full version~\cite{rosenbaum2012b}}{\appref{p-group-iso-proofs}}.  Consider an isomorphism $\phi$ from $S$ to $S'$.  We define $\hat \phi$ to be the map from $X(S)$ to $X(S')$ that maps the root $G$ of $X(S)$ to the root $H$ of $X(S')$, each node $x \in G$ to $\phi(x)$ and each $y^{(x)}$ to $(\phi(y))^{(\phi(x))}$; similarly, we define $\hat \phi(y_\lambda^{(x)}) = (\phi(y))_\lambda^{(\phi(x))}$ for each $x, y \in G$ and $\lambda \in \{\la, \ra, =\}$.  We define $\hat \phi$ on the intermediate coset nodes by $\hat \phi(x G_i) = \phi[x G_i]$ for each $x G_i \in G / G_i$ and $\hat \phi((y G_i)^{(x)}) = (\phi[y G_i])^{(\phi(x))}$ for each $x \in G$ and $y G_i \in G / G_i$.  The first half of the proof involves showing that $\hat \phi$ is an isomorphism from $S$ to $S'$.  It is straightforward to show that $\hat \phi$ is a well-defined bijection which respects the colors and tree edges.  The fact that $\hat \phi$ respects the cross edges (which encode the multiplication gadgets) follows easily from the assumption that $\phi$ is an isomorphism.  At this point we know that $f : \iso(S, S') \ra \iso(X(S), X(S')) : \phi \mapsto \hat \phi$ is a well-defined function.

The second (more difficult) half of the proof is to show that $f$ is a bijection.  It is simple to show that $f$ is injective since if $f(\phi_1) = \hat \phi_1 = \hat \phi_2 = f(\phi_2)$ for $\phi_1, \phi_2 \in \iso(S, S')$ then we have $\phi_1 = \restr{\hat \phi_1}{G} = \restr{\hat \phi_2}{G} = \phi_2$.  Showing that it is surjective is more involved.  Let $\theta : X(S) \ra X(S')$ be an isomorphism.  Because of our coloring of the nodes, we know that $\theta$ maps the root $G$ of $X(S)$ to the root $H$ of $X(S')$ which implies that $\theta$ maps the \nth{k} level of $X(S)$ to the \nth{k} level of $X(S')$.  Then $\phi = \restr{\theta}{G} : G \ra H$ is a bijection.  We now argue that $\theta(y^{(x)}) = (\phi(y))^{(\phi(x))}$ for all $x, y \in G$.  This follows from the observation that the unique shortest path from $x$ to $y$ in $X(S)$ that passes through a node colored ``left'' and then one colored ``right'' goes through the node $y^{(x)}$ and later $x^{(y)}$.  To show that $\phi$ is an isomorphism from $G$ to $H$, we note that for all $x, y \in G$ with $z = xy$, $\theta$ maps the multiplication gadget $(y^{(x)}, x^{(y)}, y^{(z)})$ to $((\phi(y))^{(\phi(x))}, (\phi(x))^{(\phi(y))}, (\phi(y))^{(\phi(z))})$ which implies that $\phi(xy) = \phi(x) \phi(y)$.  The fact that $\phi[G_i] = H_i$ for each $i$ follows by considering the action of $\theta$ on the descendants of the node $G_i$ which are contained in $G$.

The correctness of our reduction follows.

\begin{corollary}
  \label{cor:p-red-cor}
  Let $S$ and $S'$ be composition series.  Then $S \cong S'$ \ifft $X(S) \cong X(S')$.
\end{corollary}

In order to obtain an efficient algorithm for $p$-group composition-series isomorphism, we must show that the degree of the graph is not too large.

\begin{lemma}
  \label{lem:p-graph}
  Let $G$ be a group with a composition series $S$ such that $\alpha$ is an upper bound for the order of any factor.  Then the graph $X(S)$ has degree at most $\max\{\alpha + 1, 4\}$ and size $O(n^2)$.
\end{lemma}

\begin{proof}
  The tree $T(S)$ has size $O(n)$ and $X(S)$ contains $n + 1$ copies of this tree.  Connecting three additional nodes $y_{\ell}^{(x)}$ to each leaf $y^{(x)}$ in each $T_S^{(x)}$ only increases the number of nodes by a constant factor.  Thus, $X(S)$ has $O(n^2)$ nodes.  Let $x, y \in G$.  The degree of the leaf node $x^{(y)}$ is $4$ since it is connected to its three children $x_\lambda^{(x)}$ where $\lambda \in \{\la, \ra, =\}$ and its parent.  By construction, the degree of any internal node is at most $\alpha + 1$.
\end{proof}

We are now in a position to obtain an algorithm for composition-series isomorphism.

\begin{theorem}
  \label{thm:alpha-comp-iso}
  Let $G$ and $H$ be groups with composition series $S$ and $S'$ such that $\alpha$ is an upper bound for the order of any factor.  Then we can test if $S \cong S'$ in $n^{O(\alpha / \log \alpha)}$ deterministic time.
\end{theorem}

\begin{proof}
  We can compute the graphs $X(S)$ and $X(S')$ in polynomial time.  By \corref{p-red-cor}, $S \cong S'$ \ifft $X(S) \cong X(S')$ so this reduction is correct.  By \lemref{p-graph}, the number of nodes in $X(S)$ is $O(n^2)$ and the degree is at most $\max\{\alpha + 1, 4\} = O(\alpha)$.  Then we test if $X(S) \cong X(S')$ in $n^{O(\alpha / \log \alpha)}$ time using \thmref{const-deg-iso}.
\end{proof}

For the randomized and quantum variants of our algorithm, it is necessary to adapt \thmref{alpha-comp-iso} to composition-series canonization.  First, we define complete polynomial-size invariants.

\begin{definition}
  \label{defn:comp-inv}
  Let $S$ be a composition series for a group $G$.  A complete polynomial-size invariant $\invar(S)$ for $S$ has the following properties:

  \begin{enumerate}
  \item If $S$ and $S'$ are composition series for\inlong{ the groups} $G$ and $H$ then $S \cong S'$ \ifft $\invar(S) = \invar(S')$.
  \item $\abs{\invar(S)} = \poly(n)$.
  \end{enumerate}
\end{definition}

\begin{definition}
  \label{defn:comp-can}
  Let $S$ be a composition series for a group $G$.  We define the canonical form $\can(S)$ to be a tuple $(M, \psi(G_0), \ldots, \psi(G_m))$ with following properties:

  \begin{enumerate}
  \item $M$ is an $n \times n$ matrix with entries in $[n]$.
  \item $M$ is the multiplication table for a group which is isomorphic to $G$ under $\psi : G \ra [n]$.
  \item $\can(S)$ is a complete polynomial-size invariant.
  \end{enumerate}
\end{definition}

\inlong{In the \emph{composition-series canonization problem}, we are given a composition series $S$ for a group $G$ and must compute $\can(S)$.  }We now give an algorithm\inshortorlong{ for composition-series canonization}{ for this problem}.

\begin{restatable}{theorem}{compcan}
  \label{thm:comp-can}
  Let $S$ be a composition series for a group $G$ such that $\alpha$ is an upper bound for the order of any factor.  Then we can compute a canonical form $\can(S)$ for $S$ in $n^{O(\alpha)}$ deterministic time.
\end{restatable}

We give a sketch and defer the complete proof to \inshortorlong{the full version~\cite{rosenbaum2012b}}{\appref{p-group-iso-proofs}}.  Although there are many details that must be checked, the idea is simple.  It is easy to see that $\can(X(S))$ is a complete polynomial-size invariant for $S$.  Moreover, because of the multiplication gadgets in the graph $X(S)$, we can deterministically compute a multiplication table $M(G)$ for a group with the underlying set $[n]$ that is isomorphic to $G$.  Let $\theta$ be an isomorphism from $X(S)$ to $\can(X(S))$.  We can find the node $\theta(1)$ in $\can(X(S))$ using this multiplication table.  Each node $\theta(G_i)$ is on the path from the root $\theta(G)$ to $\theta(1)$ so we can find these nodes as well.  The descendants of each $\theta(G_i)$ that are contained in $\theta[G]$ can be used to compute the subgroup of the group described by $M(G)$ which corresponds to $G_i$.  Because this is a deterministic computation which depends only on $\can(X(S))$, $n$ and the composition length of $G$, we obtain a canonical form $\can(S)$.

\section{Composition series: lifting\insodaornot{\\}{ }isomorphisms and bounds}
\label{sec:comp-series}
In this section, we prove some useful lemmas that are used to derive our main result.\inlong{  First, we review some standard results which are proved for convenience in~\appref{comp-series-proofs}.}\inshortorlong{  Let $G$ and $H$ be groups and consider an isomorphism $\phi : G \ra H$ where $N \trile G$ and $N' \trile H$.  Then the induced isomorphism $\restr{\phi}{G / N} : G / N \ra H / N' : g N \mapsto \phi[g N]$ is well-defined.}{

  \begin{definition}
    \label{defn:induced-iso}
    Let $G$ and $H$ be groups with normal subgroups $N$ and $N'$.  Let $\phi : G \ra H$ be an isomorphism such that $\phi[N] = N'$.  Then we say that $\phi$ induces the isomorphism $\restr{\phi}{G / N} : G / N \ra H / N' : g N \mapsto \phi[g N]$.
  \end{definition}

\begin{restatable}{proposition}{indisowell}
  The induced isomorphism is well-defined and is indeed an isomorphism.
\end{restatable}}\inshortorlong{  }{

}The following proposition is useful in proving the correctness of our reduction.

\begin{restatable}[cf.~\cite{roman2011a}]{proposition}{respseries}
  \label{prop:resp-series}
  Let $G$ and $H$ be groups with normal subgroups $N$ and $N'$.  \inshortorlong{Let $T$ and $T'$ be subnormal series for $G / N$ and $H / N'$ and let $S$ and $S'$ be the series obtained from $T$ and $T'$ by taking the preimage of each subgroup under the canonical maps $\varphi : G \ra G / N$ and $\varphi : H \ra H / N'$.  Then an isomorphism $\phi : G \ra H$ respects $S$ and $S'$ \ifft $\indphi{G / N}$ respects $T$ and $T'$.}{Let $T$ be a subnormal series $K_0 = \{N\} \tril \cdots \tril K_m = G / N$ for $G / N$ and let $T'$ be a subnormal series $L_0 = \{N'\} \tril \cdots \tril L_m = H / N'$ for $H / N'$.  Denote the canonical maps by $\varphi : G \ra G / N$ and $\varphi' : H \ra H / N'$.  Let $S$ be the subnormal series $1 \tril G_0 \tril \cdots \tril G_m = G$ and let $S'$ be the subnormal series $1 \tril H_0 \tril \cdots \tril H_m = H$ where $G_i = \varphi^{-1}[K_i]$ and $H_i = \varphi'^{-1}[L_i]$.  Consider an isomorphism $\phi : G \ra H$ where $\phi[N] = N'$.  Then $\phi$ respects $S$ and $S'$ \ifft the induced isomorphism $\indphi{G / N} : G / N \ra H / N'$ respects $T$ and $T'$.}
\end{restatable}

We now prove a series of lemmas that bound the number of composition series we consider.

\begin{lemma}
  \label{lem:mns-bound}
  If $G$ is a group such that $p$ is the smallest prime that divides its order $n$, then $G$ has at most $(n - 1) / (p - 1)$ minimal normal subgroups.
\end{lemma}

\begin{proof}
  The intersection of two normal subgroups of $G$ is itself a normal subgroup of $G$ so it follows that all minimal normal subgroups intersect trivially.  By Lagrange's theorem, the order of every subgroup divides the order of the group.  Then each minimal normal subgroup of $G$ contains at least $p - 1$ non-identity elements that are not contained in any other minimal normal subgroup of $G$.  \inshortorlong{I}{Then i}f there are $\ell$ minimal normal subgroups, $\ell (p - 1) + 1 \leq n$ so $\ell \leq (n - 1) / (p - 1)$.
\end{proof}

\begin{lemma}
  \label{lem:smns-factor}
  Let $G$ be a group and consider $N \trile G$.  If $S$ is a simple minimal normal subgroup of $G$ such that $N \cap S = 1$, then $N \times S / N$ is a simple minimal normal subgroup of $G / N$.
\end{lemma}

\begin{proof}
  Let $S$ be a simple minimal normal subgroup of $G$ that intersects $N$ trivially.  Then $N \times S \trile G$ so $N \times S / N \trile G / N$.  By the second isomorphism theorem, $N \times S / N \cong S$ so $N \times S / N$ is simple.  Suppose that $H / N \leq N \times S / N$ where $H / N \trile G / N$.  Since $N \times S / N$ is simple, $H / N \in \{N / N, N \times S / N\}$ so $H \in \{N, N \times S\}$.  Thus, $N \times S / N$ is a simple minimal normal subgroup\inlong{ of $G / N$}.
\end{proof}

\begin{corollary}
  \label{cor:smns-dist}
  Let $G$ be a group and consider $N \trile G$.  If $S_1$ and $S_2$ are distinct simple minimal normal subgroups of $G$ such that $N \cap S_i = 1$ and $N \times S_1 \not= N \times S_2$, then $N \times S_1 / N$ and $N \times S_2 / N$ are distinct simple minimal normal subgroups of $G / N$.
\end{corollary}

The next lemma is the source of the $1 / 2$ in the exponent of the runtime of our main result.

\begin{lemma}
  \label{lem:comp-prod-bound}
  Let $G$ be a group and let $p$ be the smallest prime that divides $G$.  There are at most $n p^{(1 / 2) \log_p^2 n} = n^{(1 / 2) \log_p n + 1}$ composition series\inlong{ for $G$ of the form} $1 \tril S_1 \tril \cdots \tril \bigtimes_{i = 1}^k S_i = G$ where each $S_i$ is a simple minimal normal subgroup of $G$.
\end{lemma}

\begin{proof}
  First, we note that if $G$ cannot be written as a direct product of its simple minimal normal subgroups then no such composition series exist so the bound holds trivially.  Otherwise, there are at most $n - 1$ choices for $S_1$ by \lemref{mns-bound}.  Once we choose $S_1$, it follows from \lemref{mns-bound} and \corref{smns-dist} that there are at most $n / \abs{S_1} - 1 \leq n / p$ distinct direct products $S_1 S_2$ such that $S_2$ is a simple minimal normal subgroup of $G$ that intersects $S_1$ trivially.  More generally, if $\bigtimes_{i = 1}^j S_i$ is a direct product where each $S_i$ is a simple minimal normal subgroup of $G$ then there are at most $n / \prod_{i = 1}^j \abs{S_i} - 1 \leq n / p^j$ distinct direct products $\left(\bigtimes_{i = 1}^j S_i\right) \times S_{j + 1}$ such that $S_{j + 1}$ is a simple minimal normal subgroup of $G$ that intersects $\bigtimes_{i = 1}^j S_i$ trivially.  Let $t = \lfloor \log_p n \rfloor$.  Then the number of composition series consisting of direct products of simple minimal normal subgroups is at most \inshortorlong{$\prod_{i = 0}^t (n / p^i) \leq n p^{t (\log_p n - t / 2 - 1 / 2)} \leq n p^{(1 / 2) \log_p^2 n}$.\qedhere}{

  \begin{align*}
    \prod_{i = 0}^t (n / p^i) & \leq n \frac{n^t}{\prod_{i = 1}^t p^i} \\
                           {} &    = n p^{t (\log_p n - t / 2 - 1 / 2)} \\
                           {} & \leq n p^{\log_p n (\log_p n - (\log_p n - 1) / 2 - 1 / 2)} \\
                           {} &    = n p^{(1 / 2) \log_p^2 n} \qedhere
  \end{align*}}

\end{proof}

Note that we did not use the $p - 1$ in the denominator of \lemref{mns-bound} so the above analysis can be tightened; however, this only reduces the number of possibilities by a polynomial factor.

\section{A Turing reduction from group\insodaornot{\\}{ }isomorphism to composition-series\insodaornot{\\}{ }isomorphism}
\label{sec:comp-red}
\newcommand{\groupcompalg}{
  \begin{algorithm}[H]
    \begin{algorithmic}[1]
      \Input A group $G$ specified by its multiplication table
      \Output A composition series $S$ for $G$ with the intermediate subgroups that are socles or preimages of socles labelled ``socle''
      \Function{\comps}{$G$}
        \State $S \coleq (1, G)$ \inlong{\Comment{The composition series for $G$ will be stored here}}
        \State $\{N_1, \ldots, N_k\} \coleq \Call{\mns}{G}$
        \State $\soc(G) \coleq \genb{N_i}{1 \leq i \leq k}$
        \If{$G$ is not simple}
          \If{$\soc(G) \not= G$}
            \State $(K_0 \tril \cdots \tril K_m) \coleq \Call{\comps}{G / \soc(G)}$ \label{line:recur}
            \For{$i = 1, \ldots, m - 1$}
              \State Insert $G_i \coleq \varphi^{-1}[K_i]$ into $S$
              \inshortorlong{\State C}{ and c}opy the label of $K_i$ to $G_i$
            \EndFor
          \EndIf
          \State $T \coleq \bigcup_{i = 1}^k \Call{\sims}{N_i}$ \label{line:comp-T}
          \State $K \coleq 1$
          \While{$K \not= \soc(G)$} \label{line:prod-loop}
            \State Choose $L \in T$ \label{line:choose-L}
            \State $K \coleq K \times L$
            \State Insert $K$ into $S$
            \State $T' = \emptyset$ \label{line:set-T'}
            \For{$L \in T$} \label{line:loop-reps}
              \If{$K \cap L = 1 \text{ and } K \times L \not= K \times L'$ \inshort{\Statex \hspace{78.5pt} \nonumber}for all $L' \in T'$}
                \State Insert $L$ into $T'$
              \EndIf
            \EndFor \label{line:loop-reps-end}
            \State $T \coleq T'$ \label{line:set-T}
          \EndWhile
          \State Label $\soc(G)$ as ``socle'' in $S$
        \EndIf
        \State \Return $S$
      \EndFunction
    \end{algorithmic}
    \caption{An algorithm for computing a composition series; $\varphi : G \ra G / \soc(G)$ denotes the canonical map}
    \label{alg:group-comp}
  \end{algorithm}}

\newcommand{\simsubalg}{
  \begin{algorithm}[H]
    \begin{algorithmic}[1]
      \Input A characteristically simple group $G$ specified by its multiplication table
      \Output The set of all simple minimal normal subgroups of $G$
      \Function{\sims}{$G$}
        \State $\{N_1, \ldots, N_k\} \coleq \Call{\mns}{G}$
        \State \Return $\setb{N_i}{N_i \text{ is simple}}$
      \EndFunction
    \end{algorithmic}
    \caption{An algorithm for computing all simple minimal normal subgroups}
    \label{alg:sim-sub}
  \end{algorithm}}

We now present our method for constructing a composition series and prove \thmref{group-red-comp}.  Guessing intermediate subgroups directly instead of their generators is the crucial difference between our method and previous work by Wagner~\cite{wagner2011b} that enables us to obtain the $n^{(1 / 2) \log_p n + O(1)}$ time Turing reduction from group isomorphism to composition-series isomorphism ($p$ is the smallest prime dividing the order of the group).  If we used generators for the intermediate subgroups as in Wagner's paper~\cite{wagner2011b}, we would obtain an $n^{\log_p n + O(1)}$ time reduction.  We start with a coarse-grained recursive group-decomposition process due to Luks~\cite{luks1986a}.

Consider the groups $G$ and $H$.  We start by taking the socles of $G$ and $H$ to obtain the subnormal series $1 \trile \soc(G) \trile G$ and $1 \trile \soc(H) \trile H$.  We continue this process by recursively obtaining a decomposition of $G / \soc(G)$ and $H / \soc(H)$ and take the preimages of these series under the canonical maps to obtain decompositions for $G$ and $H$.  Any isomorphism from $G$ to $H$ must respect $\soc(G)$ and $\soc(H)$ and the inverse images obtained in the recursive decomposition process.  We then decompose the socles into direct products of simple subgroups.  This yields composition series for $\soc(G)$ and $\soc(H)$.  The only choices we have are the towers of direct products for the socles.

Our algorithm for finding a composition series of $G$ works by first computing $\soc(G)$; we then construct a composition series for $\soc(G)$ where the intermediate subgroups are direct products of simple minimal normal subgroups of $\soc(G)$.  To obtain the second half of the composition series for $G$, we recursively construct a composition series for $G / \soc(G)$; by taking the preimage of this series under the canonical map $\varphi : G \ra G / \soc(G)$, we obtain a tower of subgroups of $G$ that contain $\soc(G)$ such that every subgroup is normal in the next subgroup and the factor groups are simple.  This process is described more precisely in \algref{group-comp}; $\Call{\mns}{}$ is an algorithm that takes a group and returns its minimal normal subgroups.

\notinsoda{\groupcompalg \simsubalg}

The composition series that is obtained from this process depends on which composition series consisting of direct products of simple minimal normal subgroups is chosen for each socle.  We refer to this as the \emph{choice of socle decomposition}.  This corresponds to the choice of simple minimal normal subgroups of the socles on \linref{choose-L} in \algref{group-comp}.  Each such choice results in a different composition series.
\insoda{
  \groupcompalg \simsubalg
}
It is straightforward to prove the correctness of \algref{group-comp}; \inshortorlong{see the full version~\cite{rosenbaum2012b}}{we defer the proof to \appref{p-group-iso-proofs}}.

\begin{restatable}{lemma}{compcor}
  \label{lem:comp-cor}
  \algref{group-comp} returns a composition series for $G$; moreover, every subgroup in the series which is a socle or a preimage of a socle is labelled ``socle.''
\end{restatable}

\begin{lemma}
  \label{lem:comp-eff}
  Algorithms~\ref{alg:group-comp} and \ref{alg:sim-sub} run in deterministic time polynomial in $n$.
\end{lemma}

\begin{proof}
  \algref{sim-sub} runs in polynomial time since we can compute all minimal normal subgroups by taking the normal closure of all nontrivial elements and eliminating those that are contained in some other normal closure.  We can test if a group is simple by checking that the normal closure of every nontrivial element is the group.  Then \algref{group-comp} runs in polynomial time.
\end{proof}

\begin{lemma}
  \label{lem:all-comp-choice}
  Let $G$ and $H$ be groups and let $S$ be a composition series obtained by running \algref{group-comp} on $G$.  Fix an isomorphism $\phi : G \ra H$; then for the choice of socle decomposition ($*$), running \algref{group-comp} on $H$ returns a composition series $S'$ such that $\phi$ is an isomorphism from $S$ to $S'$.  The choice ($*$) of socle decomposition for $H$ corresponds to the images of the simple minimal normal subgroups in the socle decomposition chosen for $G$.
\end{lemma}

\begin{proof}
  Let $\phi : G \ra H$ be an isomorphism.  Let $S'$ be the image of $S$ under $\phi$.  We claim that $S'$ can be obtained by running \algref{group-comp} for some choice of socle decomposition.  The proof is by induction on the number of subgroups labelled ``socle'' in \algref{group-comp}.  For the base case, $G$ is simple.  This is trivial since the only composition series for $G$ and $H$ are $1 \tril G$ and $1 \tril H$.

  For the inductive case, let $T$ and $T'$ be the simple minimal normal subgroups of $\soc(G)$ and $\soc(H)$ which are computed on \linref{comp-T}.  Then $\phi$ acts as a bijection from $T$ to $T'$.  Let $\soc(G) = \bigtimes_{i = 1}^j L_i$ be the direct product that was computed for $\soc(G)$.  Then we have $\soc(H) = \bigtimes_{i = 1}^j \phi[L_i]$ and each $\phi[L_i]$ is a simple minimal normal subgroup of $\soc(H)$ in $T'$.  Thus, $\phi$ respects the portions of $S$ and $S'$ up to the socles for some choice of socle decompositions.  If $\soc(G) = G$ then we are done.  Otherwise, by induction, the induced isomorphism $\indphi{G / \soc(G)} : G / \soc(G) \ra H / \soc(H)$ respects the composition series for $G / \soc(G)$ and $H / \soc(H)$ computed recursively on \linref{recur} for some choice of socle decompositions.  It follows from \propref{resp-series} that $\phi$ respects $S$ and $S'$ for some choice of socle decompositions.
\end{proof}

We now consider all possible composition series $S'$ for $H$ that can be constructed according to \algref{group-comp}.  First, we need a simple lemma which \inshortorlong{can be proved easily}{we prove in \appref{p-group-iso-proofs}}.

\begin{restatable}{lemma}{maxtwonorm}
  \label{lem:max-2norm}
  The maximum value of $\norm{x}_2^2$ \inshortorlong{where}{subject to the constraints that} $x \in \bbR^d$, each $x_i \geq 1$\inshortorlong{,}{ and} $\sum_{i = 1}^d x_i = t$ \inshortorlong{and}{where} $t \geq d$ is $(d - 1) + (t - (d - 1))^2$.
\end{restatable}

\begin{lemma}
  \label{lem:num-comp-choices}
  Let $G$ be a group and let $p$ be the smallest prime that divides its order.  Then the number of choices of socle decompositions in \algref{group-comp} is at most $n^{5 / 2} p^{(1 / 2) (\log_p n - \ell)^2 + 1} \leq n^{(1 / 2) \log_p n + O(1)}$ where $\ell$ is the number of subgroups labelled ``socle.''
\end{lemma}

\begin{proof}
  Define $F_0 = G$, $F_{i + 1} = F_i / \soc(F_i)$ and let $\varphi_i : F_i \ra F_{i + 1}$ be the canonical map.  We note that $\soc(F_i)$ is a nontrivial subgroup of $F_i$ for $i < \ell$.  Note that $\soc(F_{\ell - 1})$ is the last subgroup labelled ``socle.''

  The subgroups labelled ``socle'' in a composition series $S$ computed by \algref{group-comp} are the preimages $\varphi_0^{-1} \cdots \varphi_{i - 1}^{-1}[\soc(F_i)]$ for $0 \leq i < \ell$.  Let $n_{i + 1} = \abs{\soc(F_i)}$; since $\soc(F_{\ell - 1}) \leq F_{\ell - 1}$, we know that $\prod_{i = 1}^\ell n_i \leq n$ and $n_i \geq p$.  The number of choices of socle decompositions for the \nth{i} level recursive call is bounded by the number of composition series for $\soc(F_i)$ in which the subgroups are direct products of simple minimal normal subgroups of $\soc(F_i)$.  By \lemref{comp-prod-bound}, this is at most $n_i p^{(1 / 2) \log_p^2 n_i}$.  Then the total number of choices is at most $n \prod_{i = 1}^\ell p^{(1 / 2) \log_p^2 n_i} = n p^{(1 / 2) \sum_{i = 1}^\ell \log_p^2 n_i}$.  To obtain a bound, we maximize this quantity subject to the constraints that $\prod_{i = 1}^\ell n_i \leq n$ and $n_i \geq p$.  We know that $\prod_{i = 1}^\ell n_i = n$ at the maximum.  Then setting $x_i = \log_p n_i$, this is equivalent to maximizing $\sum_{i = 1}^\ell x_i^2 = \norm{x}_2^2$ such that $\sum_{i = 1}^\ell x_i = \log_p n$ and $x_i \geq 1$.  By \lemref{max-2norm}, the maximum value is $\norm{x}_2^2 = (\ell - 1) + (\log_p n - (\ell - 1))^2$.  Then the total number of choices is at most \inshortorlong{$n p^{(1 / 2) ((\ell - 1) + (\log_p n - (\ell - 1))^2)} \leq n^{3 / 2} p^{(1 / 2) (\log_p n - (\ell - 1))^2} \leq n^{5 / 2} p^{(1 / 2) ((\log_p n - \ell)^2 + 1)}$}{

  \begin{align*}
    n p^{(1 / 2) ((\ell - 1) + (\log_p n - (\ell - 1))^2)} & \leq n^{3 / 2} p^{(1 / 2) (\log_p n - (\ell - 1))^2} \\
                                                        {} & =    n^{3 / 2} p^{(1 / 2) (\log_p^2 n - 2 (\log_p n)(\ell - 1) + \ell^2 - 2 \ell + 1)} \\
                                                        {} &    = n^{3 / 2} p^{(1 / 2) (\log_p^2 n - 2 (\log_p n)\ell + 2 \log_p n + \ell^2 - 2 \ell + 1)} \\
                                                        {} & \leq n^{5 / 2} p^{(1 / 2) ((\log_p n - \ell)^2 + 1)} \qedhere
  \end{align*}}

\end{proof}

\groupredcomp*

\begin{proof}
  Suppose we have an algorithm for composition-series isomorphism.  We start by fixing a composition series $S$ for $G$ by calling \algref{group-comp} on $G$ using arbitrary choices for the socle decompositions.  Then we compute the composition series $S'$ for $H$ for all possible choices of socle decompositions.  By \lemref{all-comp-choice}, $S$ is isomorphic to one of these composition series $S'$ \ifft $G \cong H$ from which correctness is immediate.  From \lemref{num-comp-choices}, computing all of the composition series $S'$ takes $n^{(1 / 2) \log_p n + O(1)}$ time where $p$ is the smallest prime dividing $n$.
\end{proof}

We now show that this reduction can be performed more efficiently for randomized and quantum algorithms.  Note that our results for randomized and quantum algorithms differ slightly from the deterministic case since it is now necessary to reduce to the slightly more powerful composition-series canonization problem.  As we shall see, using the algorithm for graph canonization instead of the algorithm for testing isomorphism does not have a detrimental effect on the final runtime bounds.  We remark that both of the following results also hold when composition-series canonization is replaced by the slightly weaker problem of computing a complete polynomial-size invariant.

\begin{restatable}{theorem}{groupredrandcomp}
  \label{thm:group-red-rand-comp}
  Group isomorphism is $n^{(1 / 4) \log_p n + O(1)}$ time randomized Turing reducible to composition-series canonization where $p$ is the smallest prime dividing $n$.
\end{restatable}

\begin{proof}
  Let $G$ and $H$ be groups and consider the sets $\cA$ and $\cB$ of all composition series that can be computed for $G$ and $H$ using \algref{group-comp} for some choice of socle decompositions.  Fix an isomorphism $\phi : G \ra H$.  For each $S \in \cA$, the image $S'$ of $S$ under $\phi$ is in $\cB$ by \lemref{all-comp-choice}.  Conversely, the preimage $S$ under $\phi$ of any $S' \in \cB$ is in $\cA$ by \lemref{all-comp-choice}.  Therefore, $\phi$ induces a bijection from $\cA$ to $\cB$.

  Consider the randomized algorithm that randomly selects subsets $A \subseteq \cA$ and $B \subseteq \cB$ which are each of size $n^{(1 / 4) \log_p n + O(1)}$.  If $G \cong H$, then with high probability there exist $S \in A$ and $S' \in B$ such that $S \cong S'$.  To check if there is such a pair, we compute the canonical forms $\can(S)$ and $\can(S')$ for each $S \in A$ and $S' \in B$; we then sort the canonical forms which result from $A$ and $B$ and store the results in the vectors $\bmA$ and $\bmB$.  This takes time at most $n^{(1 / 4) \log_p n + O(1)}$.  By merging $\bmA$ and $\bmB$, we can check if such a pair exists in $n^{(1 / 4) \log_p n + O(1)}$ time.  If we find $S \in A$ and $S' \in B$ where $S \cong S'$, we output $G \cong H$; otherwise, we output $G \not\cong H$.  Then we can decide if $G \cong H$ with one-sided error in $n^{(1 / 4) \log_p n + O(1)}$ time.
\end{proof}

\begin{restatable}{theorem}{groupredqcomp}
  \label{thm:group-red-q-comp}
  Group isomorphism is $n^{(1 / 6) \log_p n + O(1)}$ time quantum Turing reducible to composition-series canonization where $p$ is the smallest prime dividing $n$.
\end{restatable}

\begin{proof}
  Let $G$ and $H$ be groups.  We apply the same argument as in the proof of \thmref{group-red-rand-comp} except that we use quantum claw detection~\cite{bassard1997b} instead of guessing $A \subseteq \cA$ and $B \subseteq \cB$.
\end{proof}

We remark that the methods of Theorems~\ref{thm:group-red-rand-comp} and \ref{thm:group-red-q-comp} can also be applied to the generator-enumeration algorithm; this yields an $n^{(1 / 2) \log_p n + O(1)}$ randomized algorithm and an $n^{(1 / 3) \log_p n + O(1)}$ quantum algorithm for general group isomorphism (where $p$ is the smallest prime dividing the order of the group).  See \inshortorlong{the full version~\cite{rosenbaum2012b}}{\appref{gen-rq-can}} for the details.

\section{Algorithms for $p$-group isomorphism}
\label{sec:p-algorithms}
The intermediate results of Sections~\ref{sec:graph-red} and \ref{sec:comp-red} put us in a position to prove \thmref{p-group-iso}.  First, we note that the generator-enumeration algorithm is efficient for $p$-groups when $p$ is large.

\begin{lemma}
  \label{lem:p-gen-enum}
  Let $G$ and $H$ be $p$-groups.  Then we can test if $G \cong H$ in $n^{\log_p n + O(1)}$ time.
\end{lemma}

\begin{proof}
  Observe that every $p$-group has a generating set of size at most $\log_p n$.
\end{proof}

The proofs of the next two lemmas are easy\inlong{ and we defer them to \appref{p-group-iso-proofs}}.

\begin{restatable}{lemma}{smallpbound}
  \label{lem:small-p-bound}
  If $c_1$ and $c_2$ are positive constants then $c_1 \log_p n + c_2 p \leq c_1 \log n + O(1)$ for $2 \leq p \leq c_1 \ln n / c_2 \ln^2 p$.
\end{restatable}

\begin{restatable}{lemma}{largepbound}
  \label{lem:large-p-bound}
  If $c$ is a positive constant then $\log_p n = O(\log n / \log \log n)$ for $p \geq c \ln n / \ln^2 p$.
\end{restatable}

We now prove an extension of \thmref{p-group-iso}.

\begin{theorem}
  \label{thm:p-group-iso-drq}
  Let $G$ and $H$ be $p$-groups.  Then we can test if $G \cong H$ in 
  
  \begin{enumerate}
  \item $n^{(1 / 2) \log n + O(1)}$\inshort{ deterministic} time\inlong{ for deterministic algorithms}.
  \item $n^{(1 / 4) \log n + O(1)}$\inshort{ randomized} time\inlong{ for randomized algorithms}.
  \item $n^{(1 / 6) \log n + O(1)}$\inshort{ quantum} time\inlong{ for quantum algorithms}.
  \end{enumerate}
\end{theorem}

\begin{proof}
  Suppose we have an algorithm $M$ which can decide $p$-group isomorphism in $n^{c_1 \log_p n + c_2 p}$ time where $0 < c_1 < 1$ and $c_2 > 0$ are constants.  By \lemref{small-p-bound}, we know that $n^{c_1 \log_p n + c_2 p} = n^{c_1 \log n + O(1)}$ for $2 \leq p \leq c_1 \ln n / c_2 \ln^2 p$.  \lemref{large-p-bound} implies that $n^{\log_p n + O(1)} = n^{O(\log n / \log \log n)}$ for $p \geq c_1 \ln n / c_2 \ln^2 p$.  We can define an algorithm which runs $M$ when $2 \leq p \leq c_1 \ln n / c_2 \ln^2 p$ and runs the generator-enumeration algorithm when $p > c_1 \ln n / c_2 \ln^2 p$.  The overall complexity is\inlong{ then} $n^{c_1 \log n + O(1)}$.  Applying Theorems~\ref{thm:group-red-comp} and \ref{thm:comp-can} together with the above argument proves part (a).  Parts (b) and (c) follow from using Theorems~\ref{thm:group-red-rand-comp} and \ref{thm:group-red-q-comp} instead of \thmref{group-red-comp}.
\end{proof}

We remark that in the above proof we could have used the $n^{(1 / 2) \log_p n + p / \log p}$ algorithm which follows from Theorems~\ref{thm:group-red-comp} and \ref{thm:alpha-comp-iso}.  In fact, one can generalize Lemmas~\ref{lem:small-p-bound} and \ref{lem:large-p-bound} so that \thmref{p-group-iso-drq} can also be proved when our algorithms for isomorphism testing and graph of graphs of degree at most $d$ run in $n^{O(d \cdot \polylog(d))}$ time.  Regardless of the exact expression for the $\polylog(d)$ term, we still obtain the same runtime.  Thus, $\polylog(d)$ factors in the exponents of graph isomorphism testing and canonization algorithms are unimportant for our purposes.


The deterministic version of our algorithm can be adapted to perform $p$-group canonization; see \inshortorlong{the full version~\cite{rosenbaum2012b}}{\appref{group-can}} for the details.

\section{Background on Sylow bases}
\label{sec:sylow-bases}
In this section, we cover basic properties of Sylow bases which are used later in the paper.

\begin{definition}
  Let $G$ be a group whose order has the prime factorization $n = \prod_{i = 1}^\ell p_i^{e_i}$.  A Sylow basis is a set $\setb{P_i}{1 \leq i \leq \ell}$ where each $P_i$ is a Sylow $p_i$-subgroup of $G$ and $P_i P_j = P_j P_i$ for all $i$ and $j$.
\end{definition}

 When we speak of a Sylow basis $\setb{P_i}{1 \leq i \leq \ell}$, we always assume that each $P_i$ is a Sylow $p_i$-subgroup of $G$.  We say that the Sylow bases $\setb{P_i}{1 \leq i \leq \ell}$ of $G$ and $\setb{Q_i}{1 \leq i \leq \ell}$ of $H$ are isomorphic if there exists an isomorphism $\phi : G \ra H$ such that $\phi[P_i] = Q_i$ for all $i$.  As we shall see, this is the case \ifft $G$ is isomorphic to $H$.

\begin{proposition}
  If $\setb{P_i}{1 \leq i \leq \ell}$ is a Sylow basis of a solvable group $G$ whose order has the prime factorization $n = \prod_{i = 1}^\ell p_i^{e_i}$, then for any $\cI \subseteq [\ell]$, $\prod_{i \in \cI} P_i$ is a subgroup of $G$ of order $\prod_{i \in \cI} p_i^{e_i}$.
\end{proposition}

\begin{proof}
  First, we remark that since $\setb{P_i}{1 \leq i \leq \ell}$ is a Sylow basis, there is no ambiguity in which order the product is taken.  The result is clearly true when $\abs{\cI} = 1$.  Suppose that the result holds for $\abs{\cI} \leq k$ and let $\abs{\cI} = k + 1$.  Choose some $j \in \cI$ and let $P = \prod_{i \in \cI \setminus \{j\}} P_i$.  By assumption, $P$ is a subgroup of $G$ of order $\prod_{i \in \cI \setminus \{j\}} p_i^{e_i}$.  Let $x_1, y_1 \in P$ and $x_2, y_2 \in P_j$ and suppose that $x_1 x_2 = y_1 y_2$.  Then $x_2 y_2^{-1} = x_1^{-1} y_1$ which implies that $x_i = y_i$ since the orders of $P$ and $P_j$ are coprime.  Thus, $\abs{P P_j} = \prod_{i \in \cI} p_i^{e_i}$.  Because $P P_j = P_j P$, it is easy to show that $P P_j$ is a subgroup of $G$.
\end{proof}

This immediately implies the following.

\begin{proposition}
  \label{prop:syl-basis-uniq}
  If $\setb{P_i}{1 \leq i \leq \ell}$ is a Sylow basis of a solvable group $G$, then every $x \in G$ is uniquely expressible as a product $x_1 \cdots x_\ell$ where each $x_i \in P_i$.
\end{proposition}

The following two theorems of Hall are crucial in our utilization of Sylow bases.

\begin{theorem}[Hall~\cite{hall1938a}, cf.~\cite{robinson1996a}]
  \label{thm:sol-syl}
  A group $G$ is solvable \ifft it has a Sylow basis.
\end{theorem}

Two Sylow bases $\cP = \setb{P_i}{1 \leq i \leq \ell}$ and $\cQ = \setb{Q_i}{1 \leq i \leq \ell}$ of $G$ are conjugate if there exists $g \in G$ such that for all $i$, $P_i^g = Q_i$.  We denote by $\cP^g$ the Sylow basis $\setb{P_i^g}{1 \leq i \leq \ell}$.

\begin{theorem}[Hall~\cite{hall1938a}, cf.~\cite{robinson1996a}]
  \label{thm:sol-conj}
  Any two Sylow bases of a solvable group are conjugate.
\end{theorem}

\begin{corollary}
  \label{cor:syl-iso}
  Let $G$ and $H$ be isomorphic solvable groups.  If $\setb{P_i}{1 \leq i \leq \ell}$ and $\setb{Q_i}{1 \leq i \leq \ell}$ are Sylow bases for $G$ and $H$ then there is an isomorphism $\phi : G \ra H$ such that $\phi[P_i] = Q_i$ for each $i$.
\end{corollary}

\begin{proof}
  Let $\psi : G \ra H$ be an isomorphism.  Clearly, $\setb{\psi[P_i]}{1 \leq i \leq \ell}$ is a Sylow basis for $H$.  Then by \thmref{sol-conj}, there exists $h \in H$ such that inner automorphism $\iota_h$ maps each $\psi[P_i]$ to $Q_i$.  Therefore, $\phi = \iota_h \psi : G \ra H$ is an isomorphism from $G$ to $H$ such that $\phi[P_i] = Q_i$ for each $i$.
\end{proof}

\inlong{
  \section{A sketch of the algorithm for solvable groups}
  \label{sec:sol-sketch}
  In this section, we provide a sketch of the generalization to solvable groups.  The full construction is shown in Sections~\ref{sec:hcomp-red} -- \ref{sec:sol-algorithms}.  Consider a group $G$ whose order has the prime factorization $n = \prod_{i = 1}^\ell p_i^{e_i}$ and let $\cP = \setb{P_i}{1 \leq i \leq \ell}$ be a \emph{Sylow basis} of $G$.  Since a group is solvable \ifft it has a Sylow basis, the Second Sylow Theorem implies that the number of possible choices of Sylow $p_i$-subgroups for all $i$ is at most $n^\ell$.  A standard bound on the number of prime divisors of $n$ then implies that we can find a Sylow basis $\cP$ of a solvable group in $n^{(1 + o(1)) (\ln n / \ln \ln n)}$ time.

Because all Sylow bases of a group are conjugate, it follows that a group has at most $n$ Sylow bases.  Thus, if $G$ and $H$ are isomorphic solvable groups, we can assume that we have Sylow bases $\cP = \setb{P_i}{1 \leq i \leq \ell}$ and $\cQ = \setb{Q_i}{1 \leq i \leq \ell}$ for $G$ and $H$ which are mapped to each other under some isomorphism.  We then construct a composition series $S_i$ for each $P_i$.  For an isomorphism $\phi : G \ra H$, there exists some choice of socle decompositions of each composition series $S_i'$ for $Q_i$ such that $\phi$ is an isomorphism from $S_i$ to $S_i'$.  We call $(G, \cP, \cS)$ and $(H, \cQ, \cS')$ Hall composition series where $\cS = \setb{S_i}{1 \leq i \leq \ell}$ and $\cS' = \setb{S_i'}{1 \leq i \leq \ell}$.  Our reduction to Hall composition-series isomorphism follows from considering all possible choices of $\cS'$.

To reduce Hall composition-series isomorphism to low-degree graph isomorphism, we reindex the $P_i$ and $Q_i$ so that their prime divisors are in decreasing order.  We fix generators for the $P_1, \ldots, P_\kappa$ with large prime divisors and guess the corresponding generators for $Q_1, \ldots, Q_\kappa$.  These generators induce unique orders on the elements of the subgroups $P = P_1 \cdots P_\kappa$ and $Q = Q_1 \cdots Q_\kappa$ of $G$ and $H$.  This allows us to construct binary trees whose leaves represent the elements of $P$ and $Q$; we fix the leaves by coloring them according to their positions in the orderings induced by the generators.  For each $x \in P$, we construct a copy of the tree $T(S_{\kappa + 1})$ as for $p$-groups and attach its root to the leaf $x$ of the binary tree for $P$.  The leaves of the resulting tree then represent elements of $P P_{\kappa + 1}$; we continue in this manner by attaching a copy of the tree $T(S_{i + 1})$ to every leaf of the current tree for each $i > \kappa + 1$.

Applying the same ideas as before, we use $n + 1$ copies of the resulting tree together with multiplication gadgets to create a cone graph which represents a Hall composition-series of $G$.  We then construct the same graph for $H$ for all possible choices of Hall composition series and generators.  The graph for $G$ is isomorphic to one of the graphs for $H$ \ifft $G \cong H$.
}

\section{A Turing reduction from\insodaornot{\\}{ }solvable-group isomorphism to Hall\insodaornot{\\}{ }composition-series isomorphism}
\label{sec:hcomp-red}
In this section, we describe our Turing reduction from solvable-group isomorphism to Hall composition-series isomorphism and thereby prove \thmref{sol-red-hcomp}.  To start, we define Hall composition-series isomorphism.

\begin{definition}
  \label{defn:hcomp-iso-prob}
  In the Hall composition-series isomorphism problem, we are given two groups $G$ and $H$ specified by their multiplication tables whose orders have the prime factorization $n = \prod_{i = 1}^\ell p_i^{e_i}$; we are also given Sylow bases $\cP = \setb{P_i}{1 \leq i \leq \ell}$ and $\cQ = \setb{Q_i}{1 \leq i \leq \ell}$ for $G$ and $H$ and composition series $S_i$ and $S_i'$ for each $P_i$ and $Q_i$.  We denote the sets of composition series by $\cS = \setb{S_i}{1 \leq i \leq \ell}$ and $\cS' = \setb{S_i'}{1 \leq i \leq \ell}$.  The subgroups of each composition series $S_i$ and $S_i'$ are $P_{i0} = 1 \tril \cdots \tril P_{i, m_i} = P_i$ and $Q_{i0} = 1 \tril \cdots \tril Q_{i, m_i} = Q_i$.  The Sylow bases and composition series are represented by the corresponding subsets of $G$ and $H$.  We say that $\cS$ and $\cS'$ are isomorphic if there exists an isomorphism $\phi : G \ra H$ such that for all $i$, $\phi[P_i] = Q_i$ and each pair of composition series $S_i$ and $S_i'$ are isomorphic under the restriction of $\phi$ to $P_i$.  The Hall composition-series isomorphism problem is to determine if $\cS$ and $\cS'$ are isomorphic.
\end{definition}

For the sake of brevity, we call $(G, \cP, \cS)$ a Hall composition series.  Here, $G$, $\cP$ and $\cS$ are as in \defref{hcomp-iso-prob}.  To start with, we show an algorithm for finding a Sylow basis of a solvable group $G$.  For this it is convenient to use the following standard result from number theory. 

\begin{proposition}
  \label{prop:prime-div}
  The number of distinct prime divisors of $n$ is at most $(1 + o(1)) \frac{\ln n}{\ln \ln n}$.
\end{proposition}

\begin{lemma}
  \label{lem:comp-syl-basis}
  Let $G$ be a solvable group whose order has the prime factorization $n = \prod_{i = 1}^\ell p_i^{e_i}$.  Then in $n^{(1 + o(1)) (\ln n / \ln \ln n)}$ deterministic time, we can compute a Sylow basis of $G$.
\end{lemma}

\begin{proof}
  \thmref{sol-syl} implies that $G$ has a Sylow basis.  By \propref{comp-sylow}, we can find a Sylow $p_i$-subgroup $P_i$ of $G$ for each $i$ in polynomial time.  By the Second Sylow Theorem, for each $i$ there exists $g_i \in G$ such that $\cP = \setb{P_i^{g_i}}{1 \leq i \leq \ell}$ is a Sylow basis of $G$.  For any choices of $g_i \in G$, we can test if $\setb{P_i^{g_i}}{1 \leq i \leq \ell}$ is a Sylow basis of $G$ in polynomial time by checking if the relations $P_i^{g_i} P_j^{g_j} = P_j^{g_j} P_i^{g_i}$ hold for all $i$ and $j$.  We simply enumerate all possible $g_1, \ldots, g_\ell \in G$ to find a Sylow basis.  Since the order $p_i^{e_i}$ of each $P_i$ divides $n$, the number of Sylow subgroups of $G$ for distinct primes is at most the number of distinct prime divisors of $n$.  The bound on the runtime is immediate from \propref{prime-div}.
\end{proof}

\lemref{comp-syl-basis} suffices for the purposes of this paper.  However, finding a Sylow basis for a solvable group can be done in deterministic polynomial time.

\begin{restatable}{lemma}{compsylbasispoly}
  \label{lem:comp-syl-basis-poly}
    Let $G$ be a solvable group whose order has the prime factorization $n = \prod_{i = 1}^\ell p_i^{e_i}$.  Then we can compute a Sylow basis of $G$ deterministically in polynomial time.
\end{restatable}

The proof is based on other results by Hall.  (See \inshortorlong{the full version~\cite{rosenbaum2012b}}{\appref{sol-group-iso-proofs}} for the details.)  This can also be derived from a result of Kantor and Taylor~\cite{kantor1988a}.  We now bound the number of Hall composition series in analogy to \lemref{all-comp-choice}. 

\begin{lemma}
  \label{lem:all-comp-choice-syl}
  Let $G$ and $H$ be groups whose orders have the prime factorization $n = \prod_{i = 1}^\ell p_i^{e_i}$ such that $\cP = \setb{P_i}{1 \leq i \leq \ell}$ and $\cQ = \setb{Q_i}{1 \leq i \leq \ell}$ are Sylow bases for $G$ and $H$.  Fix an isomorphism $\phi : G \ra H$ which is also an isomorphism from $\cP$ to $\cQ$.  For each $i$, let $S_i$ be a composition series for $P_i$ obtained by running \algref{group-comp} on $G$.  Then for some choice of socle decompositions, running \algref{group-comp} on each $Q_i$ returns a composition series $S_i'$ such that $\phi$ is an isomorphism from $S_i$ to $S_i'$.
\end{lemma}

\begin{proof}
  Let $\phi_i = \restr{\phi}{P_i} : P_i \ra Q_i$ for each $i$ and apply \lemref{all-comp-choice} to each $S_i$ and $\phi_i$.
\end{proof}

An extension\inlong{ of the bound} of \lemref{num-comp-choices} is also required.

\begin{lemma}
  \label{lem:num-hcomp-choices}
  Let $G$ be a group and let $p$ be the smallest prime that divides its order.  Then there are at most $n^{(1 / 2) \log_p n + O(1)}$ Hall composition series $(G, \cP, \cS)$ which can be obtained by some choice of Sylow basis $\cP = \setb{P_i}{1 \leq i \leq \ell}$ and some $\cS = \setb{S_i}{1 \leq i \leq \ell}$ where each $S_i$ is a composition series for $P_i \in \cP$ which results from some choice of socle decompositions in \algref{group-comp}.
\end{lemma}

\begin{proof}
  By \thmref{sol-conj}, there are at most $n$ choices of a Sylow basis $\cP$ for $G$.  Suppose the order of $G$ has the prime factorization $n = \prod_{i = 1}^\ell p_i^{e_i}$.  By \lemref{num-comp-choices}, the number of choices for each composition series $S_i$ for $P_i$ which can be obtained by some choice of socle decompositions in \algref{group-comp} is at most $p_i^{(1 / 2) e_i^2 + O(e_i)}$.  Then the total number of choices of socle decompositions for all $S_i$ is at most \inshortorlong{$\prod_i p_i^{(1 / 2) e_i^2 + O(e_i)} \leq n^{(1 / 2) \log_p n + O(1)}$.}{

  \begin{align*}
    \prod_i p_i^{(1 / 2) e_i^2 + O(e_i)} &    = \prod_i \left(p_i^{O(e_i)} p_i^{(1 / 2) e_i^2}\right) \\
                                      {} & \leq \left(\prod_i p_i^{e_i} \right)^{O(1)} \left(\prod_i p_i^{e_i}\right)^{(1 / 2) \max \setb{e_i}{1 \leq i \leq \ell}} \\
                                      {} & \leq n^{O(1)} n^{(1 / 2) \log_p n} \\
                                      {} &    = n^{(1 / 2) \log_p n + O(1)} \qedhere
  \end{align*}
  The result follows.}
\end{proof}

Because it is needed for the randomized and quantum variants of our algorithm, we now introduce\inlong{ the} Hall composition-series canonization\inlong{ problem}.

\begin{definition}
  \label{defn:hcomp-inv}
  Let $(G, \cP, \cS)$ be a Hall composition series.  A complete polynomial-size invariant $\invar(\cS)$ for $\cS$ has the following properties:

  \begin{enumerate}
  \item If $(G, \cP, \cS)$ and $(H, \cQ, \cS')$ are Hall composition series then $\cS \cong \cS'$ \ifft $\invar(\cS) = \invar(\cS')$.
  \item $\abs{\invar(\cS)} = \poly(n)$.
  \end{enumerate}
\end{definition}

\begin{definition}
  \label{defn:hcomp-can}
  Let $(G, \cP, \cS)$ be a Hall composition series and suppose that the order of $G$ has the prime factorization $n = \prod_{i = 1}^\ell p_i^{e_i}$.  We reindex so that $p_i > p_j$ for $i < j$ and define the canonical form $\can(\cS)$ to be a tuple \inshortorlong{$(M, \psi[P_1], \ldots, \psi[P_\ell], \Psi_0, \ldots, \Psi_\ell)$ where each $\Psi_i = (\psi[P_{i 0}], \ldots, \psi[P_{i, m_i}])$}{$(M, \psi[P_1], \ldots, \psi[P_\ell], \psi[P_{10}], \ldots, \psi[P_{1, m_1}], \ldots, \psi[P_{\ell 0}], \ldots, \psi[P_{\ell, m_\ell}])$} such that:

  \begin{enumerate}
  \item $M$ is an $n \times n$ matrix with entries in $[n]$.
  \item $M$ is the multiplication table for a group which is isomorphic to $G$ under $\psi : G \ra [n]$.
  \item $\can(\cS)$ is a complete polynomial-size invariant.
  \end{enumerate}
\end{definition}

In the \emph{Hall composition-series canonization problem}, we are given a Hall composition series $(G, \cP, \cS)$ and must compute $\can(\cS)$.  A stronger and more general version of \thmref{sol-red-hcomp} now follows from \lemref{comp-syl-basis-poly}, the technique of \thmref{group-red-comp} and \lemref{all-comp-choice-syl}.  As for our reductions to composition-series in \secref{comp-red}, the result also holds when Hall composition-series canonization is replaced by the problem of computing a complete polynomial-size invariant.

\begin{theorem}
  \label{thm:sol-red-hcomp-drq}
  Solvable-group isomorphism is Turing reducible to Hall composition-series canonization.  The reduction requires

  \begin{enumerate}
  \item $n^{(1 / 2) \log_p n + O(1)}$\inshort{ deterministic} time\inlong{ for deterministic algorithms}
  \item $n^{(1 / 4) \log_p n + O(1)}$\inshort{ randomized} time\inlong{ for randomized algorithms}
  \item $n^{(1 / 6) \log_p n + O(1)}$\inshort{ quantum} time\inlong{ for quantum algorithms}
  \end{enumerate}
  where $p$ is the smallest prime dividing the order of the group.  In the case of deterministic algorithms, we can Turing reduce solvable-group isomorphism to Hall composition-series isomorphism instead of\inlong{ Hall composition-series} canonization.
\end{theorem}

\begin{proof}
  Suppose that we have an algorithm for Hall composition-series canonization and let $G$ and $H$ be solvable groups.  Using \lemref{comp-syl-basis-poly}, we compute Sylow bases $\cP = \setb{P_i}{1 \leq i \leq \ell}$ and $\cQ = \setb{Q_i}{1 \leq i \leq \ell}$ for $G$ and $H$ in $\poly(n)$ time.  By \corref{syl-iso}, $\cP$ and $\cQ$ are isomorphic.

  We proceed in a manner similar to \thmref{group-red-comp}.  We fix a composition series $S_i$ for each $P_i$ by calling \algref{group-comp} on each $P_i$ using arbitrary choices for the socle decompositions.  Then we compute the composition series $S_i'$ for each $Q_i$ for all possible choices of socle decompositions.  Let $\cS = \setb{S_i}{1 \leq i \leq \ell}$ and $\cS' = \setb{S_i'}{1 \leq i \leq \ell}$; by \lemref{all-comp-choice-syl}, $\cS$ is isomorphic to some choice of Hall composition series $\cS'$ \ifft $G \cong H$.  Correctness follows immediately.  The total time required to enumerate all choices of Hall composition series is at most $n^{(1 / 2) \log_p n + O(1)}$ by \lemref{num-hcomp-choices}.

  To prove part (a), we simply note that two calls to any algorithm for Hall composition-series canonization suffice to decide Hall composition-series isomorphism.  Parts (b) and (c) follow from the collision detection arguments of the proofs of Theorems~\ref{thm:group-red-rand-comp} and \ref{thm:group-red-q-comp}.
\end{proof}

\section{A Karp reduction from Hall\insodaornot{\\}{ }composition-series isomorphism to\insodaornot{\\}{ }low-degree graph isomorphism}
\label{sec:hgraph-red}
In this section, we generalize our technology for $p$-groups and show that Hall composition-series isomorphism is polynomial-time Karp reducible to graph isomorphism as claimed in the introduction.  In the case of $p$-groups, all of the nodes of our tree $T(S)$ had order bounded by $p + O(1)$; this is because every composition factor of a $p$-group has order $p$.  However, for general solvable groups, the tree constructed in the obvious way does not have this property since the primes $p_i$ which divide $n$ could have widely varying orders.  It is therefore necessary to reduce the degree of the parts of the tree which correspond to large primes $p_i$.  We accomplish this by using Wagner's idea of fixing the generators of the large composition factors.  Although his trick does not work in general, it does work in our case since by construction, the degrees of the nodes in our tree decrease as we move away from the root.

We start by constructing a fixed binary tree that will correspond to the top of our new tree.  This portion of the new tree will handle all large composition factors with the degree of the tree bounded by a constant.  First, we require a useful fact which allows us to order the elements of a group uniquely given a generating set.  The proof is simple\inshortorlong{; see the full version~\cite{rosenbaum2012b}}{and we defer it to \appref{sol-group-iso-proofs}}.

\begin{restatable}[Wagner~\cite{wagner2011b}]{proposition}{genord}
  \label{prop:gen-ord}
  Let $G$ be a group which is generated by the elements $\bmg = (g_1, \ldots, g_k)$.  We can define a total order $\preceq_{\bmg}$ on $G$ that depends only on $\bmg$ where the first $k + 1$ elements are $e \prec_{\bmg} g_1 \prec_{\bmg} \cdots \prec_{\bmg} g_k$.  Given $x, y \in G$, we can decide if $x \preceq_{\bmg} y$ in polynomial time.  Moreover, if $H$ is a group which is generated by the elements $\bmh = (h_1, \ldots, h_k)$ and $\phi : G \ra H$ is an isomorphism such that $\phi(g_i) = h_i$ for all $i$ then for every $x, y \in G$, $x \preceq_{\bmg} y$ \ifft $\phi(x) \preceq_{\bmh} \phi(y)$.
\end{restatable}

We use the following binary tree in the degree reduction step.

\begin{definition}
  \label{defn:right-bin-tree}
  Consider a vector $\bmx = (x_1, \ldots, x_k)$ of $k$ nodes.  We let $B([k])$ be any binary tree of size $O(k)$ with $k$ leaves that are each labeled by a unique element of $[k]$; we require that the distance from the root to a leaf is the same for all leaves.  We define $B(\bmx)$ by replacing the node labelled $i$ in $B([k])$ with $x_i$.
\end{definition}

The reason for first constructing $B([k])$ is so that $B(\bmx)$ cannot depend on the values of the labels $x_i$ of the nodes $\bmx$.  This ensures that if $\bmx = (x_1, \ldots, x_k)$ and $\bmy = (y_1, \ldots, y_k)$ then the partial function $\psi : B(\bmx) \ra B(\bmy) : x_i \mapsto y_i$ for all $i$ extends to a unique isomorphism from $B(\bmx)$ to $B(\bmy)$.  This fact will be used later to ensure that isomorphisms between the graphs constructed for our groups $G$ and $H$ correspond to mapping each $x_i$ to $y_i$.

We now construct our tree which fixes the generators of the large composition factors.  We remark that since every composition factor is simple and each $P_i$ is a $p_i$-group, we only need one generator for each composition factor\inlong{\footnote{This is because the only simple $p$-group is $\bbZ_p$.}}.  We start by reindexing so that $p_i > p_j$ for $i < j$.  We assume that $p_1 \geq \alpha$ since otherwise our previous techniques suffice.  Let $\kappa$ be the largest value of $i$ such that $p_i \geq \alpha$.  We construct a fixed binary tree $T$ using \propref{gen-ord} and \defref{right-bin-tree} which represents the elements of $P_1 \cdots P_\kappa$.  To each leaf of this tree, we attach a copy $T(S_{\kappa + 1})$ as described in \defref{TS}.  For each $i > \kappa$, we attach a copy of $T(S_{i + 1})$ to each leaf of the current tree and continue this process until $i = \ell$.  The leaves of the resulting tree then represent the elements of $G$.  Moreover, the degree is at most $\alpha$.  Since we shall later take $\alpha = \log n$, we see that the factor of $n^{O(\log n / \log \log n)}$ in the runtime of our algorithm for solvable groups comes from guessing the generators of large composition factors and also \thmref{const-deg-iso}.  The resulting tree is shown in \figref{TS-sol-gen}.  This construction is made precise in the following definition.

\begin{definition}
  \label{defn:TS-sol-gen}
  Let $(G, \cP, \cS)$ be a Hall composition series where the order of $G$ has the prime factorization $n = \prod_{i = 1}^\ell p_i^{e_i}$ and $p_i > p_j$ for $i < j$ and let $\alpha$ be a parameter where $\alpha \geq p_1$.  For each $i$ such that $p_i \geq \alpha$, we are given a representative $g_{i, j + 1}$ whose coset $g_{i, j + 1} P_{ij} \in P_{i, j + 1} / P_{ij}$ generates the composition factor $P_{i,j + 1} / P_{ij}$.  Let $\kappa$ be the largest value of $i$ such that $p_i \geq \alpha$.  We denote by $\bmg(\alpha) = \left(g_{11}, \ldots, g_{1, m_1}, \ldots, g_{\kappa 1}, \ldots, g_{\kappa, m_\kappa}\right)$ the vector of representatives.  We note that the subgroup $P_1 \cdots P_\kappa$ is generated by the elements of $\bmg(\alpha)$.  To construct the tree $T(\cS, \bmg(\alpha))$, we start by ordering the elements of $P_1 \cdots P_\kappa$ according to $\preceq_{\bmg(\alpha)}$.  Let $\bmu$ be the vector which corresponds to this ordering.  We construct a copy $T$ of the tree $B(\bmu)$ and denote the leaf which corresponds to $x = x_1 \cdots x_\kappa$ where each $x_j \in P_j$ by $(x_\kappa)_{x_1, \ldots, x_{\kappa - 1}}$; we color the leaf $(x_\kappa)_{x_1, \ldots, x_{\kappa - 1}}$ by the color $k$ where $x$ is the \nth{k} element of the vector $\bmu$.  The root is denoted $P_1$.  For $x = x_1 \cdots x_\kappa$ where each $x_j \in P_j$, we create a copy $T_{x_1, \ldots, x_\kappa}$ of $T(S_{\kappa + 1})$ according to \defref{TS} and identify its root with the node $(x_\kappa)_{x_1, \ldots, x_{\kappa - 1}}$.  For $i > \kappa$ where each $x_j \in P_j$, we create a copy $T_{x_1, \ldots, x_i}$ of the tree $T(S_{i + 1})$ and identify its root with the node for $x_i$ in the tree $T_{x_1, \ldots, x_{i - 1}}$.  The result is the tree $T(\cS, \bmg(\alpha))$.  For each $i > \kappa$ and $x_k \in P_k$, we denote the node for $x_i P_{ij}$ in the tree $T_{x_1, \ldots, x_{i - 1}}$ by $(x_i P_{ij})_{x_1, \ldots, x_{i - 1}}$.  For a leaf $(x_\ell)_{x_1, \ldots, x_{\ell - 1}}$ in $T_{x_1, \ldots, x_{\ell - 1}}$, we use the shorthand $x_1 \cdots x_\ell$.  For subtle technical reasons, we color each node $(x_i P_{ij})_{x_1, \ldots x_{i - 1}}$ ``not identity'' where each $x_k \in P_k$ and $\abs{\setb{x_k \not= 1}{1 \leq k \leq i - 1}} + [x_i \not\in P_{ij}] = 1$.
\end{definition}

\newcommand{\TSsolgenfig}{
  \begin{figure}[h]
    \centering
    \inshortorlong{\includegraphics[scale=0.8]{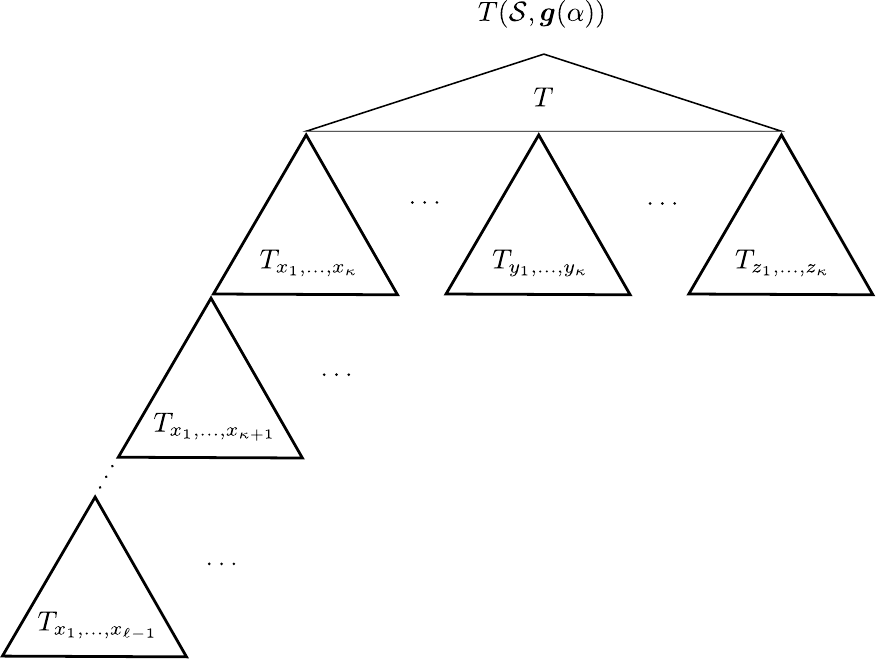}}{\includegraphics{TS-sol-gen.pdf}}
    \inshortorlong{\caption{$T(\cS, \bmg(\alpha))$ where each $x_i, y_i, z_i \in P_i$}}{\caption{The tree $T(\cS, \bmg(\alpha))$ where each $x_i, y_i, z_i \in P_i$}}
    \label{fig:TS-sol-gen}
  \end{figure}}

\TSsolgenfig

Intuitively, if we write $x = x_1 \cdots x_i$ where each $x_k \in P_k$, the nodes $(x_i P_{ij})_{x_1, \ldots x_{i - 1}}$ which are colored ``not identity'' are those which are off the path from $P_1$ to $1$ by following exactly one arbitrary path through a copy of some $T(S_{i'})$ and only passing through nodes of the form $P_{kj}$ in trees $T_{x_1, \ldots, x_{k - 1}}$ for all $k \not= i'$.  The nodes colored ``not identity'' are used to show that the isomorphism between $G$ and $H$ which corresponds to an isomorphism of between the graphs is also an isomorphism from $\cS$ to $\cS'$.

We shall call $(G, \cP, \cS, \bmg(\alpha), \kappa, \alpha)$ a Hall composition series with fixed generators where $G$, $\cP$, $\cS$, $\bmg(\alpha)$, $\kappa$ and $\alpha$ are as in \defref{TS-sol-gen}; when we speak of a Hall composition series with fixed generators, we shall always assume that some prime dividing $n$ is at least $\alpha$.

We now run into the same problem as in the case of $p$-groups; namely, if we were to simply attach multiplication gadgets to each leaf of $T(\cS, \bmg(\alpha))$, the degrees of the leaves would become large.  We remedy this in the same way as before by attaching a copy of $T(\cS, \bmg(\alpha))$ to each leaf of $T(\cS, \bmg(\alpha))$.  We show the resulting graph in \figref{TS-sol-gen-cone} and define it as follows.

\begin{definition}
  \label{defn:TS-sol-gen-cone}
  Let $(G, \cP, \cS, \bmg(\alpha), \kappa, \alpha)$ be a Hall composition series with fixed generators.  To construct $X(\cS, \bmg(\alpha))$, we start with a copy $T_\cS$ of $T(\cS, \bmg(\alpha))$.  For each $x \in G$, we create an additional copy $T_\cS^{(x)}$ of $T(\cS, \bmg(\alpha))$.  We then combine these graphs by identifying the root of each $T_\cS^{(x)}$ with the leaf $x$ of $T_\cS$.  A node in $X(\cS, \bmg(\alpha))$ is denoted by $(y_i P_{ij})^{(x)}_{y_1, \ldots, y_{i - 1}}$ where $i > \kappa$, $x \in G$ and each $y_k \in P_k$.  The superscript $(x)$ determines the tree $T_\cS^{(x)}$ which contains the node.  The subscript $y_1, \ldots, y_{i - 1}$ indicates that the node belongs to the tree $T_{y_1, \ldots, y_{i - 1}}$ within the tree $T_\cS^{(x)}$.  Finally, $y_i P_{ij}$ indicates that our node is the node for $y_i P_{ij}$ in the tree $T_{y_1, \ldots, y_{i - 1}}$ which is itself contained in the tree $T_\cS^{(x)}$.  In particular, the leaf for $y_\ell \in P_\ell$ in the tree $T_{y_1, \ldots, y_{\ell - 1}}$ which is contained in $T_\cS^{(x)}$ is denoted $(y_\ell)^{(x)}_{y_1, \ldots, y_{\ell - 1}}$.  We will sometimes use the shorthand $(y_1 \cdots y_\ell)^{(x)}$ to refer to $(y_\ell)^{(x)}_{y_1, \ldots, y_{\ell - 1}}$.  The node $(y_\kappa)_{y_1, \ldots, y_{\kappa - 1}}^{(x)}$ is the leaf of the tree $B(\bmu)$ within the tree $T_\cS^{(x)}$.  For nodes in the tree $T_\cS$, we simply omit the superscript $(x)$ and write $(y_i P_{ij})_{y_1, \ldots, y_{i - 1}}$ for the node $y_i P_{ij}$ in the tree $T_{y_1, \ldots, y_{i - 1}}$ which is contained in the tree $T_\cS$.  For each $x, y \in G$, we connect three new leaves $y_\la^{(x)}$, $y_\ra^{(x)}$ and $y_=^{(x)}$ to the leaf $y^{(x)}$ and add them to $X(\cS, \bmg(\alpha))$.  For all $x, y, z \in G$ such that $x y = z$, we draw an edge from $y_\la^{(x)}$ to $x_\ra^{(y)}$ and from $x_\ra^{(y)}$ to $y_=^{(z)}$ and color $y_\la^{(x)}$ ``left'', $x_\ra^{(y)}$ ``right'' and $y_=^{(z)}$ ``equal''.  We color the root node of $T_S$ ``root''; all other nodes that have not yet been assigned a color are colored ``internal''.
\end{definition}

\newcommand{\TSsolgenconefig}{
  \begin{figure}[h]
    \centering
    \inshortorlong{\includegraphics[scale=0.8]{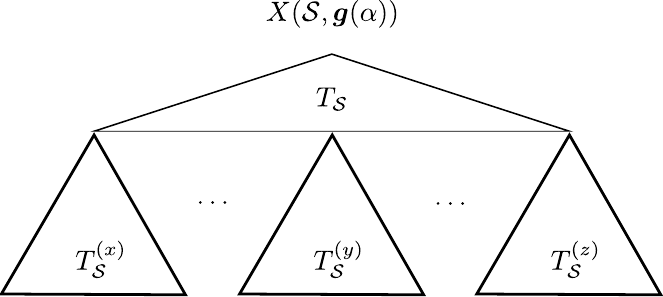}}{\includegraphics{TS-sol-gen-cone.pdf}}
    \inshortorlong{\caption{$X(\cS, \bmg(\alpha))$ where $x, y, z \in G$}}{\caption{The cone graph $X(\cS, \bmg(\alpha))$ where $x, y, z \in G$}}
    \label{fig:TS-sol-gen-cone}
  \end{figure}}

\TSsolgenconefig

We say that an isomorphism $\phi : G \ra H$ respects $\bmg(\alpha)$ and $\bmh(\alpha)$ if $\phi(g_{ij}) = h_{ij}$ for all $i$ and $j$.  Let us denote by $\iso_{\bmg(\alpha) \mapsto \bmh(\alpha)}(\cS, \cS')$ the set of all isomorphisms $\phi$ from $\cS$ to $\cS'$ which respect $\bmg(\alpha)$ and $\bmh(\alpha)$.  We say that $(\cS, \bmg(\alpha))$ and $(\cS', \bmh(\alpha))$ are isomorphic if there is some isomorphism in $\iso_{\bmg(\alpha) \mapsto \bmh(\alpha)}(\cS, \cS')$ in which case we write $(\cS, \bmg(\alpha)) \cong (\cS', \bmh(\alpha))$.  We now show a bijection from $\iso_{\bmg(\alpha) \mapsto \bmh(\alpha)}(\cS, \cS')$ to $\iso(X(\cS, \bmg(\alpha)), X(\cS', \bmh(\alpha)))$.

\begin{theorem}
  \label{thm:sol-graph-red-bij}
  Let $(G, \cP, \cS, \bmg(\alpha), \kappa, \alpha)$ and $(H, \cQ, \cS', \bmh(\alpha), \kappa, \alpha)$ be Hall composition series with fixed generators such that the partial function $\psi : P_1 \cdots P_\kappa \ra Q_1 \cdots Q_\kappa : g_{ij} \mapsto  h_{ij}$ for all $i$ and $j$ extends to an isomorphism $\psi : P_1 \cdots P_\kappa \ra Q_1 \cdots Q_\kappa$ where $\psi[P_{ij}] = Q_{ij}$ for $i \leq \kappa$ and all $j$.  Then there is a bijection between $\iso_{\bmg(\alpha) \mapsto \bmh(\alpha)}(\cS, \cS')$ and $\iso(X(\cS, \bmg(\alpha)), X(\cS', \bmh(\alpha)))$.
\end{theorem}

We give a sketch and defer the proof to \inshortorlong{the full version~\cite{rosenbaum2012b}}{\appref{sol-group-iso-proofs}}.  The proof is similar to that of \thmref{p-graph-red-bij} but the details are more complicated due to the presence of Sylow bases and fixed generators.  Consider an isomorphism $\phi \in \iso_{\bmg(\alpha) \mapsto \bmh(\alpha)}(\cS, \cS')$.  We define $\hat \phi$ to be the map from $X(\cS, \bmg(\alpha))$ to $X(\cS', \bmh(\alpha))$ that maps the root $G$ of $X(\cS, \bmg(\alpha))$ to the root $H$ of $X(\cS', \bmh(\alpha))$, each node $x \in G$ to $\phi(x)$ and each $y^{(x)}$ to $(\phi(y))^{(\phi(x))}$; similarly, we define $\hat \phi(y_\lambda^{(x)}) = (\phi(y))_\lambda^{(\phi(x))}$ for each $x, y \in G$ and $\lambda \in \{\la, \ra, =\}$.  As before, we extend the definition of $\hat \phi$ to the intermediate coset nodes in a manner that is consistent with the assignments already made; there is no ambiguity as there is exactly one extension which results in an isomorphism.  Although there are many details, it is relatively straightforward to show that $\hat \phi \in \iso(X(\cS, \bmg(\alpha)), X(\cS', \bmh(\alpha)))$.  This shows that $f : \iso_{\bmg(\alpha) \mapsto \bmh(\alpha)}(\cS, \cS') \ra \iso(X(\cS, \bmg(\alpha)), X(\cS', \bmh(\alpha))) : \phi \mapsto \hat \phi$ is a well-defined function.

As before, the more difficult step is showing that $f$ is a bijection.  The same argument as in the proof sketch for \thmref{p-graph-red-bij} implies that $f$ is injective.  As before, we let $\theta \in \iso(X(\cS, \bmg(\alpha)), X(\cS', \bmh(\alpha)))$ and define the bijection $\phi = \restr{\theta}{G} : G \ra H$.  Because of the colorings used for the binary trees in $X(\cS, \bmg(\alpha))$ and $X(\cS', \bmh(\alpha))$, it follows that $\phi$ respects the generators $\bmg(\alpha)$ and $\bmh(\alpha)$ and that each $\phi[P_{ij}] = Q_{ij}$ for $i \leq \kappa$.  For $i > \kappa$, we note that the unique shortest path from the root $P_1$ of $X(\cS, \bmg(\alpha))$ to $1$ passes through each node $(P_{ij})_{x_1, \ldots, x_{i - 1}}$ where each $x_k = 1$.  The elements $x \in P_{ij}$ in $X(\cS, \bmg(\alpha))$ are the nodes in $G$ that can be reached from the node $(P_{ij})_{x_1, \ldots, x_{i - 1}}$ where each $x_k = 1$ by following certain paths through nodes which are colored ``internal'' and ``not identity.''  The ``not identity'' color in these paths is used to ensure that the node $x = x_1 \cdots x_\ell$ at the end of the path where each $x_i \in P_i$ satisfies $x_i = 1$ except for $i = k$.  By considering these paths, it follows that $\phi[P_{ij}] = Q_{ij}$ for each $i > \kappa$ which implies that $\phi \in \iso_{\bmg(\alpha) \mapsto \bmh(\alpha)}(\cS, \cS')$.

\begin{corollary}
  \label{cor:sol-red-cor}
  Let $(G, \cP, \cS, \bmg(\alpha), \kappa, \alpha)$ and $(H, \cQ, \cS', \bmh(\alpha), \kappa, \alpha)$ be Hall composition series with fixed generators such that setting $\psi(g_{ij}) = h_{ij}$ for $i \leq \kappa$ and all $j$ defines an isomorphism $\psi : P_1 \cdots P_\kappa \ra Q_1 \cdots Q_\kappa$ where $\psi[P_{ij}] = Q_{ij}$ for all $i \leq \kappa$ and $j$.  Then $(\cS, \bmg(\alpha)) \cong (\cS', \bmh(\alpha))$ \ifft $X(\cS, \bmg(\alpha)) \cong X(\cS', \bmh(\alpha))$.
\end{corollary}

\begin{lemma}
  \label{lem:sol-graph}
  Let $(G, \cP, \cS, \bmg(\alpha), \kappa, \alpha)$ be a Hall composition series with fixed generators.  Then the graph $X(\cS, \bmg(\alpha))$ has degree at most $\max\{\alpha, 4\}$ and size $O(n^2)$.
\end{lemma}

\begin{proof}
  The tree $T(\cS, \bmg(\alpha))$ has size $O(n)$ and $X(\cS, \bmg(\alpha))$ contains $n + 1$ copies of this tree.  Connecting three additional nodes $y_{\ell}^{(x)}$ to each leaf $y^{(x)}$ in each $T_S^{(x)}$ only increases the number of nodes by a constant factor.  Thus, $X(\cS, \bmg(\alpha))$ has $O(n^2)$ nodes.  Let $x, y \in G$.  The degree of the leaf node $x^{(y)}$ is $4$ since it is connected to its three children $x_\lambda^{(x)}$ where $\lambda \in \{\la, \ra, =\}$ and its parent.  By construction, the degree of any internal node is at most $\alpha$.
\end{proof}

\begin{theorem}
  \label{thm:sol-hcomp-iso}
  Let $(G, \cP, \cS, \bmg(\alpha), \kappa, \alpha)$ and $(H, \cQ, \cS', \bmh(\alpha), \kappa, \alpha)$ be Hall composition series with fixed generators.  Then we can test if $(\cS, \bmg(\alpha)) \cong (\cS', \bmh(\alpha))$ in $n^{O(\alpha / \log \alpha)}$ deterministic time.
\end{theorem}

\begin{proof}
  We start by checking if the mapping from $\bmg(\alpha)$ to $\bmh(\alpha)$ extends to an isomorphism from $\psi : P_1 \cdots P_\kappa \ra Q_1 \cdots Q_\kappa$.  If this is the case then this isomorphism $\psi$ is unique and we check if $\psi[P_{ij}] = Q_{ij}$ for all $i \leq \kappa$ and $j$.  If both of these conditions hold, then we continue; otherwise, we return that $(\cS, \bmg(\alpha)) \not\cong (\cS', \bmh(\alpha))$.

  We can compute the graphs $X(\cS, \bmg(\alpha))$ and $X(\cS', \bmh(\alpha))$ in polynomial time.  By \corref{sol-red-cor}, $(\cS, \bmg(\alpha)) \cong (\cS', \bmh(\alpha))$ \ifft $X(\cS, \bmg(\alpha)) \cong X(\cS', \bmh(\alpha))$ so this reduction is correct.  By \lemref{sol-graph}, the number of nodes in $X(\cS, \bmg(\alpha))$ is $O(n^2)$ and the degree is at most $\max\{\alpha, 4\} = O(\alpha)$.   We can test if $X(\cS, \bmg(\alpha)) \cong X(\cS', \bmh(\alpha))$ in $n^{O(\alpha / \log \alpha)}$ time using \thmref{const-deg-iso}.
\end{proof}

As for $p$-groups, we modify \thmref{sol-hcomp-iso} for Hall composition-series canonization.

\begin{restatable}{theorem}{solhcompcan}
  \label{thm:sol-hcomp-can}
  Let $(G, \cP, \cS)$ be a Hall composition series.  Then we can compute a canonical form $\can(\cS)$\inlong{ of $(G, \cP, \cS)$} in $n^{O(\log n / \log \log n)}$ deterministic time.
\end{restatable}

We give a sketch and defer the complete proof to \inshortorlong{the full version~\cite{rosenbaum2012b}}{\appref{p-group-iso-proofs}}.  Although the full proof is complicated, our strategy is a straightforward extension of the proof of \thmref{comp-can}.  We set $\alpha = \log n / \log \log n$ and reindex so that $p_i > p_j$ for $i < j$.  Although the graph $X(\cS, \bmg(\alpha))$ in \defref{TS-sol-gen-cone} is only defined when $p_1 \geq \alpha$, there is a trick which allows us to reduce the case where $p_1 < \alpha$ to the case where $p_1 \geq \alpha$; therefore, we shall assume that $p_1 \geq \alpha$.  As usual, we let $\kappa$ be the largest value of $i$ such that $p_i$ is at least $\alpha$; we consider all possible choices of ordered generating sets $\bmg(\alpha)$ for the subgroup $P_1 \cdots P_\kappa$ which correspond to representatives of the factor groups of each $S_i$ for $i \leq \kappa$.  For each choice $\bmg(\alpha)$, we compute $\can(X(\cS, \bmg(\alpha)))$ and define $\invar(\cS)$ to be the graph $\can(X(\cS, \bmg(\alpha)))$ which comes first lexicographically.  It is easy to show that $\invar(\cS)$ is a complete polynomial-size invariant for $\cS$.  A multiplication table for a group with the underlying set $[n]$ which is isomorphic to $G$ can then be computed using arguments similar to those from \thmref{comp-can}.  We then compute the subset of $[n]$ which corresponds to each $P_i$ and each $P_{ij}$ by following certain paths in $\can(X(\cS, \bmg(\alpha)))$ which start at nodes of the form $(P_{ij})_{x_1, \ldots, x_{i - 1}}$ with each $x_k = 1$ where the nodes are colored ``internal'' at first and then ``not identity.''  Because the computations described are deterministic and depend only on $\invar(\cS)$, the composition length of each $P_i$ and $n$, we obtain a canonical form for $(G, \cP, \cS)$.

\section{Algorithms for solvable-group\insodaornot{\\}{ }isomorphism}
\label{sec:sol-algorithms}
In this section, we derive our algorithms for solvable-group isomorphism.  A generalization of \thmref{sol-group-iso} now follows from the results of Sections~\ref{sec:hcomp-red} and \ref{sec:hgraph-red}.

\begin{theorem}
  \label{thm:sol-group-iso-drq}
  Let $G$ and $H$ be solvable groups.  Then we can test if $G \cong H$ in 
  
  \begin{enumerate}
  \item $n^{(1 / 2) \log_p n + O(\log n / \log \log n)}$ \inshort{deterministic\\}time\inlong{ for deterministic algorithms}
  \item $n^{(1 / 4) \log_p n + O(\log n / \log \log n)}$\inshort{ randomized} time\inlong{ for randomized algorithms}
  \item $n^{(1 / 6) \log_p n + O(\log n / \log \log n)}$\inshort{ quantum} time\inlong{ for quantum algorithms}
  \end{enumerate}
  where $p$ is the smallest prime dividing the order of the group.
\end{theorem}

\begin{proof}
  From part (a) of \thmref{sol-red-hcomp-drq}, there are at most $n^{(1 / 2) \log_p n + O(1)}$ calls to the algorithm for Hall composition series canonization where $p$ is the smallest prime dividing the order of the group.  Each such call requires $n^{O(\log n / \log \log n)}$ time by part (a) of \thmref{sol-hcomp-can} for an overall runtime of $n^{(1 / 2) \log_p n + O(\log n / \log \log n)}$.  Parts (b) and (c) follow from parts (b) and (c) of \thmref{sol-red-hcomp-drq}.
\end{proof}

As we show in \inshortorlong{the full version~\cite{rosenbaum2012b}}{\appref{group-can}}, the deterministic variant of our algorithm can be adapted to perform solvable-group canonization.  It is also easy to show that the analysis of \thmref{sol-group-iso-drq} is tight.

  \begin{theorem}
    \label{thm:sol-tight}
    If $n$ is sufficiently large, there exists an Abelian group of order $n$ for which part (a) of \thmref{sol-group-iso-drq} requires at least $n^{(1 / 2) \log_p n + \Omega(\log n / \log \log n)}$ time.
  \end{theorem}

\inlong{
  \begin{proof}
    Consider the Abelian $p$-group $\bbZ_p^k$ of order $\Theta(n / \log n)$ and the cyclic group $\bbZ_q$ of prime order $(1 / 2) \log n / \log \log n \leq q < \log n / \log \log n$.  Provided $n$ is large enough, it follows from Bertrand's Postulate that we can always construct such groups.  We choose $G = H = \bbZ_p^k \times \bbZ_q \times \bbZ_d$ where $d = n / p^k q$.  Working through the analysis of \thmref{sol-group-iso-drq}, we find that our algorithm takes $n^{(1 / 2) \log_p n + \Omega(\log n / \log \log n)}$ time.
  \end{proof}}

\section*{Acknowledgements}
I thank L{\'a}szl{\'o} Babai for helping clarify the randomized and quantum algorithms versions of the algorithm and contributing a reference, Paul Beame and Aram Harrow for useful discussions and feedback, Joshua Grochow for additional references, Fabian Wagner for discussions on his algorithm for composition-series isomorphism and the anonymous reviewers for helpful comments.  Richard Lipton asked if the techniques used for $p$-groups could be applied to groups of order $2^a p^b$ where $p$ is an odd prime.  This inspired me to find a more efficient algorithm for solvable-group isomorphism.  I was funded by the DoD AFOSR through an NDSEG fellowship.  Partial support was provided by the NSF under grant CCF-0916400.

\inlong{
  \appendix
  
  \newpage
  \section{$p$-Group isomorphism proofs}
  \label{app:p-group-iso-proofs}
  We now supply the proofs for our $p$-group isomorphism algorithm which were omitted from the main body of the paper.

\subsection{The correctness proofs for $p$-group isomorphism}

We prove a version of \thmref{p-graph-red-bij} with an explicit bijection.

\begin{theorem}
  \label{thm:p-graph-red-bij-exp}
  Let $S$ and $S'$ be composition series for the groups $G$ and $H$ which consist of the subgroups $G_0 = 1 \tril \cdots \tril G_m = G$ and $H_0 = 1 \tril \cdots H_m = H$.  Define $f : \phi \mapsto \hphi$ where the map $\hphi : X(S) \ra X(S')$ acts as follows for $\phi \in \iso(S, S')$:

  \begin{enumerate}
  \item \label{itm:root-mapsto}
    For the root, $\hphi(G) = H$.
  \item \label{itm:TS-mapsto}
    For each node $x G_i \in G / G_i$ in $T_S$, $\hphi(x G_i) = \phi[x G_i] = \phi(x) H_i$.
  \item \label{itm:TSx-mapsto}
    For each node $(y G_i)^{(x)}$ in $T_S^{(x)}$, $\hphi((y G_i)^{(x)}) = (\phi[y G_i])^{(\phi(x))} = (\phi(y) H_i)^{(\phi(x))}$.
  \item \label{itm:almost-leaf-mapsto}
    For each node $y^{(x)}$ in $T_S^{(x)}$ with $x, y \in G$, $\hphi(y^{(x)}) = \phi(y)^{(\phi(x))}$.
  \item \label{itm:leaf-mapsto}
    For each leaf node $y_\lambda^{(x)}$ in $T_S^{(x)}$ where $x, y \in G$ and $\lambda \in \{\la, \ra, =\}$, $\hphi(y_\lambda^{(x)}) = \phi(y)_\lambda^{(\phi(x))}$.
  \end{enumerate}
  Then $f : \iso(S, S') \ra \iso(X(S), X(S')) : \phi \ra \hphi$ is a well-defined bijection.
\end{theorem}

\begin{proof}
  We note that \itmref{root-mapsto} and \itmref{almost-leaf-mapsto} are special cases of \itmref{TS-mapsto} and \itmref{TSx-mapsto}.  It therefore suffices to consider the other three lines.

  Let $\phi \in \iso(S, S')$ and observe that $\phi[G_i] = H_i$ for all $i$ so the equalities in the definition of $\hphi$ hold.  Also, \itmref{TS-mapsto} and \itmref{TSx-mapsto} both apply to the root node of $T_S^{(x)}$ since $x \in G$ is identified with $(G)^{(x)}$.  By \itmref{TS-mapsto} we have $\hphi(x) = \phi(x)$ while by \itmref{TSx-mapsto}, $\hphi((G)^{(x)}) = (H)^{(\phi(x))}$.  Since $\phi(x)$ was identified with $(H)^{(\phi(x))}$, there is no ambiguity.  Thus, $\hphi$ is well-defined.

  By definition $\hphi$ maps the root node of $X(S)$ to the root node of $X(S')$ and the nodes in the \nth{k} level of $X(S)$ to the \nth{k} level of $X(S')$.  Then it follows from our definition that $\hphi$ respects the colors on all nodes.  We start by proving that $\hphi$ is a injection.  It suffices to consider nodes in the same level.  Let $x G_i, y G_i \in G / G_i$ be nodes in $T_S$ and suppose that $\hphi(x G_i) = \hphi(y G_i)$.  Then $\phi(x) H_i = \phi(y) H_i$ so $\phi(x^{-1} y) \in H_i$ and $x^{-1} y \in \phi^{-1}[H_i] = G_i$.  Therefore, $x G_i = y G_i$.  Consider nodes $(y G_i)^{(x)}$ in $T_S^{(x)}$ and $(z G_i)^{(w)}$ in $T_S^{(w)}$ and suppose $\hphi((y G_i)^{(x)}) = \hphi((z G_i)^{(w)})$.  Then $(\phi(y) H_i)^{(\phi(x))} = (\phi(z) H_i)^{(\phi(w))}$.  This implies that $\phi(x) = \phi(w)$ so that $x = w$.  Hence $\phi(y) H_i = \phi(z) H_i$ and $y G_i = z G_i$ so that $(y G_i)^{(x)} = (z G_i)^{(w)}$.  Then it follows from the definition of $\hphi$ that it acts as an injection on the leaf nodes $(y)_\lambda^{(x)}$.  Thus, $\hphi$ is an injection so since $X(S)$ and $X(S')$ have the same cardinality it is a bijection.

  Our next step is to show that $\hphi$ respects the tree edges.  Consider an edge between the nodes $y G_i \in G / G_i$ and $x G_{i + 1} \in G / G_{i + 1}$ in $T_S$.  Then $y G_i \subseteq x G_{i + 1}$ so $\hphi(y G_i) = \phi[y G_i] \subseteq \phi[x G_{i + 1}] = \hphi(x G_{i + 1})$.  Since $\hphi(y G_i) = \phi(y) H_i$ and $\hphi(x G_{i + 1}) = \phi(x) H_{i + 1}$, there is a tree edge between $\hphi(y G_i)$ and $\hphi(x G_{i + 1})$ in $X({S'})$.  For a tree edge between the nodes $(z G_i)^{(x)}$ and $(y G_{i + 1})^{(x)}$ in $T_S^{(x)}$, we again have $z G_i \subseteq y G_{i + 1}$ and so the same analysis shows that there is a tree edge between $\hphi((z G_i)^{(x)})$ and $\hphi((y G_{i + 1})^{(x)})$ in $X(S')$.  Finally, by definition $\hphi$ preserves the edges from each node $y^{(x)}$ to the leaves $y_\lambda^{(x)}$.

  We now consider the cross edges.  Suppose that $x y = z$; we will show that the edges from $y_\la^{(x)}$ to $x_\ra^{(y)}$ and from $x_\ra^{(y)}$ to $y_=^{(z)}$ are preserved by $\hphi$.  By definition, we have

  \begin{align*}
    \hphi(y_\la^{(x)}) & = \phi(y)_\la^{(\phi(x))} \\
    \hphi(x_\ra^{(y)}) & = \phi(x)_\ra^{(\phi(y))} \\
    \hphi(y_=^{(z)})   & = \phi(y)_=^{(\phi(z))}
  \end{align*}

  Since $\phi : G \ra H$ is an isomorphism, $\phi(x) \phi(y) = \phi(z)$ so there are cross edges from $\phi(y)_\la^{(\phi(x))}$ to $\phi(x)_\ra^{(\phi(y))}$ and from $\phi(x)_\ra^{(\phi(y))}$ to $\phi(y)_=^{(\phi(z))}$ in $X(S')$.  Since $X(S)$ and $X(S')$ have the same number of edges, it follows that there is an edge between two nodes in $X(S)$ \ifft there is an edge between their images under $\hphi$ in $X(S')$.  Thus, $\hphi$ is an isomorphism so we have shown that $f : \phi \mapsto \hphi$ is a map $f : \iso(S, S') \ra \iso(X(S), X(S'))$.

  It remains to show that $f$ is bijective.  Suppose that $\phi_1, \phi_2 \in \iso(S, S')$ and $\hphi_1 = \hphi_2$.  Then in particular we have $\hphi_1(x) = \hphi_2(x)$ for $x \in G$ so $\phi_1(x) = \phi_2(x)$ and $\phi_1 = \phi_2$.  Therefore, $f$ is injective.  Our final task is to show that it is surjective.  Let $\theta \in \iso(X(S), X(S'))$.  Because the roots $G$ and $H$ of $X(S)$ and $X(S')$ are the only nodes colored ``root'', we have $\theta[G] = H$ and $\theta$ maps the \nth{k} level of $X(S)$ to the \nth{k} level of $X(S')$.  Then the map $\phi : G \ra H : x \mapsto \theta(x)$ where $x \in G$ is a node in $T_S$ is a well-defined bijection.

  We now claim that for any node $y^{(x)}$ in $T_S^{(x)}$, $\theta(y^{(x)}) = \phi(y)^{(\phi(x))}$.  We know that $T_S^{(x)}$ is rooted at the node $x$ of $T_S$ and $\theta(x) = \phi(x)$ by definition of $\phi$.  It follows that the node $\theta(y^{(x)})$ is in $T_{S'}^{(\phi(x))}$; similarly $\theta(x^{(y)})$ is in $T_{S'}^{(\phi(y))}$.  Hence, we can write $\theta(y^{(x)}) = \phi(b)^{(\phi(x))}$ and $\theta(x^{(y)}) = \phi(a)^{(\phi(y))}$ for some $a, b \in G$.  We know that in $X(S')$, the unique shortest path from $\phi(x)$ to $\phi(y)$ that passes through a node colored ``left'' and later a node colored ``right'' follows the tree $T_{S'}^{(\phi(x))}$ to $\phi(y)^{(\phi(x))}$, then the path $(\phi(y)^{(\phi(x))}, \phi(y)_\la^{(\phi(x))}, \phi(x)_\ra^{(\phi(y))}, \phi(x)^{(\phi(y))})$ and then the tree $T_{S'}^{(\phi(y))}$ to $\phi(y)$.  On the other hand, we know that the path $(\phi(b)^{(\phi(x))}, \phi(b)_\la^{(\phi(x))}, \phi(a)_\ra^{(\phi(y))}, \phi(a)^{(\phi(y))})$ exists in $X(S')$ by definition of $a$ and $b$.  Then a shortest path that passes through a node colored ``left'' and later a node colored ``right'' from $\phi(x)$ to $\phi(y)$ follows the tree $T_{S'}^{(\phi(x))}$ to $\phi(b)^{(\phi(x))}$, then the path $(\phi(b)^{(\phi(x))}, \phi(b)_\la^{(\phi(x))}, \phi(a)_\ra^{(\phi(y))}, \phi(a)^{(\phi(y))})$ and then the tree $T_{S'}^{(\phi(y))}$ to $\phi(y)$.  From uniqueness, it follows that $\phi(x) = \phi(a)$ and $\phi(y) = \phi(b)$ so $x = a$ and $y = b$.  Thus, for any node $y^{(x)}$ in $T_S^{(x)}$, $\theta(y^{(x)}) = \phi(y)^{(\phi(x))}$.

  Now, we show that $\phi$ is an isomorphism.  We have already shown that it is bijective.  Let $x, y \in G$ and set $z = x y$.  Then $\theta$ maps the path $(y_\la^{(x)}, x_\ra^{(y)}, y_=^{(z)})$ in $X(S)$ to the path $(\phi(y)_\la^{(\phi(x))}, \phi(x)_\ra^{(\phi(y))}, \phi(y)_=^{(\phi(z))})$ in $X(S')$.  Since $\theta$ preserves colors, the colors of the nodes in this path are $(\text{``left''}, \text{``right''}, \text{``equal''})$.  We conclude that $\phi(x) \phi(y) = \phi(z)$.  By definition of $z$, it follows that $\phi$ is an isomorphism from $G$ to $H$.  We claim that $\phi[G_i] = H_i$.  Consider the node $G_i \in G / G_i$; the nodes $x \in G_i$ are the leaves of the tree $T_S$ in $X(S)$ which are descendants of $G_i$.  Similarly, for the node $H_i \in G / H_i$; the nodes $y \in H_i$ are the leaves of the tree $T_{S'}$ in $X(S')$ which are descendants of $H_i$.  Because $\phi$ is an isomorphism, $\phi(1) = 1$ so that $\theta(1) = 1$.  It follows that $\theta$ maps the unique shortest path from $G$ to $1$ in $X(S)$ to the unique shortest path from $H$ to $1$ in $X(S')$.  This implies that $\theta[G_i] = H_i$ so that $\restr{\theta}{G_i}$ is a bijection from $G_i$ to $H_i$.  Therefore, $\phi[G_i] = H_i$ so $\phi$ is an isomorphism from $S$ to $S'$.  Then $\theta = \hphi$.  Hence, we have proven that $f$ is a bijection.
\end{proof}

\compcan*

\begin{proof}
  We start by showing how to obtain a complete polynomial-size invariant $\invar(S)$ for $S$ and then use it to get a canonical form.  The canonical form $\can(X(S))$ of $X(S)$ can be computed in $n^{O(\alpha)}$ time by \thmref{const-deg-can}.  We then let $\invar(S) = \can(X(S))$.  Let $S'$ be a composition series for a group $H$.  Then $S \cong S'$ \ifft $\invar(S) = \invar(S')$ since we have $S \cong S'$ \ifft $X(S) \cong X(S')$ by \corref{p-red-cor}.  Thus, $\invar(S)$ is a complete polynomial-size invariant for $S$.

  We now construct a canonical form $\can(S)$ for $S$.  There exists an isomorphism $\theta : X(S) \ra \can(X(S))$.  We start by computing $\invar(S) = \can(X(S))$ and then locate the root node $\theta(G)$ of $\can(X(S))$; this is easy since the root is the only node colored ``root.''  We then note that the composition length $m$ of $G$ is equal to the distance from the root of $X(S)$ to the nodes $x \in G$; this allows us to find the nodes $\theta(x)$ in $\can(X(S))$ where $x \in G$ using breadth-first search.  For each such $x \in G$, we then define $\lambda_G(\theta(x))$ to be the index of $\theta(x)$ in the sequence of nodes $\theta(x)$ where $x \in G$ obtained in the breadth-first search.  Thus, $\lambda_G : \theta[G] \ra [n]$ is a bijection.

  Let $x, y \in G$; we will now show how to compute $\theta(xy)$ from $\can(X(S))$ given $\theta(x)$ and $\theta(y)$.  Note that in $X(S)$, there is a unique shortest path from $x$ to $y$ that passes through a node colored ``left'' and later a node colored ``right''; namely, this is the path which follows the tree $T_S^{(x)}$ from $x$ to $y^{(x)}$, then the path $(y^{(x)}, y_\la^{(x)}, x_\ra^{(y)}, x^{(y)})$ and then the tree $T_S^{(y)}$ to $y$.  Thus, given nodes $\theta(x)$ and $\theta(y)$ in $\can(X(S))$, we can find the nodes $\theta(y^{(x)})$ and $\theta(x^{(y)})$ in $\can(X(S))$.  Once this has been done for all $x, y \in G$, the multiplication table defined by the rule $\theta(x)\theta(y) = \theta(xy)$ can be determined by inspecting the multiplication gadgets in $\can(X(S))$.

  Using the multiplication table just computed, we can find the node $\theta(1)$ in $\can(X(S))$ which corresponds to the identity.  The unique shortest path from the root $G$ of $X(S)$ to $1$ is $(G_m = G, \ldots, G_0 = 1)$.  Thus, we can find the node $\theta(G_i)$ in $\can(X(S))$ which corresponds to each $G_i$.  Moreover, for each $i$, the elements $x \in G_i$ in $X(S)$ are precisely the nodes $y \in G$ which can be reached from the node $G_i$ by moving away from the root.  Then for each $i$, we can find the nodes $\theta(x)$ in $\can(X(S))$ where $x \in G_i$ by breadth-first search on $\can(X(S))$.

  Define $\phi = \restr{\theta}{G}$ and $\psi_G = \lambda_G \phi$.  Let $M(G)$ be the $n \times n$ matrix with elements in $[n]$ defined by $(M(G))_{\psi_G(x), \psi_G(y)} = \psi_G(xy)$ for all $x, y \in G$.  We compute $\psi_G[G_i]$ for each $i$ and define the canonical form of the composition series as $\can(S) = (M(G), \psi_G[G_0], \ldots, \psi_G[G_m])$.

  We claim that $\can(S)$ is a canonical form for $S$.  Let $S'$ be a composition series for a group $H$.  By construction, $\psi_G : G \ra [n]$ is an isomorphism from $G$ to the group described by the multiplication table $M(G)$.  If $\can(S) = \can(S')$ then $\psi_H^{-1} \psi_G$ is an isomorphism from $S$ to $S'$ so $S \cong S'$.  Suppose that $S \cong S'$; then $\invar(S) = \invar(S')$.  Since we computed $\can(S)$ deterministically and this computation depends only on $\invar(S)$, $n$ and the composition length of $G$, it follows that $\can(S) = \can(S')$.
\end{proof}

\compcor*

\begin{proof}
  First, if $G$ is simple then $1 \tril G$ is a composition series for $G$; also, $\soc(G) = G$ in this case and so $G$ is properly labelled as ``socle.''  Otherwise, suppose $\soc(G) = G$.  \algref{sim-sub} is correct by definition so \linref{comp-T} sets $T$ to the set of all simple minimal normal subgroups of $\soc(G)$.

  We know that $G$ can be written as a direct product $\bigtimes_i S_i$ where each $S_i$ is a simple subgroup of $G$.  Because each $S_i$ is part of this direct product, it is also normal in $G$.  Moreover each $S_i$ is a minimal normal subgroup since if $N \leq S_i$ and $N \trile G$ then $N \in \{1, S_i\}$ since $S_i$ is simple.  Therefore, $\soc(G)$ is a direct product of the simple minimal normal subgroups of $\soc(G)$.

  The while loop on \linref{prod-loop} builds up a composition series of direct products of simple minimal normal subgroups in $T$ for $\soc(G)$.  Line~\ref{line:choose-L} chooses a simple minimal normal subgroup $L$ of $\soc(G)$ that can be added to the direct product.  Lines~\ref{line:set-T'} -- \ref{line:loop-reps-end} remove redundant subgroups from $T$ that form the same direct product with $K$ leaving only one subgroup that forms each product.  Since this does not change the set of direct products that can be formed, this part of the algorithm affects only efficiency but not correctness\footnote{These lines only matter when we consider all possible choices on \linref{choose-L}.}.  We claim that there is always a choice of $L$ on \linref{choose-L} such that $K \cap L = 1$.  Let us write $K$ as the direct product $\bigtimes_i L_i$ where each $L_i$ is a simple minimal normal subgroup of $\soc(G)$ chosen from $T$ during a previous iteration of the loop.  From the loop condition, $K \tril \soc(G)$.  We know that $\soc(G)$ is characteristically simple so $K$ is not a characteristic subgroup of $\soc(G)$.  Then there exists $\phi \in \aut(\soc(G))$ such that $\phi[K] \not= K$ so for some $i$, $\phi[L_i] \not\subseteq K$.  Since $L_i$ is a minimal normal subgroup of $\soc(G)$ so is $\phi[L_i]$ and it follows that $K \cap \phi[L_i] = 1$.  Thus, there is a simple minimal normal subgroup of $\soc(G)$ which intersects $K$ trivially.  Since lines~\ref{line:set-T'} -- \ref{line:set-T} do not affect which direct products can be formed using subgroups in $T$ with $K$, either $\phi[L_i] \in T$ or there is some other simple minimal normal subgroup of $\soc(G)$ in $T$ that forms the same product with $K$.  Thus, a valid choice of $L$ always exists on \linref{choose-L}.

  We have shown that the part of $S$ from $1$ to $\soc(G)$ is indeed a composition series.  This suffices to prove the base case where $\soc(G) = G$.

  We note that $K_0 \tril \cdots \tril K_m$ is a composition series for $G / \soc(G)$ by induction.  We have already argued that the part of $S$ from $1$ to $\soc(G)$ is a composition series.  The factors of the portion of $S$ from $\soc(G)$ to $G$ are simple since the subgroups in that part of the series are the preimages of the composition series $K_0 \tril \cdots \tril K_m$ for $G / \soc(G)$.  It follows that $S$ is a composition series.  Moreover, all socles and their preimages are labelled ``socle'' since each socle is labelled ``socle'' when it is initially added to the series and labels are preserved when taking preimages.
\end{proof}

\subsection{Obtaining the $n^{(1 / 2) \log n + O(1)}$ runtime for $p$-group isomorphism}
\begin{lemma}
  \label{lem:square-sum}
  If $0 \leq \eps \leq x, y$ then $x^2 + y^2 \leq (x + y - \eps)^2 + \eps^2$.
\end{lemma}

\begin{proof}
  We see that $x^2 + y^2 \leq (x + y - \eps)^2 + \eps^2$ holds \ifft $0 \leq x y - \eps (x + y) + \eps^2$.  Define $f(\eps) = x y - \eps (x + y) + \eps^2$.  Then $f'(\eps) = 2 \eps - (x + y) \leq 0$ for $0 \leq \eps \leq x, y$.  AWLOG that $x \leq y$.  Then the minimal value of $f$ occurs at $\eps = x$ in which case $f(\eps) = 0$.  It follows that $0 \leq x y - \eps (x + y) + \eps^2$ for all $0 \leq \eps \leq x, y$.
\end{proof}

\maxtwonorm*

\begin{proof}
  Choose some $x \in \bbR^d$ that satisfies the constraints.  Then by \lemref{square-sum}

  \begin{align*}
    \norm{x}_2^2 & =    x_1^2 + \cdots + x_d^2 \\
              {} & \leq 1 + x_2^2 + \cdots + x_{d - 1} + (x_d + x_1 - 1)^2 \\
              {} & \leq (d - 1) + (x_d + \sum_{i = 1}^{d - 1} (x_i - 1))^2 \\
              {} & =    (d - 1) + (t - (d - 1))^2
  \end{align*}
  Moreover, we note that $(d - 1) + (t - (d - 1))^2$ is achieved by $x^*$ where $x^*_i = 1$ for $i < d$ and $x^*_d = t - (d - 1)$.
\end{proof}

\smallpbound*

\begin{proof}
  Define

  \begin{align*}
    f(p) & = c_1 \log_p n + c_2 p \\
      {} & = c_1 \ln n / \ln p + c_2 p
  \end{align*}
  We fix $n$ and choose $2 \leq p \leq c_1 \ln n / c_2 \ln^2 p$ such that $f(p)$ is maximal.  We have
  
  \begin{align*}
    f'(p) & = - \frac{c_1 \ln n}{p \ln^2 p} + c_2
  \end{align*}
  Then $f'(p) \leq 0$ \ifft

  \begin{align*}
    c_2 & \leq \frac{c_1 \ln n}{p \ln^2 p} \\
      p & \leq \frac{c_1 \ln n}{c_2 \ln^2 p}
  \end{align*}
  It follows that $f(p) \leq f(2)$ for $2 \leq p \leq c_1 \ln n / c_2 \ln^2 p$ from which our result is immediate.
\end{proof}

\largepbound*

\begin{proof}
  First suppose that $c \ln n / \ln^2 p \leq p \leq e^{\sqrt[3]{\ln n}}$.  Then $\ln \ln p \leq (1 / 3) \ln \ln n$ so

  \begin{align*}
    \log_p n & =    \frac{\ln n}{\ln p} \\
          {} & \leq \frac{\ln n}{\ln(c \ln n / \ln^2 p)} \\
          {} & =    \frac{\ln n}{\ln c + \ln \ln n - 2 \ln \ln p} \\
          {} & =    O(\ln n / \ln \ln n) \\
          {} & =    O(\log n / \log \log n)
  \end{align*}
  For the case where $p \geq e^{\sqrt[3]{\ln n}}$ we have $\log_p n \leq (\ln n)^{2 / 3} = o(\log n / \log \log n)$. 
\end{proof}

  \newpage
  \section{Solvable-group isomorphism proofs}
  \label{app:sol-group-iso-proofs}
  \subsection{Computing a Sylow basis in deterministic polynomial time}
We now prove that a Sylow basis of a solvable group can be computed deterministically in polynomial time.  First, it is necessary to introduce a few more definitions and results.

\begin{definition}
  Let $\pi$ be a set of primes.  A group $G$ is a $\pi$-group if for every $g \in G$, the prime divisors of the order of $g$ are contained in $\pi$.
\end{definition}

\begin{definition}
  Let $G$ be a group whose order has the prime factorization $n = \prod_{i = 1}^\ell p_i^{e_i}$ and let $\pi$ be a subset of $\setb{p_i}{1 \leq i \leq \ell}$.  A Hall $\pi$-subgroup of $G$ is a $\pi$-subgroup of order $\prod_{p_i \in \pi} p_i^{e_i}$.
\end{definition}

For convenience, we denote by $p_{-i}$ the set of all primes except for $p_i$ which divide the order of $G$.

\begin{definition}
  Let $G$ be a group whose order has the prime factorization $n = \prod_{i = 1}^\ell p_i^{e_i}$.  A Sylow system of $G$ is a collection of subgroups $\cQ = \setb{Q_i}{1 \leq i \leq \ell}$ such that each $Q_i$ is a Hall $p_{-i}$-subgroup of $G$.
\end{definition}

Hall showed that a group is solvable \ifft it has a Sylow system.

\begin{theorem}[Hall~\cite{hall1938a}, cf.~\cite{rotman1995a,robinson1996a}]
  \label{thm:sol-syl-sys}
  A group $G$ is solvable \ifft it has a Sylow system.
\end{theorem}

Hall subgroups need not exist in a general group.  However, the above shows that Hall $p_{-i}$-subgroups always exist in a solvable group.  In fact it holds more generally that for any set of primes $\pi$ which divide the order of a solvable group, there exists a Hall $\pi$-subgroup.

\begin{theorem}[Hall~\cite{hall1938a}, cf.~\cite{rotman1995a,robinson1996a}]
  \label{thm:sol-contain}
  Let $G$ be a solvable group whose order has the prime factorization $n = \prod_{i = 1}^\ell p_i^{e_i}$ and let $\pi$ be a subset of $\setb{p_i}{1 \leq i \leq \ell}$.  Then every $\pi$-subgroup of $G$ is contained in a Hall $\pi$-subgroup of $G$.
\end{theorem}

We remark that since $1$ is a $\pi$-subgroup, in particular this implies that Hall $\pi$-subgroups exist in a solvable group.

\begin{lemma}
  Let $G$ be a solvable group whose order has the prime factorization $n = \prod_{i = 1}^\ell p_i^{e_i}$ and let $\pi$ be a subset of $\setb{p_i}{1 \leq i \leq \ell}$.  Then we can compute a Hall $\pi$-subgroup of $G$ in deterministic polynomial time.
\end{lemma}

\begin{proof}
  Suppose that $H$ is a $\pi$-subgroup of $G$ which is not a Hall $\pi$-subgroup of $G$.  By \thmref{sol-contain}, $H$ is properly contained in some Hall $\pi$-subgroup $K$ of $G$.  For every $g \in G \setminus H$, we can check if $\langle H, g \rangle$ is a $\pi$-subgroup of $G$ in polynomial time.  For $g \in K \setminus H$, $\langle H, g \rangle$ is a subgroup of the Hall $\pi$-group $K$ so $\langle H, g \rangle$ is a $\pi$-subgroup of $G$.  Thus, we can find some $g \in G \setminus H$ such that $L = \langle H, g \rangle$ is a $\pi$-subgroup of $G$ in polynomial time.  Since $H < L$, we have found a $\pi$-subgroup of $G$ which properly contains $H$.  Since $1$ is a $\pi$-subgroup of $G$, it follows that we can compute a Hall $\pi$-subgroup of $G$ in polynomial time.
\end{proof}

This lemma was previously proven by Kantor and Taylor~\cite{kantor1988a} in the setting of permutation groups.  It is immediate that we can compute a Sylow system of a solvable group in deterministic polynomial time.

\begin{corollary}
  \label{cor:comp-syl-sys-poly}
  Let $G$ be a solvable group.  Then we can compute a Sylow system of $G$ in deterministic polynomial time.
\end{corollary}

The desired result now follows from a result which shows a close relation between Sylow bases and Sylow systems.

\begin{proposition}[cf.~\cite{robinson1996a}]
  \label{prop:syl-sys-basis-bij}
  Let $G$ be a group whose order has the prime factorization $n = \prod_{i = 1}^\ell p_i^{e_i}$ and consider a Sylow system $\cQ = \setb{Q_i}{1 \leq i \leq \ell}$ of $G$.  Then the map defined by $\cQ \mapsto \cP$ where $\cP = \setb{P_i}{1 \leq i \leq \ell}$ and $P_i = \bigcap_{j \not= i} Q_i$ is a bijection from the Sylow systems of $G$ to the Sylow bases of $G$.
\end{proposition}

\compsylbasispoly*

\begin{proof}
  Observe that the map from \propref{syl-sys-basis-bij} can be evaluated in deterministic polynomial time given a Sylow system.  The result is then immediate from \corref{comp-syl-sys-poly}.
\end{proof}

\subsection{The correctness proof for solvable-group isomorphism}

\genord*

\begin{proof}
  We consider the Cayley graph $X = \cay(G, K)$ for the group $G$ with generators in $K = \{g_1, \ldots, g_k\}$.  We define the order $\preceq_{\bmg}$ as follows.  Let $W_{\bmg} = K^*$.  For each $x \in G$, we say that $w = (x_1, \ldots, x_j) \in W_{\bmg}$ represents $x$ if $x = x_1 \cdots x_j$ where the empty product corresponds to the identity $e$.  For $x, y \in W_{\bmg}$, we say that $x \preceq^{\bmg} y$ if

  \begin{enumerate}
  \item $\abs{x} < \abs{y}$ or
  \item $\abs{x} = \abs{y}$ and $x$ is lexicographically at most $y$ with respect to the ordering $g_1 \preceq^{\bmg} \cdots \preceq^{\bmg} g_k$.
  \end{enumerate}

  Among the set of all words which represent $x$, there is some word of minimal length $j_{\bmg}(x)$.  We can calculate $j_{\bmg}(x)$ in polynomial time using a simple breadth-first search starting from $1$ in $X$.  Among the words of length $j_{\bmg}(x)$ which represent $x$, there is some word $w_x$ which is lexicographically less than every other word of length $j_{\bmg}(x)$ which represents $x$.  We show that $w_x$ can be found in polynomial time.  To do this, we check for each $g_i$ if there is a path from $1$ to $x$ of length $j_{\bmg}(x)$ where the first edge is labelled $g_i$ using breadth-first search; this is the case \ifft there is a word of length $j_{\bmg}(x)$ which represents $x$ and starts with $g_i$.  We choose the smallest such $i$ for which this is the case.  Then $w_x$ starts with $g_i$; we can continue this process recursively from the node $g_i$ to compute the rest of the elements of $w_x$.

  We can now define the order $\preceq_{\bmg}$; for any $x, y \in G$, we compute $w_x$ and $w_y$ in polynomial time.  We define $x \preceq_{\bmg} y$ \ifft $w_x \preceq^{\bmg} w_y$ where $\preceq^{\bmg}$ is as defined above.  By definition, this order can be computed in polynomial time.

  Finally, consider a group $H$ which is generated by the elements $\bmh = (h_1, \ldots, h_k)$ and an isomorphism $\phi : G \ra H$ such that $\phi(g_i) = h_i$ for all $i$ and $j$.  Let $x, y \in G$.  We first show that $j_{\bmg}(x) = j_{\bmh}(\phi(x))$.  If $w = (x_1, \ldots, x_j) \in W_{\bmg}$ represents $x$, then it follows that $w' = (\phi(x_1), \ldots, \phi(x_j)) \in W_{\bmh}$ represents $\phi(x)$.  Thus, $j_{\bmg}(x) \geq j_{\bmh}(\phi(x))$; similarly, if $w' = (y_1, \ldots, y_j) \in W_{\bmh}$ represents $\phi(x)$, then $w = (\phi^{-1}(y_1), \ldots, \phi^{-1}(y_j)) \in W_{\bmg}$ represents $x$ so $j_{\bmh}(\phi(x)) \geq j_{\bmg}(x)$.  Hence, $j_{\bmg}(x) = j_{\bmh}(\phi(x))$.

  Consider the lexicographically minimal word $w_x = (x_1, \ldots, x_j) \in W_{\bmg}$ of length $j = j_{\bmg}(x)$ which represents $x$.  We claim that $w_{\phi(x)} = w'$ where $w' = (\phi(x_1), \ldots, \phi(x_j))$.  First, we note that $w'$ represents $\phi(x)$ so it follows that $w_{\phi(x)} \preceq^{\bmh} w'$.  Suppose $w_{\phi(x)} \prec^{\bmh} w'$ and let $w_{\phi(x)} = (y_1, \ldots, y_j)$; we see that $(\phi^{-1}(y_1), \ldots, \phi^{-1}(y_k)) \prec^{\bmg} w_x = (x_1, \ldots, x_k)$.  Since $(\phi^{-1}(y_1), \ldots, \phi^{-1}(y_k))$ represents $x$, this contradicts the minimality of $w_x$; we conclude that $w_{\phi(x)} = w'$.  It follows that $x \preceq_{\bmg} y$ \ifft $\phi(x) \preceq_{\bmh} \phi(y)$.
\end{proof}

We prove a version of \thmref{sol-graph-red-bij} with an explicit bijection.

\begin{theorem}
  \label{thm:sol-graph-red-bij-exp}
  Let $(G, \cP, \cS, \bmg(\alpha), \kappa, \alpha)$ and $(H, \cQ, \cS', \bmh(\alpha), \kappa, \alpha)$ be Hall composition series with fixed generators such that the partial function $\psi : P_1 \cdots P_\kappa \ra Q_1 \cdots Q_\kappa : g_{ij} \mapsto  h_{ij}$ for all $i$ and $j$ extends to an isomorphism $\psi : P_1 \cdots P_\kappa \ra Q_1 \cdots Q_\kappa$ where $\psi[P_{ij}] = Q_{ij}$ for $i \leq \kappa$ and all $j$.  For $\phi \in \iso_{\bmg(\alpha) \mapsto \bmh(\alpha)}(\cS, \cS')$ let each $\phi_i : G \ra Q_i$ be defined by $\phi_i(g) = h_i$ for $g \in G$ where $\phi(g) = h_1 \cdots h_\ell$ and each $h_j \in Q_j$.  We define $f : \phi \mapsto \hphi$ where the map $\hphi : X(\cS, \bmg(\alpha)) \ra X(\cS', \bmh(\alpha))$ acts as follows for $\phi \in \iso_{\bmg(\alpha) \mapsto \bmh(\alpha)}(\cS, \cS')$.

  \begin{enumerate}
  \item \label{itm:sol-root-mapsto}
    For the root, $\hphi(P_1) = Q_1$.
  \item \label{itm:sol-TS-mapsto}
    For each node $(x_i P_{ij})_{x_1, \ldots, x_{i - 1}}$ in $T_\cS$ with $x = x_1 \cdots x_i$ and each $x_k \in P_k$ where $i > \kappa$, $\hphi((x_i P_{ij})_{x_1, \ldots, x_{i - 1}}) = (\phi[x_i P_{ij}])_{\phi_1(x), \ldots, \phi_{i - 1}(x)} = (\phi(x_i) Q_{ij})_{\phi_1(x), \ldots, \phi_{i - 1}(x)}$.
  \item \label{itm:sol-TSx-mapsto}
    For each node $(y_i P_{ij})_{y_1, \ldots, y_{i - 1}}^{(x)}$ in $T_\cS^{(x)}$ with $x \in G$, $y = y_1 \cdots y_i$ and each $y_k \in P_k$  where $i > \kappa$, $\hphi((y_i P_{ij})_{y_1, \ldots, y_{i - 1}}^{(x)}) = (\phi[y_i P_{ij}])_{\phi_1(y), \ldots, \phi_{i - 1}(y)}^{(\phi(x))} = (\phi(y_i) Q_{ij})_{\phi_1(y), \ldots, \phi_{i - 1}(y)}^{(\phi(x))}$.
  \item \label{itm:sol-almost-leaf-mapsto}
    For each node $y^{(x)}$ in $T_\cS^{(x)}$ with $x, y \in G$, $\hphi(y^{(x)}) = \phi(y)^{(\phi(x))}$.
  \item \label{itm:sol-leaf-mapsto}
    For each leaf node $y_\lambda^{(x)}$ in $T_\cS^{(x)}$ where $x, y \in G$ and $\lambda \in \{\la, \ra, =\}$, $\hphi(y_\lambda^{(x)}) = \phi(y)_\lambda^{(\phi(x))}$.
  \end{enumerate}
  Then $f : \iso_{\bmg(\alpha) \mapsto \bmh(\alpha)}(\cS, \cS') : \phi \ra \hphi$ is a well-defined bijection.
\end{theorem}

\begin{proof}
  The proof is similar to that of \thmref{p-graph-red-bij} but is significantly more complicated due to the trees which were glued together and the generators which were fixed in \defref{TS-sol-gen}.  Let $\phi \in \iso_{\bmg(\alpha) \mapsto \bmh(\alpha)}(\cS, \cS')$.  Observe that for $g_i \in P_i$, $\phi(g_i) = \phi_i(g_i)$ so $\restr{\phi_i}{P_i} : P_i \ra Q_i$ is an isomorphism.  We note that the maps $\phi_i$ need not be homomorphisms; however, this is not necessary for our proof.  

  We shall assume throughout that $\kappa < \ell$ since otherwise the proof is easy as we have $\abs{\iso_{\bmg(\alpha) \mapsto \bmh(\alpha)}(\cS, \cS')} = \abs{\iso(X(\cS, \bmg(\alpha)), X(\cS', \bmh(\alpha)))} = 1$.  Then \itmref{sol-almost-leaf-mapsto} is a special case of \itmref{sol-TSx-mapsto} so it suffices to consider the other four lines.

  We remark that we did not define $\hphi$ on all the nodes in the binary trees $B(\bmu)$ from \defref{TS-sol-gen}.  This is intentional as we defined $\hphi$ on the root and leaves of $B(\bmu)$ so there is at most one definition of $\hphi$ on the remaining nodes of $B(\bmu)$ which respects the edges (and we shall argue that it exists).

  By assumption, $\phi[P_{ij}] = Q_{ij}$ for all $i$ and $j$ so the equalities in the definition of $\hphi$ hold.  We also see that \itmref{sol-TS-mapsto} applies to the root of each tree $T_{x_1, \ldots, x_i}$ in $T_\cS$ where each $x_j \in P_j$ and $i > \kappa$ in two different ways.  This is because the root $(P_{i + 1})_{x_1, \ldots, x_i}$ is identified with the node $(x_i)_{x_1, \ldots, x_{i - 1}}$ in $T_\cS$.  However, letting $x = x_1 \cdots x_i$, we have $\hphi((P_{i + 1})_{x_1, \ldots, x_i}) = (Q_{i + 1})_{\phi_1(x), \ldots, \phi_i(x)}$ and $\hphi((x_i)_{x_1, \ldots, x_{i - 1}}) = (\phi(x_i))_{\phi_1(x), \ldots, \phi_{i - 1}(x)}$.  Since $(Q_{i + 1})_{\phi_1(x), \ldots, \phi_i(x)}$ is identified with $(\phi(x_i))_{\phi_1(x), \ldots, \phi_{i - 1}(x)}$ this isn't a problem.  A similar issue occurs when \itmref{sol-TSx-mapsto} is applied to the root of the tree $T_{y_1, \ldots, y_i}^{(x)}$ in $T_\cS$ where $x \in G$, each $y_j \in P_j$ and $i > \kappa$.  However, there is no ambiguity in this case for exactly the same reason.  Thus, $\hphi$ is well-defined.

  We proceed by proving that $\hphi$ is a injection.  By definition $\hphi$ maps the root node of $X(\cS, \bmg(\alpha))$ to the root node of $X(\cS', \bmh(\alpha))$ and the nodes in the \nth{k} level of $X(\cS, \bmg(\alpha))$ to the \nth{k} level of $X(\cS', \bmh(\alpha))$.  Thus, it suffices to consider nodes in the same level.  Obviously, this condition holds for the root $P_1$ of $T_\cS$.  We consider all of the remaining cases in the definition of $\hphi$.  Consider the nodes $(x_i P_{ij})_{x_1, \ldots, x_{i - 1}}$ and $(y_i P_{ij})_{y_1, \ldots, y_{i - 1}}$ in $T_\cS$ with $x = x_1 \cdots x_i$, $y = y_1 \cdots y_i$ and each $x_k, y_k \in P_k$ where $i > \kappa$.  Suppose that $\hphi((x_i P_{ij})_{x_1, \ldots, x_{i - 1}}) = \hphi((y_i P_{ij})_{y_1, \ldots, y_{i - 1}})$.  Then we have $(\phi(x_i) Q_{ij})_{\phi_1(x), \ldots, \phi_{i - 1}(x)} = (\phi(y_i) Q_{ij})_{\phi_1(y), \ldots, \phi_{i - 1}(y)}$ which implies that $\phi_k(x) = \phi_k(y)$ for $1 \leq k \leq i - 1$.  Then $x_k = y_k$ for $1 \leq k \leq i - 1$ so $\phi(x_i) Q_{ij} = \phi(y_i) Q_{ij}$ and it follows that $x_i P_{ij} = y_i P_{ij}$.  Suppose $\hphi((y_i P_{ij})_{y_1, \ldots, y_{i - 1}}^{(x)}) = \hphi((z_i P_{ij})_{z_1, \ldots, z_{i - 1}}^{(w)})$ with $x, w \in G$, $y = y_1 \cdots y_i$, $z = z_1 \cdots z_i$ and each $y_k, z_k \in P_k$ where $i > \kappa$.  Then $(\phi(y_i) Q_{ij})_{\phi_1(y), \ldots, \phi_{i - 1}(y)}^{(\phi(x))} = (\phi(z_i) Q_{ij})_{\phi_1(z), \ldots, \phi_{i - 1}(z)}^{(\phi(w))}$.  Thus, $\phi(x) = \phi(w)$ so $x = w$.  Hence, $\phi_k(y) = \phi_k(z)$ for $1 \leq k \leq i - 1$ so $x_k = z_k$ for $1 \leq k \leq i - 1$.  Then $\phi(y_i) Q_{ij} = \phi(z_i) Q_{ij}$ so $y_i P_{ij} = z_i P_{ij}$.  It follows from the definition of $\hphi$ that it also acts as an injection on the leaf nodes $y_\lambda^{(x)}$.  Thus, $\hphi$ is an injection and hence a bijection.

  We now argue that $\hphi$ respects all colors.  It follows from our definition that $\hphi$ respects the colors ``left'', ``right'', ``equal'', and ``root'' on all nodes.  For a node $(x_i P_{ij})_{x_1, \ldots, x_{i - 1}}$ colored ``not identity'' where  $x = x_1 \cdots x_i$, each $x_k \in P_k$ and $i > \kappa$, we have $\hphi((x_i P_{ij})_{x_1, \ldots, x_{i - 1}}) = (\phi(x_i) Q_{ij})_{\phi_1(x), \ldots, \phi_{i - 1}(x)}$ which is also colored ``not identity''.  Similarly, for a node $(y_i P_{ij})_{y_1, \ldots, y_{i - 1}}^{(x)}$ colored ``not identity'' where $x \in G$, $y = y_1 \cdots y_i$, each $y_k \in P_k$ and $i > \kappa$, we have $\hphi((y_i P_{ij})_{y_1, \ldots, y_{i - 1}}^{(x)}) = (\phi(y_i) Q_{ij})_{\phi_1(y), \ldots, \phi_{i - 1}(y)}^{(\phi(x))}$ which is also colored ``not identity''.  Therefore, $\hphi$ also respects the ``not identity'' color.  Since $\hphi$ is a bijection, we can apply the same reasoning to $\left(\hphi\right)^{-1}$ to show that a node in $X(\cS, \bmg(\alpha))$ is colored not identity \ifft its image under $\hphi$ is colored not ``identity''.  It follows that $\hphi$ respects the ``internal'' color.

  We denote by $i_{\bmg(\alpha)}(x)$ the position of $x \in P_1 \cdots P_\kappa$ in the order on $P_1 \cdots P_\kappa$ determined by $\preceq_{\bmg(\alpha)}$.  Recall that the node $(x_\kappa)_{x_1, \ldots, x_{\kappa - 1}}$ has color $i_{\bmg(\alpha)}(x)$ where $x = x_1 \cdots x_\kappa$ and each $x_j \in P_j$.  It follows that $(x_\kappa)_{x_1, \ldots, x_{\kappa - 1}}$ and $\hphi((x_\kappa)_{x_1, \ldots, x_{\kappa - 1}}) = (\phi(x_\kappa))_{\phi_1(x), \ldots, \phi_{\kappa - 1}(x)}$ have the same color $i_{\bmg(\alpha)}(x) = i_{\bmh(\alpha)}(y)$ by \propref{gen-ord} where $y = \phi_1(x) \cdots \phi_\kappa(x) = \phi(x)$.  Similarly for $x \in G$, $(y_\kappa)_{y_1, \ldots, y_{\kappa - 1}}^{(x)}$ and $\hphi((y_\kappa)_{y_1, \ldots, y_{\kappa - 1}}^{(x)}) = (\phi(y_\kappa))_{\phi_1(y), \ldots, \phi_{\kappa - 1}(y)}^{(\phi(x))}$ have the same color $i_{\bmg(\alpha)}(y) = i_{\bmh(\alpha)}(z)$ where $y = y_1 \cdots y_\kappa$, $z = \phi_1(y) \cdots \phi_\kappa(y) = \phi(y)$ and each $y_j \in P_j$.  Hence, $\hphi$ respects all the colors.

  Our next step is to show that $\hphi$ respects the tree edges.  There are two types of tree edges: those which come from a binary tree $B(\bmu)$ constructed according to \defref{right-bin-tree} and those which come from the copies of the trees $T(S_i)$ in \defref{TS-sol-gen}.  We start by considering the edges which come from the binary trees.  For a node $(x_\kappa)_{x_1, \ldots, x_{\kappa - 1}}$ where $x = x_1 \cdots x_\kappa$ and each $x_j \in P_j$, we have $\hphi((x_\kappa)_{x_1, \ldots, x_{\kappa - 1}}) = (\phi(x_\kappa))_{\phi_1(x), \ldots, \phi_{\kappa - 1}(x)}$.  Let $\bmu$ be the ordering of the nodes $(x_\kappa)_{x_1, \ldots, x_{\kappa - 1}}$ from \defref{TS-sol-gen}; similarly, let $\bmv$ be the ordering of the nodes $(\phi(x_\kappa))_{\phi_1(x_1), \ldots, \phi(x_{\kappa - 1})}$.  As we already noted, \propref{gen-ord} implies that $\hphi$ maps the \nth{k} node in $\bmu$ to the \nth{k} node in $\bmv$.  By \defref{right-bin-tree}, the trees $B(\bmu)$ and $B(\bmv)$ depend only on the orderings $\bmu$ and $\bmv$ and not the labels $x_1 \cdots x_\kappa$ and $\phi_1(x_1) \cdots \phi_\kappa(x_\kappa)$.  Thus, there exists a unique definition of $\hphi$ on the intermediate nodes of $B(\bmu)$ which respects the edges of $B(\bmu)$ and $B(\bmv)$.  Similarly, Consider a node $(y_\kappa)_{y_1, \ldots, y_{\kappa - 1}}^{(x)}$ where $x \in G$, $y = y_1 \cdots y_\kappa$ and each $y_j \in P_j$; we have $\hphi((y_\kappa)_{y_1, \ldots, y_{\kappa - 1}}^{(x)}) = (\phi(y_\kappa))_{\phi_1(y), \ldots, \phi_{\kappa - 1}(y)}^{(\phi(x))}$.  We fix $x \in G$ and let $\bmu$ be the ordering of the nodes $(y_\kappa)_{y_1, \ldots, y_{\kappa - 1}}^{(x)}$ for all $y_j \in P_j$ from \defref{TS-sol-gen}; similarly, let $\bmv$ be the ordering of the nodes $(\phi(y_\kappa))_{\phi_1(y_1), \ldots, \phi(y_{\kappa - 1})}^{(\phi(x))}$ for all $y_j \in P_j$.  \propref{gen-ord} again implies that $\hphi$ maps the \nth{k} node in $\bmu$ to the \nth{k} node in $\bmv$.  It follows that there exists a unique definition of $\hphi$ on the intermediate nodes of each $B(\bmu)$ which respects the edges of each $B(\bmu)$ and $B(\bmv)$.

  Next, we consider the edges which come from the trees $T_{x_1, \ldots, x_{i - 1}}$ in \defref{TS-sol-gen} where each $x_k \in P_k$ and $i > \kappa$.  Consider an edge between the nodes $(y_i P_{ij})_{y_1, \ldots, y_{i - 1}}$ and $(x_i P_{i, j + 1})_{x_1, \ldots, x_{i - 1}}$ where $x = x_1 \cdots x_i$, $y = y_1 \cdots y_i$, each $x_k, y_k \in P_k$, $x_k = y_k$ for $k < i$ and $i > \kappa$.  Then $y_i P_{ij} \subseteq x_i P_{i, j + 1}$ so $\phi[y_i P_{ij}] \subseteq \phi[x_i P_{i, j + 1}]$.  Thus, there is an edge between $\hphi((y_i P_{ij})_{x_1, \ldots, x_{i - 1}}) = (\phi(y_i) Q_{ij})_{\phi_1(x), \ldots, \phi_{i - 1}(x)}$ and $\hphi((x_i P_{i, j + 1})_{x_1, \ldots, x_{i - 1}}) = (\phi(x_i) Q_{i, j + 1})_{\phi_1(x), \ldots, \phi_{i - 1}(x)}$.  Consider an edge between the nodes $(z_i P_{ij})_{z_1, \ldots, z_{i - 1}}^{(x)}$ and $(y_i P_{i, j + 1})_{y_1, \ldots, y_{i - 1}}^{(x)}$ where $x \in G$, $z = z_1 \cdots z_i$, $y = y_1 \cdots y_i$, each $z_k, y_k \in P_k$, $z_k = y_k$ for $k < i$ and $i > \kappa$; we again have $z_i P_{ij} \subseteq y_i P_{i, j + 1}$ and so the same analysis shows that there is a tree edge between $\hphi((z_i P_{ij})_{y_1, \ldots, y_{i - 1}}^{(x)}) = (\phi(z_i) Q_{ij})_{\phi_1(y), \ldots, \phi_{i - 1}(y)}^{(\phi(x))}$ and $\hphi((y_i P_{i, j + 1})_{y_1, \ldots, y_{i - 1}}^{(x)}) = (\phi(y_i) Q_{i, j + 1})_{\phi_1(y), \ldots, \phi_{i - 1}(y)}^{(\phi(x))}$.

  We now consider the cross edges.  Suppose that $x y = z$ where $x, y, z \in G$; we will show that the edges from $y_\la^{(x)}$ to $x_\ra^{(y)}$ and from $x_\ra^{(y)}$ to $y_=^{(z)}$ are preserved by $\hphi$.  By definition, we have

  \begin{align*}
    \hphi(y_\la^{(x)}) & = \phi(y)_\la^{(\phi(x))} \\
    \hphi(x_\ra^{(y)}) & = \phi(x)_\ra^{(\phi(y))} \\
    \hphi(y_=^{(z)})   & = \phi(y)_=^{(\phi(z))}
  \end{align*}
  Since $\phi : G \ra H$ is an isomorphism, $\phi(x) \phi(y) = \phi(z)$ so there are cross edges from $\phi(y)_\la^{(\phi(x))}$ to $\phi(x)_\ra^{(\phi(y))}$ and from $\phi(x)_\ra^{(\phi(y))}$ to $\phi(y)_=^{(\phi(z))}$.  Since $X(\cS, \bmg(\alpha))$ and $X(\cS', \bmh(\alpha))$ have the same number of edges, it follows that there is an edge between two nodes in $X(\cS, \bmg(\alpha))$ \ifft there is an edge between their images under $\hphi$ in $X(\cS', \bmh(\alpha))$.  Thus, $\hphi$ is an isomorphism so we have shown that $f : \phi \mapsto \hphi$ is a map $f : \iso_{\bmg(\alpha) \mapsto \bmh(\alpha)}(\cS, \cS') \ra \iso(X(\cS, \bmg(\alpha)), X(\cS', \bmh(\alpha)))$.

  It remains to show that $f$ is bijective.  Suppose that $\phi_1, \phi_2 \in \iso_{\bmg(\alpha) \mapsto \bmh(\alpha)}(S, S')$ and $\hphi_1 = \hphi_2$.  Then in particular we have $\hphi_1(x) = \hphi_2(x)$ for $x \in G$ so $\phi_1(x) = \phi_2(x)$ and $\phi_1 = \phi_2$.  Therefore, $f$ is injective.  Our final task is to show that it is surjective.  Let $\theta \in \iso(X(\cS, \bmg(\alpha)), X(\cS', \bmh(\alpha)))$.  Because the roots $P_1$ and $Q_1$ of $X(\cS, \bmg(\alpha))$ and $X(\cS', \bmh(\alpha))$ are the only nodes colored ``root'', we have $\theta(P_1) = Q_1$ and $\theta$ maps the \nth{k} level of $X(\cS, \bmg(\alpha))$ to the \nth{k} level of $X(\cS', \bmh(\alpha))$.  In particular, a restriction of $\theta$ is a bijection from the nodes $x \in G$ of $T_\cS$ to the nodes $y \in H$ of $T_{\cS'}$.  Then the map $\phi : G \ra H : x \mapsto \theta(x)$ where $x \in G$ is a node in $T_\cS$ is a well-defined bijection.

  We now claim that for any node $y^{(x)}$ in $T_\cS^{(x)}$, $\theta(y^{(x)}) = \phi(y)^{(\phi(x))}$.  We know that for $x \in G$, $T_\cS^{(x)}$ is rooted at the node $x$ of $T_\cS$ and $\theta(x) = \phi(x)$ by definition of $\phi$.  It follows that the node $\theta(y^{(x)})$ is in $T_{\cS'}^{(\phi(x))}$; similarly $\theta(x^{(y)})$ is in $T_{\cS'}^{(\phi(y))}$.  Hence, we can write $\theta(y^{(x)}) = \phi(b)^{(\phi(x))}$ and $\theta(x^{(y)}) = \phi(a)^{(\phi(y))}$ for some $a, b \in G$.  We know that in $X(\cS', \bmh(\alpha))$, the unique shortest path from $\phi(x)$ to $\phi(y)$ that passes through a node colored ``left'' and later a node colored ``right'' follows the tree $T_{\cS'}^{(\phi(x))}$ to $\phi(y)^{(\phi(x))}$, then the path $(\phi(y)^{(\phi(x))}, \phi(y)_\la^{(\phi(x))}, \phi(x)_\ra^{(\phi(y))}, \phi(x)^{(\phi(y))})$ and then the tree $T_{\cS'}^{(\phi(y))}$ to $\phi(y)$.  On the other hand, we know that the path $(\phi(b)^{(\phi(x))}, \phi(b)_\la^{(\phi(x))}, \phi(a)_\ra^{(\phi(y))}, \phi(a)^{(\phi(y))})$ exists in $X(\cS', \bmh(\alpha))$ by definition of $a$ and $b$.  Then a shortest path that passes through a node colored ``left'' and later a node colored ``right'' from $\phi(x)$ to $\phi(y)$ follows the tree $T_{\cS'}^{(\phi(x))}$ to $\phi(b)^{(\phi(x))}$, then the path $(\phi(b)^{(\phi(x))}, \phi(b)_\la^{(\phi(x))}, \phi(a)_\ra^{(\phi(y))}, \phi(a)^{(\phi(y))})$ and then the tree $T_{\cS'}^{(\phi(y))}$ to $\phi(y)$.  From uniqueness, it follows that $\phi(x) = \phi(a)$ and $\phi(y) = \phi(b)$ so $x = a$ and $y = b$.  Thus, for any node $y^{(x)}$ in $T_\cS^{(x)}$, $\theta(y^{(x)}) = \phi(y)^{(\phi(x))}$.

  Now, we show that $\phi$ is an isomorphism.  We have already argued that it is bijective.  Let $x, y \in G$ and set $z = x y$.  Then $\theta$ maps the path $(y_\la^{(x)}, x_\ra^{(y)}, y_=^{(z)})$ in $X(\cS, \bmg(\alpha))$ to the path $(\phi(y)_\la^{(\phi(x))}, \phi(x)_\ra^{(\phi(y))}, \phi(y)_=^{(\phi(z))})$ in $X(\cS', \bmh(\alpha))$.  Since $\theta$ preserves colors, the colors of the nodes in this path are $(\text{``left''}, \text{``right''}, \text{``equal''})$.  We conclude that $\phi(x) \phi(y) = \phi(z)$.  By definition of $z$, it follows that $\phi$ is an isomorphism from $G$ to $H$.  

  We now assert that $\phi[P_{ij}] = Q_{ij}$ for all $i$ and $j$.  We know that $\theta$ maps the node $(x_\kappa)_{x_1, \ldots, x_{\kappa - 1}}$ which is colored $i_{\bmg(\alpha)}(x)$ where $x = x_1 \cdots x_\kappa$ and each $x_j \in P_j$ to a node $(y_\kappa)_{y_1, \ldots, y_{\kappa - 1}}$ which is colored $i_{\bmh(\alpha)}(y) = i_{\bmg(\alpha)}(x)$ where $y = y_1 \cdots y_\kappa$ and each $y_j \in Q_j$.  By assumption, the mapping from the generators $\bmg(\alpha)$ to $\bmh(\alpha)$ extends to an isomorphism from $\psi : P_1 \cdots P_\kappa \ra Q_1 \cdots Q_\kappa$ (which must be unique).  By \propref{gen-ord}, we see that $i_{\bmg(\alpha)}(x) = i_{\bmh(\alpha)}(\psi(x))$.  It follows that $\theta(x) = \psi(x)$ which implies that $\phi(g_{ij}) = h_{ij}$ and $\phi[P_{ij}] = Q_{ij}$ for $i \leq \kappa$ and all $j$.

  Now  suppose $i > \kappa$ and fix $j$.  Because $\phi$ is an isomorphism, $\phi(1) = 1$ so $\theta(1) = 1$.  It follows that $\theta$ maps the unique shortest path from $P_1$ to $1$ in $T_\cS$ to the unique shortest path from $Q_1$ to $1$ in $T_{\cS'}$.  Then $\theta((P_{ij})_{x_1, \ldots, x_{i - 1}}) = (Q_{ij})_{y_1, \ldots, y_{i - 1}}$ where each $x_k = 1$ and each $y_k = 1$ for $k < i$.  Consider the paths of the form

  \begin{align}
    \label{eq:Pij-path}
    ((P_{ij})_{x_1, \ldots, x_{i - 1}}, (x_i P_{i, j - 1})_{x_1, \ldots, x_{i - 1}}, \ldots, (x_i P_{i0})_{x_1, \ldots, x_{i - 1}}, (P_{i + 1, m_{i + 1} - 1})_{x_1, \ldots, x_i}, \ldots, (P_{\ell 0})_{x_1, \ldots, x_{\ell - 1}})
  \end{align}
  where $x_i \not= 1 \in P_{ij}$ and $x_k = 1$ for all $k \not= i$.  We note that this is the unique shortest path to the node $x = x_i \in P_{ij}$.  Moreover, at least one of the nodes on this path is colored ``not identity''.  Since none of the nodes on the unique shortest path from $(P_{ij})_{x_1, \ldots, x_{i - 1}}$ to $1$ are colored ``not identity'', we see that the paths of the form of \eqref{Pij-path} are exactly the unique shortest paths from $(P_{ij})_{x_1, \ldots, x_{i - 1}}$ to the non-identity nodes $x \in P_{ij}$ of $T_\cS$.  We can also think of the paths of the form of \eqref{Pij-path} as the paths from $(P_{ij})_{x_1, \ldots, x_{i - 1}}$ to nodes $z \in G$ which pass through nodes labelled ``internal'' until reaching a node $(x_i P_{ik})_{x_1, \ldots, x_{i - 1}}$ which is colored ``not identity'' where $1 \leq k < j$ such that after this point, they only pass through nodes which are colored ``not identity''.  Applying the same argument, we find that the paths of the form

  \begin{align}
    \label{eq:Qij-path}
    ((Q_{ij})_{y_1, \ldots, y_{i - 1}}, (y_i Q_{i, j - 1})_{y_1, \ldots, y_{i - 1}}, \ldots, (y_i Q_{i0})_{y_1, \ldots, y_{i - 1}}, (Q_{i + 1, m_{i + 1} - 1})_{y_1, \ldots, y_i}, \ldots, (Q_{\ell 0})_{y_1, \ldots, y_{\ell - 1}})
  \end{align}
  where $y_i \not= 1 \in Q_{ij}$ and $y_k = 1$ for all $k \not= i$ are the paths from $(Q_{ij})_{y_1, \ldots, y_{i - 1}}$ to nodes $w \in H$ which pass through nodes labelled ``internal'' until reaching a node $(y_i Q_{ik})_{y_1, \ldots, y_{i - 1}}$ which is colored ``not identity'' where $1 \leq k < j$ such that after this point, they only pass through nodes which are colored ``not identity''.
  
  We have already established that $\theta((P_{ij})_{x_1, \ldots, x_{i - 1}}) = (Q_{ij})_{y_1, \ldots, y_{i - 1}}$ so it follows that $\theta$ maps each path of the form of \eqref{Pij-path} to a path of the form of \eqref{Qij-path}.  Since we already know that $\theta(1) = 1$, this implies that $\theta$ maps the nodes in $P_{ij}$ to the nodes in $Q_{ij}$; it follows that $\phi[P_{ij}] = Q_{ij}$.

  Therefore, $\phi$ is an isomorphism from $\cS$ to $\cS'$ which respects $\bmg(\alpha)$ and $\bmh(\alpha)$.  Then $\theta = \hphi$.  Hence, we have proven that $f$ is a bijection.
\end{proof}

\solhcompcan*

\begin{proof}
  We transform the Hall composition series $(G, \cP, \cS)$ into a Hall composition series with fixed generators.  Let us define $\alpha = \log n / \log \log n$.  We reindex so that $p_i > p_j$ for $i < j$.  Recall that our construction of the graph $X(\cS, \bmg(\alpha))$ in \defref{TS-sol-gen-cone} is only defined when $p_1 \geq \alpha$.  The reason for this is because it simplifies the correctness proof and definitions (which are already quite complex) and is not due to any fundamental difficulty.  However, a side effect is that it makes the present proof slightly more difficult.  We first assume that $p_1 \geq \alpha$ and then show a simple trick which reduces the general case to the case where $p_1 \geq \alpha$.  As before, we start by showing a complete polynomial-size invariant $\invar(\cS)$ and then use it to obtain a canonical form.

  Suppose $p_1 \geq \alpha$ and let $\kappa$ be the largest value of $i$ such that $p_i \geq \alpha$.  We compute the set $\cG_\alpha = \setb{\bmg(\alpha) = (g_{1, 1}, \ldots, g_{1, m_1}, \ldots, g_{\kappa, 1}, \ldots, g_{\kappa, m_\kappa})}{\langle g_{i, j + 1} P_{ij} \rangle = P_{i, j + 1} / P_{ij} \text{ for each } i \leq \kappa}$ of all possible choices of representatives for the factor groups $P_{i + 1, j} / P_{ij}$ where $i \leq \kappa$.  For each $\bmg(\alpha) \in \cG_\alpha$, we construct the Hall composition series with fixed generators $(G, \cP, \cS, \bmg(\alpha), \kappa, \alpha)$ and compute the graph $X(\cS, \bmg(\alpha))$.  We then compute the canonical form of each $X(\cS, \bmg(\alpha))$ and store the results in a list $\bmA$.  We sort $\bmA$ lexicographically and choose the lexicographically smallest element of $\bmA$ to be $\invar(\cS)$ of $(G, \cP, \cS)$.

  We already know from \lemref{sol-graph} that $\abs{\invar(\cS)} = \poly(n)$.  We now argue that if $(H, \cQ, \cS')$ is a Hall composition series then $\cS \cong \cS'$ \ifft $\invar(\cS) = \invar(\cS')$.  First, if $\invar(\cS) = \invar(\cS')$ then there exist $\bmg(\alpha) \in \cG_\alpha$ and $\bmh(\alpha) \in \cH_\alpha$ such that $(\cS, \bmg(\alpha)) \cong (\cS', \bmh(\alpha))$ by \corref{sol-red-cor} which implies that $\cS \cong \cS'$.  Now suppose that $\cS \cong \cS'$.  Then for each $\bmg(\alpha) \in \cG_\alpha$, there exists $\bmh(\alpha) \in \cH_\alpha$ such that $(\cS, \bmg(\alpha))$ is isomorphic to $(\cS', \bmh(\alpha))$.  Similarly, for each $\bmh(\alpha) \in \cH_\alpha$, there exists $\bmg(\alpha) \in \cG_\alpha$ such that $(\cS', \bmh(\alpha))$ is isomorphic to $(\cS, \bmg(\alpha))$.  Let $\bmA$ be the sorted list of the canonical forms of the graphs $X(\cS, \bmg(\alpha))$ where each $\bmg(\alpha) \in \cG_\alpha$; similarly, let $\bmB$ be the sorted list of the canonical forms of the graphs $X(\cS', \bmh(\alpha))$ where each $\bmh(\alpha) \in \cH_\alpha$.  Then $\bmA = \bmB$ by the above argument so $\invar(\cS) = \invar(\cS')$.  Thus, $\invar(\cS)$ is a complete polynomial-size invariant.

  We claim that the above procedure runs in $n^{O(\log n / \log \log n)}$ time.  Note that $\log_\alpha n = O(\log n / \log \log n)$.  Because $p_i \geq \alpha$ for $i \leq \kappa$, $\abs{\cG_\alpha} \leq n^{O(\log n / \log \log n)}$.  Each graph $X(\cS, \bmg(\alpha))$ has at most $O(n^2)$ nodes and degree at most $\max\{\alpha, 4\}$ by \lemref{sol-graph}.  Thus we can compute the canonical form of each $X(\cS, \bmg(\alpha))$ in $n^{O(\alpha)} = n^{O(\log n / \log \log n)}$ time.  Computing the canonical forms for all $n^{O(\log n / \log \log n)}$ graphs takes a total of $n^{O(\log n / \log \log n)}$ time.  Sorting the $n^{O(\log n / \log \log n)}$ canonical forms of the graphs also takes $n^{O(\log n / \log \log n)}$ time so the overall runtime is $n^{O(\log n / \log \log n)}$.

  To complete our argument for the case where $p_1 \geq \alpha$, we show how to compute the canonical form $\can(\cS)$ from $\invar(\cS)$ in polynomial time.  Recall that $\invar(\cS) = \can(X(\cS, \bmg(\alpha)))$ for some $\bmg(\alpha) \in \cG_\alpha$ so there exists an isomorphism $\theta :  X(\cS, \bmg(\alpha)) \ra \can(X(\cS, \bmg(\alpha)))$.  First, we find the root node $\theta(P_1)$ of $\can(X(\cS, \bmg(\alpha)))$ which is easy since it is the only node colored ``root.''  The nodes of the form $\theta(x)$ where $x \in G$ are simply those nodes which are some distance $d$ from the root.  The distance $d$ is equal to the sum of the height of the binary tree $T$ which represents the elements of the group $P_1 \cdots P_\kappa$ and the heights of the trees $T(S_i)$ for $\kappa < i \leq \ell$.  The height of $T$ is $\lceil \log \abs{P_1 \cdots P_\kappa} \rceil = \lceil \log \prod_{j = 1}^\kappa p_j^{e_j} \rceil$; the height of $T(S_i)$ is simply the length of $S_i$ which is $e_i$ since $S_i$ is a composition series for the $p_i$-group $P_i$.  Thus, $d$ depends only on the order $n$ of $G$.  Then we can find the elements $\theta(x)$ for $x \in G$ of $\can(X(\cS, \bmg(\alpha)))$ by breadth-first search.  For each such $x \in G$, we then define $\lambda_G(\theta(x))$ to be the index of $\theta(x)$ in the sequence of nodes $\theta(x)$ where $x \in G$ obtained in the breadth-first search.  Thus, $\lambda_G : \theta[G] \ra [n]$ is a bijection.

  Let $x, y \in G$; we will now show how to compute $\theta(xy)$ from $\can(X(\cS, \bmg(\alpha)))$ given $\theta(x)$ and $\theta(y)$.  Note that in $X(\cS, \bmg(\alpha))$, there is a unique shortest path from $x$ to $y$ that passes through a node colored ``left'' and later a node colored ``right''; namely, this is the path which follows the tree $T_\cS^{(x)}$ from $x$ to $y^{(x)}$, $(y^{(x)}, y_\la^{(x)}, x_\ra^{(y)}, x^{(y)})$ and then the tree $T_\cS^{(y)}$ to $y$.  Thus, given nodes $\theta(x)$ and $\theta(y)$ in $\can(X(\cS, \bmg(\alpha)))$, we can find the nodes $\theta(y^{(x)})$ and $\theta(x^{(y)})$ in $\can(X(\cS, \bmg(\alpha)))$.  Once this has been done for all $x, y \in G$, the multiplication table $M'(G)$ defined by the rule $\theta(x)\theta(y) = \theta(xy)$ can be determined by inspecting the multiplication gadgets in $\can(X(\cS, \bmg(\alpha)))$.

  Using the multiplication table just computed, we can find the node $\theta(1)$ in $\can(X(S))$ which corresponds to the identity.  Now since $\theta$ is an isomorphism, it maps unique shortest path from the root $P_1$ to $1$ in $X(\cS, \bmg(\alpha))$ to the unique shortest path from the root $\theta(P_1)$ to $\theta(1)$ in $\can(X(\cS, \bmg(\alpha)))$.  For $i > \kappa$, let $x_k = 1$ for $k < i$; the node $\theta((P_{ij})_{x_1, \ldots, x_{i - 1}})$ in $\can(X(\cS, \bmg(\alpha)))$ is then the node at a distance of $\lceil \log \prod_{j = 1}^\kappa p_j^{e_j} \rceil + \sum_{j = \kappa + 1}^{i - 1} e_j + m_i - j$ from the root $\theta(P_1)$ along the unique shortest path to $\theta(1)$ in $\can(X(\cS, \bmg(\alpha)))$.  Thus, we can find the node $\theta((P_{ij})_{x_1, \ldots, x_{i - 1}})$ in $\can(X(\cS, \bmg(\alpha)))$ which corresponds to each $P_{ij}$ where $i > \kappa$.  Moreover, for each $i > \kappa$, the nontrivial elements $x \in P_{ij}$ in $X(\cS, \bmg(\alpha))$ are precisely the nodes $y \in G$ which can be reached from the node $(P_{ij})_{x_1, \ldots, x_{i - 1}}$ where each $x_k = 1$ for $k < i$ and $x_i \in P_{ij}$ by moving away from the root through nodes colored ``internal'' until reaching a node $(x_i P_{ik})_{x_1, \ldots, x_{i - 1}}$ which is colored ``not identity'' and then through nodes colored ``not identity'' until reaching $x_i \in G$.  Thus, we can find these nodes by breadth-first search on $\can(X(\cS, \bmg(\alpha)))$.

  It remains to show how to find the nodes of the form $\theta(x)$ in $\can(X(\cS, \bmg(\alpha)))$ for $x \in P_{ij}$ where $i \leq \kappa$.  Let $\bmu$ be the vector of elements of $P_1 \cdots P_\kappa$ ordered according to $\preceq_{\bmg(\alpha)}$ as defined in \defref{TS-sol-gen}.  Recall from \defref{TS-sol-gen} that the \nth{k} element of $\bmu$ was colored $k$ in the binary tree $T$.  From \propref{gen-ord}, we know that the first elements in the order $\preceq_{\bmg(\alpha)}$ are the identity followed by the elements of $\bmg(\alpha)$.  This allows us to locate the elements of the form $\theta((x_\kappa)_{x_1, \ldots, x_{\kappa - 1}})$ in $\can(X(\cS, \bmg(\alpha)))$ where each $x_j = 1$ for $j \not= k$ and $x_k \in P_k$ is an element of the vector $\bmg(\alpha)$.

  Now we find the nodes of the form $\theta(x)$ in $\can(X(\cS, \bmg(\alpha)))$ where $x$ is an element of the vector $\bmg(\alpha)$.  We can write $x$ as the product $x_1 \cdots x_\ell$ where $x_j = 1$ for $j \not= k$ and $x_k \in P_k$ is an element of the vector $\bmg(\alpha)$.  Thus, the node for $x$ can also be written as $(x_\ell)_{x_1, \ldots, x_{\ell - 1}}$.  Note that by \defref{TS-sol-gen}, $(x_\kappa)_{x_1, \ldots, x_{\kappa - 1}}$ is colored ``not identity.''  Moreover, there is a unique path from the node $(x_\kappa)_{x_1, \ldots, x_{\kappa - 1}}$ which moves away from the root and passes through nodes which are colored ``not identity''; since this path is from $(x_\kappa)_{x_1, \ldots, x_{\kappa - 1}}$ to $(x_\ell)_{x_1, \ldots, x_{\ell - 1}}$, we can find each node of the form $\theta(x)$ in $\can(X(\cS, \bmg(\alpha)))$ where $x$ is an element of $\bmg(\alpha)$.

  Recall that $\bmg(\alpha) = (g_{1, 1}, \ldots, g_{1, m_1}, \ldots, g_{\kappa, 1}, \ldots, g_{\kappa, m_\kappa})$ where $\langle g_{i, j + 1} P_{ij} \rangle = P_{i, j + 1} / P_{ij}$.  Since $\bmu$ starts with $1$ followed by $\bmg(\alpha)$, we can find the node $\theta(g_{i, j + 1})$ which corresponds to each $g_{i, j + 1}$ for $i \leq \kappa$.  Using the multiplication table $M'(G)$ which we computed previously, we can then determine the nodes of the form $\theta(x)$ in $\can(X(\cS, \bmg(\alpha)))$ where $x \in P_{i, j + 1}$ for $i \leq \kappa$.

At this point, for each $P_{ij}$ we know the nodes of the form $\theta(x)$ where $x \in P_{ij}$.  Since $P_i = P_{i, m_i}$, we have all the information that is required to construct $\can(\cS)$.  Define $\phi = \restr{\theta}{G}$ and $\psi_G = \lambda_G \phi$.  Let $M(G)$ be the $n \times n$ matrix with elements in $[n]$ defined by $M(G)_{\psi_G(x), \psi_G(y)} = \psi_G(xy)$ for all $x, y \in G$.  Thus, we can compute $\psi[P_i]$ and $\psi_G[P_{ij}]$ for each $i$ and $j$.  We then set $\can(\cS) = (M(G), \psi_G[P_1], \ldots, \psi_G[P_\ell], \psi_G[P_{10}], \ldots, \psi_G[P_{1, m_i}], \ldots, \psi_G[P_{\ell 0}], \ldots, \psi_G[P_{\ell, m_\ell}])$ as in \defref{hcomp-can}.

  We claim that $\can(\cS)$ is a canonical form for $S$.  Let $(H, \cQ, \cS')$ be a Hall composition series.  By construction, $\psi_G : G \ra [n]$ is an isomorphism from $G$ to the group described by the multiplication table $M(G)$.  If $\can(\cS) = \can(\cS')$ then $\psi_H^{-1} \psi_G$ is an isomorphism from $\cS$ to $\cS'$ so $\cS \cong \cS'$.  Suppose that $\cS \cong \cS'$; then $\invar(\cS) = \invar(\cS')$.  Since we computed $\can(\cS)$ deterministically and this computation depends only on $\invar(\cS)$, $n$ and the composition length of each $P_i$ and , it follows that $\can(\cS) = \can(\cS')$.

  We have shown how to construct a canonical form when $p_1 \geq \alpha$.  Now suppose that $p_1 < \alpha$ after reindexing so that $p_i > p_j$ for $i < j$.  We choose a prime $2 \alpha \leq q \leq 4 \alpha$ which exists for sufficiently large $n$ by Bertrand's Postulate.  We then construct a new group $\bar G = \bbZ_q \times G$ and consider the Hall composition series $(\bar G, \bar \cP, \bar \cS)$ where $\bar \cP = \cP \cup \{\bbZ_q\}$ and $\bar \cS = \cS \cup \{1 \tril \bbZ_q\}$.  We note that $\abs{\bar G} = q n$ and $q \geq 2 \alpha$.  Moreover,

  \begin{align*}
       \beta &    = \log(q n) / \log \log(q n) 
    \inlongaln    = \frac{\log q + \log n}{\log(\log q + \log n)}
    \inlongaln \leq \frac{2 \log n}{\log \log n}
    \inlongaln    = 2 \alpha
  \end{align*}
  so $q \geq \beta$.  Thus, we can compute a canonical form $\can(\bar \cS)$ for the Hall composition series $(\bar G, \bar P, \bar \cS)$ in time

  \begin{align*}
   (q n)^{O(\log(q n) / \log \log(q n))} & = (q n)^{O(\log n / \log \log n)}
                                \inlongaln = n^{O(\log n / \log \log n)} \text{ since $q \leq n$}
  \end{align*}

  For some $\bar \bmg(\alpha) \in \bar \cG_\beta$, we have $\can(\bar \cS) = \can(X(\bar \cS, \bar \bmg(\beta)))$.  Let $\theta : X(\bar \cS, \bar \bmg(\beta)) \ra \can(X(\bar \cS, \bar \bmg(\beta)))$ be the isomorphism which was used to define $\lambda_{\bar G}$.  Let $I = \setb{\lambda_{\bar G}(\theta(x))}{x \in G} \subseteq [nq]$ and define $\rho_G : I \ra [n] : \lambda_{\bar G}(\theta(x)) \mapsto k$ where $\lambda_{\bar G}(\theta(x))$ is the \nth{k} smallest element of $I$.  Let $\psi_G = \rho_G \lambda_{\bar G} \theta$.  By \defref{hcomp-can}, $\can(\bar \cS)$ contains $\psi_G[P_i]$ and $\psi_G[P_{ij}]$ for each $i$ and $j$.  Moreover, we can extract a multiplication table $M(G)$ from $M(\bar G)$ by deleting all rows and columns except those which correspond to elements of $\psi_G[G]$ (which we can compute deterministically in the same manner as before).  Set $\can(\cS) = (M(G), \psi_G[P_1], \ldots, \psi_G[P_\ell], \psi_G[P_{10}], \ldots, \psi_G[P_{1, m_i}], \ldots, \psi_G[P_{\ell 0}], \ldots, \psi_G[P_{\ell, m_\ell}])$ as in \defref{hcomp-can}.

  Finally, we prove the correctness of our reduction from the general case.  Let $(H, \cQ, \cS')$ be a Hall composition series.  By construction, $\psi_G : G \ra [n]$ is an isomorphism from $G$ to the group described by the multiplication table $M(G)$.  If $\can(\cS) = \can(\cS')$ then $\psi_H^{-1} \psi_G$ is an isomorphism from $\cS$ to $\cS'$ so $\cS \cong \cS'$.  Suppose that $\cS \cong \cS'$.  This implies that $\bar \cS \cong \bar \cS'$ so that $\can(\bar \cS) = \can(\bar \cS')$.  Since we computed $\can(\cS)$ deterministically and this computation depends only on $\can(\bar \cS)$, $n$, $q$ and the composition length of each $P_i$, it follows that $\can(\cS) = \can(\cS')$.
\end{proof}

  \newpage
  \section{Variants of the generator-enumeration algorithm}
  \label{app:gen-rq-can}
  In this section, we derive more efficient randomized and quantum versions of the generator-enumeration algorithm using collision detection.  We also show how to modify the generator-enumeration algorithm for the group canonization problem.  These algorithms are based on sorting the rows and columns of the multiplication table~\cite{miller1978a} according to the ordering from \propref{gen-ord}.

\begin{definition}
  \label{defn:gen-mull-table}
  Let $G$ be a group and let $\bmg$ be an ordered generating set for $G$.  There is a unique order $\preceq_{\bmg}$ on the elements of $G$ by \propref{gen-ord}.  We relabel each element of $G$ by its position in the ordering $\preceq_{\bmg}$; the elements of $G$ are then labelled by $[n]$.  We reorder the rows and columns of the multiplication table of $G$ so that the elements for the rows and columns appear in the order $1, \ldots, n$ and denote the result by $M_{\bmg}$.
\end{definition}

\begin{lemma}
  \label{lem:gen-mull-table-perm}
  Let $G$ and $H$ be groups.  Let $\cG_\ell$ and $\cH_\ell$ be the collections of all ordered generating sets of $G$ and $H$ of size at most $\ell$ and define $M_\ell(G) = \setb{M_{\bmg}}{\bmg \in \cG_\ell}$.

  \begin{enumerate}
  \item For each $\bmg \in \cG$, the group defined by the multiplication table $M_{\bmg}$ is isomorphic to $G$.
  \item If $G \not\cong H$ then $M_\ell(G) \cap M_\ell(H) = \emptyset$.
  \item $M_\ell(G) = M_\ell(H)$ \ifft $G \cong H$.
  \item If $\phi : G \ra H$ is an isomorphism, then the map $\Phi : M_\ell(G) \ra M_\ell(H) : M_{\bmg} \mapsto M_{\bmh}$ (where $\bmh = \phi(\bmg)$ is the vector obtained by applying $\phi$ to each element of $\bmg$) is the identity.
  \end{enumerate}
\end{lemma}

\begin{proof}
  Observe that for any $\bmg \in \cG_\ell$, the group defined by the multiplication table $M_{\bmg}$ is simply the multiplication table of $G$ up to the relabeling performed in \defref{gen-mull-table}.  Thus, this group is isomorphic to $G$ which proves (a).  For part (b), if $G \not\cong H$ but $M \in M_\ell(G) \cap M_\ell(H)$ then $G$ would be isomorphic to the group defined by the multiplication table $M$ which is also isomorphic to $H$.  For part (c), we already know from (b) that if $G \not\cong H$ then $M_\ell(G) \cap M_\ell(H) = \emptyset$ so that $M_\ell(G) \not= M_\ell(H)$.  For the converse, we argue as follows.

  Fix an isomorphism $\phi : G \ra H$ and define $\Phi$ as in the statement of the lemma.  We claim that for every ordered generating set $\bmg$ of $G$, we have $M_{\bmg} = M_{\bmh}$ where $\bmh = \phi(\bmg)$.  We know from \propref{gen-ord} that for $x, y \in G$, $x \preceq_{\bmg} y$ \ifft $\phi(x) \preceq_{\bmh} \phi(y)$.  Since $\phi(x) \phi(y) = \phi(x y)$ as $\phi$ is an isomorphism, it follows that $M_{\bmg} = M_{\bmh}$ so $\Phi(M_{\bmg}) = M_{\bmg}$.  Thus, $\Phi$ is the identity and therefore a well-defined injection.  It remains only to show that $M_\ell(G)$ is not a proper subset of $M_\ell(H)$.  Consider the map $\Phi' : M_\ell(H) \ra M_\ell(G) : M_{\bmh} \mapsto M_{\bmg}$ where $\bmg = \phi^{-1}(\bmh)$.  The same argument implies that $\Phi'$ is an injection which shows that $\abs{M_\ell(H)} \leq \abs{M_\ell(G)}$.  It follows that $M_\ell(G) = M_\ell(H)$ which proves (c) and (d).
\end{proof}

\begin{theorem}
  \label{thm:gen-drq}
  Let $G$ and $H$ be groups.  Then we can test if $G \cong H$ in 
  
  \begin{enumerate}
  \item $n^{\log_p n + O(1)}$ time for deterministic algorithms
  \item $n^{(1 / 2) \log_p n + O(1)}$ time for randomized algorithms
  \item $n^{(1 / 3) \log_p n + O(1)}$ time for quantum algorithms
  \end{enumerate}
  where $p$ is the smallest prime dividing the order of the group.
\end{theorem}

\begin{proof}
  Part (a) is immediate from the generator-enumeration algorithm.  For part (b), we modify the generator-enumeration algorithm to use \defref{gen-mull-table} so that our collision argument applies.

  Let $\cG_\ell$ and $\cH_\ell$ be the collections of all ordered generating sets of $G$ and $H$ of size at most $\ell = \log_p n$ and define $M_\ell(G) = \setb{M_{\bmg}}{\bmg \in \cG_\ell}$.  We already know from part (c) of \lemref{gen-mull-table-perm} that $M_\ell(G) = M_\ell(H)$ \ifft $G \cong H$.  We guess subsets $A$ and $B$ of $\cG_\ell$ and $\cH_\ell$ of size $n^{(1 / 2) \log_p n + O(1)}$.  By a collision detection argument, if $G \cong H$ then there exist $\bmg \in A$ and $\bmh \in B$ such that $M_{\bmg} = M_{\bmh}$ with high probability.  

We compute $M_{\bmg}$ and $M_{\bmh}$ for all $\bmg \in A$ and $\bmh \in B$ and store the results in the lists $\bmA$ and $\bmB$ respectively.  Since each $M_{\bmg}$ and $M_{\bmh}$ is an $n \times n$ matrix with elements in $[n]$, we can sort these lists lexicographically.  This takes only $n^{(1 / 2) \log_p n + O(1)} ((1 / 2) \log_p n +O(1)) \log n = n^{(1 / 2) \log_p n + O(1)}$ time.  Once $\bmA$ and $\bmB$ have been sorted, we can test if there exist $\bmg \in A$ and $\bmh \in B$ such that $M_{\bmg} = M_{\bmh}$ by merging $\bmA$ and $\bmB$ and checking for duplicates.  This yields our $n^{(1 / 2) \log_p n + O(1)}$ randomized generator-enumeration algorithm which proves part (b).  For part (c), we use quantum claw detection~\cite{bassard1997b} instead of guessing the subsets $A$ and $B$ uniformly at random.
\end{proof}

The argument of \thmref{gen-drq} can also be used to obtain an $n^{\log_p n + O(1)}$ algorithm for group canonization.  First, we define the group canonization problem precisely.

\begin{definition}
  \label{defn:group-invar}
  Let $G$ be a group $G$.  We say that $\invar(S)$ is a complete polynomial-size invariant for $G$ if the following hold:
  
  \begin{enumerate}
  \item If $G$ and $H$ are groups then $G \cong H$ \ifft $\invar(G) = \invar(H)$.
  \item $\abs{\invar(G)} = \poly(n)$
  \end{enumerate}
\end{definition}

\begin{definition}
  \label{defn:group-can}
  Let $G$ be a group.  We define the canonical form $\can(G)$ to be an $n \times n$ matrix such that:

  \begin{enumerate}
  \item The entries of $\can(G)$ are elements of $[n]$.
  \item $\can(G)$ is the multiplication table for a group which is isomorphic to $G$.
  \item $\can(G)$ is a complete polynomial-size invariant.
  \end{enumerate}
\end{definition}

In the \emph{group canonization problem}, we are given a group $G$ and must compute $\can(G)$.

\begin{theorem}
  \label{thm:gen-can}
  Let $G$ be a group.  We can compute a canonical form $\can(G)$ for $G$ in $n^{\log_p n + O(1)}$ time where $p$ is the smallest prime dividing the order of the group.
\end{theorem}

\begin{proof}
  Let $\cG_\ell$ and $\cH_\ell$ be the collections of all ordered generating sets of $G$ and $H$ of size at most $\ell = \log_p n$ and define $M_\ell(G) = \setb{M_{\bmg}}{\bmg \in \cG_\ell}$.  We compute $M_\ell(G)$ and store the results in a list $\bmA$; this takes $n^{\log_p n + O(1)}$ time.  Because each $M_{\bmg}$ is an $n \times n$ matrix with entries in $[n]$, we can sort $\bmA$ lexicographically in $n^{\log_p n + O(1)}$ time.  Let $M_{\bmg}$ be the first element of $\bmA$.  We define $\can(G) = M_{\bmg}$.  This is the lexicographically least element of $M_\ell(G)$.  We claim that $G \cong H$ \ifft $\can(G) = \can(H)$.

  There exist $\bmg \in \cG_\ell$ and $\bmh \in \cH_\ell$ such that $M_{\bmg} = \can(G)$ and $M_{\bmh} = \can(H)$.  If $\can(G) = \can(H)$, then by part (a) of \lemref{gen-mull-table-perm}, $G$ is isomorphic to the group defined by the multiplication table $M_{\bmg} = M_{\bmh}$ which is isomorphic to $H$.  Conversely, if $G \cong H$ then part (c) of \lemref{gen-mull-table-perm} implies that $M_\ell(G) = M_\ell(H)$ so that $\can(G) = \can(H)$.
\end{proof}

  \newpage
  \section{Algorithms for group canonization}
  \label{app:group-can}
  In this section, we adapt the deterministic variants of our algorithms for group canonization (see \defref{group-can}).  First, we prove an analogue of \thmref{group-red-comp}.

\begin{theorem}
  \label{thm:group-red-comp-can}
  Group canonization is $n^{(1 / 2) \log_p n + O(1)}$ time deterministic Turing reducible to composition-series canonization where $p$ is the smallest prime dividing the order of the group.
\end{theorem}

\begin{proof}
  We compute all of the $n^{(1 / 2) \log_p n + O(1)}$ composition series for $G$ which arise from some choice of socle decompositions in \algref{group-comp} by \lemref{num-comp-choices}.  For each such composition series $S$, we compute canonical form $\can(S)$ and store the result in a list $\bmA$.  We sort $\bmA$ lexicographically in $n^{(1 / 2) \log_p n + O(1)}$ time and select the first element $\can(S)$ of $\bmA$.  By \defref{comp-can}, $\can(S)$ contains a multiplication table $M$ for a group with the underlying set $[n]$ which is isomorphic to $G$.  We define $\can(G) = M(G)$ and claim that this is a canonical form for $G$.

  Let $H$ be a group and suppose that $\can(G) = \can(H)$; then $G$ and $H$ are both isomorphic to the group described by the multiplication table $M(G) = M(H)$ so $G \cong H$.  Conversely, assume that $G \cong H$.  Then for every composition series $S$ for $G$ which is computed by \algref{group-comp} for some choice of socle decompositions, there exists a choice of socle decompositions such that \algref{group-comp} computes a composition series $S'$ for $H$ which is isomorphic to $S$ by \lemref{all-comp-choice}.  Similarly, for every composition series $S'$ for $H$ which is computed by \algref{group-comp} for some choice of socle decompositions, there exists a choice of socle decompositions such that \algref{group-comp} computes a composition series $S$ for $G$ which is isomorphic to $S'$ by \lemref{all-comp-choice}.  It follows that the list $\bmA$ of canonical forms $\can(S)$ of composition series for $G$ is equal to the list $\bmB$ of canonical forms $\can(S')$ of composition series for $H$.  Therefore, $\can(G) = \can(H)$.
\end{proof}

This yields an algorithm for $p$-group canonization.

\begin{theorem}
  \label{thm:p-group-can}
  Let $G$ be a $p$-group.  Then we can compute the canonical form $\can(G)$ for $G$ in $n^{(1 / 2) \log n + O(1)}$ deterministic time.
\end{theorem}

\begin{proof}
  By Theorems~\ref{thm:group-red-comp-can} and \ref{thm:comp-can}, we can perform $p$-group canonization in $n^{(1 / 2) \log_p n + c p}$ time for some constant $c > 0$.  Let us call this algorithm $A$.  From \lemref{small-p-bound}, we know that $n^{(1 / 2) \log_p n + c p} = n^{(1 / 2) \log n + O(1)}$ for $2 \leq p \leq \ln n / 2 c \ln^2 p$.  \lemref{large-p-bound} implies that $n^{\log_p n + O(1)} = n^{O(\log n / \log \log n)}$ for $p \geq  \ln n / 2 c \ln^2 p$.  We can define an algorithm which runs $A$ when $2 \leq p \leq \ln n / 2 c \ln^2 p$ and runs the algorithm of \thmref{gen-can} when $p > \ln n / 2 c \ln^2 p$.  The overall complexity is then $n^{(1 / 2) \log n + O(1)}$.
\end{proof}

Now we adapt our deterministic algorithm for solvable groups to perform canonization.

\begin{theorem}
  \label{thm:sol-red-hcomp-can}
  Solvable-group canonization is $n^{(1 / 2) \log_p n + O(1)}$ time deterministic Turing reducible to Hall composition-series canonization where $p$ is the smallest prime dividing the order of the group.
\end{theorem}

\begin{proof}
  Suppose that we have an algorithm for Hall composition-series canonization.  Let $G$ be a solvable group whose order has the prime factorization $n = \prod_{i = 1}^\ell p_i^{e_i}$.  The result follows from the same argument as \thmref{group-red-comp-can} except that we use Hall composition series instead of composition series.  Let $p$ be the smallest prime which divides the order of $G$.  By \lemref{num-hcomp-choices}, the number of Hall composition series $(G, \cP, \cS)$ for $G$ for some choice of Sylow basis $\cP = \setb{P_i}{1 \leq i \leq \ell}$ and some choice of socle decompositions in \algref{group-comp} for computing each composition series $S_i \in \cS$ for $P_i \in \cP$ is at most $n^{(1 / 2) \log_p n + O(1)}$.  We compute the canonical form of every such Hall composition series $(G, \cP, \cS)$ and store the results in the list $\bmA$.  We sort $\bmA$ lexicographically and select the first element $\can(\cS)$ of $\bmA$.  Recall from \defref{hcomp-can} that $\can(\cS)$ contains a multiplication table $M(G)$ for a group with the underlying set $[n]$ which is isomorphic to $G$.  Let us define $\can(G) = M(G)$.
  
  We can compute a Sylow basis $\cP = \setb{P_i}{1 \leq i \leq \ell}$ for $G$ in $\poly(n)$ time by \lemref{comp-syl-basis-poly}.  All Sylow bases of $G$ are conjugate by \thmref{sol-conj} so once we know a Sylow bases of $G$ the others can be found in polynomial time.  It follows that $\can(G)$ can be computed using $\poly(n)$ preprocessing time and $n^{(1 / 2) \log_p n + O(1)}$ calls to the algorithm for Hall composition-series canonization.
  
  Let $H$ be a solvable group.  We claim that $G \cong H$ \ifft $\can(G) = \can(H)$.  If $\can(G) = \can(H)$ then it follows that $G$ and $H$ are both isomorphic to the group described by the multiplication table $M(G) = M(H)$ so $G \cong H$.  Conversely, let us suppose that $G \cong H$ and fix an isomorphism $\phi : G \ra H$.  For any Sylow basis $\cP = \setb{P_i}{1 \leq i \leq \ell}$ of $G$, there is a Sylow basis $\cQ = \setb{Q_i}{1 \leq i \leq \ell}$ of $H$ such that $\phi$ is an isomorphism from $\cP$ to $\cQ$.  Consider a set $\cS = \setb{S_i}{1 \leq i \leq \ell}$ where each $S_i$ is a composition series for $P_i$ which arises from some choice of socle decompositions in \algref{group-comp}.  Then \lemref{all-comp-choice-syl} implies that there exists a set $\cS' = \setb{S_i'}{1 \leq i \leq \ell}$ where each $S_i'$ is a composition series for $Q_i$ which arises from some choice of socle decompositions in \algref{group-comp} such that $\phi$ is an isomorphism from each $S_i$ to $S_i'$.  Therefore, $\cS \cong \cS'$ so for any Hall composition series that can be constructed for $G$ for some choice of Sylow basis and some choice of socle decompositions, there exists a choice of Sylow basis and some choice of socle decompositions which yields an isomorphic Hall composition series for $H$.  The same argument shows for any Hall composition series that can be constructed for $H$ for some choice of Sylow basis and some choice of socle decompositions, there exists a choice of Sylow basis and some choice of socle decompositions which yields an isomorphic Hall composition series for $G$.  It follows that the list $\bmA$ of the canonical forms of the Hall composition series for $G$ contains the same elements as the list $\bmB$ of the canonical forms of the Hall composition series for $H$.  Therefore, after we sort $\bmA$ and $\bmB$, the two lists are equal so $\can(G) = \can(H)$.
\end{proof}

\begin{theorem}
  \label{thm:sol-group-can}
  Let $G$ be a solvable group.  Then we can compute the canonical form $\can(G)$ for $G$ in $n^{(1 / 2) \log_p n + O(\log n / \log \log n)}$ deterministic time.
\end{theorem}

\begin{proof}
  This follows immediately from Theorems~\ref{thm:sol-red-hcomp-can} and \ref{thm:sol-hcomp-can}.
\end{proof}

  \newpage
  \section{The flaws in Wagner's algorithm}
  \label{app:flaw}
  In this section, we describe the flaws in Wagner's $n^{O(\log n / \log \log n)}$ algorithm~\cite{wagner2011b} for general composition-series isomorphism.  The most serious of these is a dependence on the coset representatives chosen.  As a result, Wagner's algorithm for composition series isomorphism only applies to the $p$-groups (and a slight generalization).  Following communication with the author, Wagner attempted to fix this flaw in later revisions of his paper~\cite{wagner2012a,wagner2012b}; however, the fix described in these versions contains a more subtle variant of the same problem and is therefore incorrect (see \ssecref{wagner-fix} below).  Moreover, the flaw in Wagner's attempted fix~\cite{wagner2012a,wagner2012b} appears to be serious and we suspect that it cannot be salvaged without major changes (see \ssecref{fix-wagner-fix} below).  After further discussions with the author, Wagner retracted his attempted fix~\cite{wagner2012a,wagner2012b} in the latest revision of his paper~\cite{wagner2012c} which scales back his claims so that only \ssecref{other-flaws} of this appendix is applicable.

\subsection{Wagner's algorithm for general composition-series isomorphism}
\label{ssec:wagner-alg}
Before describing Wagner's algorithm~\cite{wagner2011b} for general composition series isomorphism, it is necessary to introduce his modification of the generator-enumeration algorithm for composition series isomorphism.  Let $S$ and $S'$ be composition series consisting of the subgroups $G_0 = 1 \tril \cdots \tril G_m = G$ and $H_0 = 1 \tril \cdots \tril H_m = H$.  Then the generator-enumeration algorithm can be modified for composition-series isomorphism by guessing generators for each factor of $S$.  Since composition factors are simple we now make use of a corollary of the classification of finite simple groups.

\begin{proposition}[\cite{malle1994a,liebeck1995a}]
  Every finite simple group has a generating set of size at most $2$.
\end{proposition}

We can find a generating set of size at most $2$ for any simple group by brute force since there are only $n \choose 2$ possibilities.  Then we find a generating set for each factor of $S$ and choose representatives of the two cosets that generate the factor group.  The resulting set is a generating set $K$ for $G$.  We then consider all mappings from $K$ to $H$.  If $\alpha$ is a lower bound on the order of each composition factor in $S$ then at most $2 \log_{\alpha} n$ elements of $H$ must be guessed so the algorithm runs in $n^{2 \log_{\alpha} n + O(1)}$ time.

Wagner's reduction from $p$-group composition-series isomorphism to graph isomorphism also applies to composition series for general groups.  However, if some of the composition factors have large orders then the resulting graph can have high degree.  In this case, bounded-degree graph isomorphism algorithms will perform poorly.  Wagner's idea was to circumvent this problem by integrating his modification of the generator-enumeration algorithm into the reduction in order to eliminate the high degree nodes.  To do this, fix a parameter $\alpha$.  For each composition factor $G_{i + 1} / G_i$ of $S$ that has order greater than $\alpha$, guess a generating set $K_i$ of size at most $2$ and representatives $a_i, b_i \in G_{i + 1}$ of the two cosets (if there is only one coset then set $a_i = b_i$).  We then construct the graph $X(S)$ as before (see \secref{graph-red}) but then modify it to eliminate nodes of high degree.  For each $\abs{G_{i + 1} / G_i} > \alpha$, remove the edges between the nodes in the sets $G / G_{i + 1}$ and $G / G_i$.  For each $x G_{i + 1}$, consider the set $C$ of nodes $y G_i \subseteq x G_{i + 1}$ to which it was connected.  Wagner noted that the elements of $C$ may be written as $x g G_i$ where $g$ is an element generated by $a_i$ and $b_i$.  The generators $a_i$ and $b_i$ can also be used to define a total order on the elements $g$ that they generate.  Suppose this order is $g_1 \prec \cdots \prec g_k$.  The idea was to order the nodes in $C$ as $x g_1 G_i, \ldots, x g_k G_i$.  A binary tree that has the nodes in $C$ as its leaves can be constructed by attaching a new node to the two leftmost nodes in $C$.  Next, we attach another new node to the next two leftmost nodes and continue until we have run out of nodes.  We continue the process recursively on the binary trees of depth $1$ (and the singleton node on the far right if the cardinality of $C$ is odd).  This yields a binary tree with the nodes in $C$ as its leaves.  The root of this tree is then identified with $x G_{i + 1}$.  Finally, we color the nodes $x a_i G_i$ and $x b_i G_i$ special colors by attaching path graphs of constant length.

We then construct the graph for $S'$ in the same way.  Let our guesses for the representatives of the generators of the factors $H_{i + 1} / H_i$ of $S'$ which have order more than $\alpha$ be denoted $a_i'$ and $b_i'$.

\subsection{The main flaw in Wagner's algorithm}
\label{ssec:main-flaw}
Let us denote by $Y(S)$ and $Y(S')$ the graphs that result from the process described in \ssecref{wagner-alg}.  These graphs have degree $\alpha + O(1)$ so it can be decided if $Y(S) \cong Y(S')$ in $n^{\alpha + O(1)}$ time.  Wagner claimed~\cite{wagner2011b} that there is an isomorphism from $S$ to $S'$ that maps $a_i$ to $a_i'$ and $b_i$ to $b_i'$ for all $i$ \ifft $Y(S) \cong Y(S')$.  The problem is that the ordering $x g_1 G_i, \ldots, x g_k G_i$ that was used to construct the binary trees depends not only on $a_i$ and $b_i$ but also on the coset representative $x$.  For example, if we had chosen the representative $z = x a_i^{-1}$ instead of $x$, then we would obtain a different ordering.

This is unfortunate since if Wagner's result for general composition-series isomorphism were correct, our results would generalize to arbitrary groups.  We remark however that if the orders of the composition factors decrease as we move away from the top group $G$ in the composition series then Wagner's trick for fixing the generators works correctly.  This is because in this case we have guessed all of the generators above each high-degree node in the tree so there is a unique choice of coset representatives.

\subsection{Natural attempts to patch the flaw}
\label{ssec:natural-fix}
We now discuss why various natural methods for fixing the flaw described in \ssecref{main-flaw} do not work.  The most obvious idea is to simply guess all of the coset representatives $x$ for every node $x G_i$.  This yields a correct algorithm which has a runtime that is much slower than the generator-enumeration algorithm.  A slightly better idea is to determine the position of the composition factor $G_{i + 1} / G_i$ in $S$ with cardinality more than $\alpha$ which is the farthest from the top of the composition series; we then guess representatives for the generators of the composition factors from the top of the tree down to this last large composition factor\footnote{This was proposed to us by an anonymous reviewer as a possible fix for Wagner's algorithm.}.  This fixes the problem of the ordering depending on which representatives $x$ of each $x G_i$ are chosen; however, if the last composition factor in $S$ has cardinality greater than $\alpha$ then this will involve guessing representatives for the generators of all composition factors in $S$.  As a result, the worst-case performance of this approach is no better than the generator-enumeration algorithm.

\subsection{Wagner's attempt to fix the flaw}
\label{ssec:wagner-fix}
Following communication with the author, Wagner revised his paper with a proposed fix~\cite{wagner2012a,wagner2012b} for the flaw described in \ssecref{main-flaw}.  The revised construction is nontrivial and adds six pages to the last version without the proposed fix~\cite{wagner2011b}.  Wagner~\cite{wagner2012a,wagner2012b} gives a sketch\footnote{The core of Wagner's argument is that the proof is essentially the same as before; this ignores the differences between the original construction and the revised version.} of a correctness proof for the proposed fix but does not address the details and subtleties that would arise in a rigorous proof.  Unfortunately, the proof is incorrect and the proposed fix fails to truly repair the flaw.  The result is a new algorithm which fails due to a more subtle dependence on the coset representatives which are chosen.  Moreover, we shall argue that the revised construction~\cite{wagner2012a,wagner2012b} breaks the structure of the group so it appears that it is not much closer to yielding an algorithm for general composition-series isomorphism than Wagner's previous algorithm~\cite{wagner2011b}.  In our work, we avoid this problem by exploiting the additional structure which exists in solvable groups.

Wagner's idea for fixing the flaw is that the socle can be written as the direct product $\soc(G) = M_1 \times N_1$ where $M_1$ contains the parts of the socle which result in large composition factors and $N_1$ contains the parts of the socle which result in small composition factors\footnote{It is easy to compute $M_1$ and $N_1$ by calculating each minimal normal subgroup $L_i$ of $G$.  Every minimal normal subgroup $L_i$ can be written as the direct product $\bigtimes_j S_{ij}$ where each $S_{ij}$ is simple and $S_{ij} \cong S_{ik}$ for all $j$ and $k$.  Then we can compute $M_1 = \bigtimes_{i : \abs{S_{ij}} > \alpha} S_{ij}$ and $N_1 = \bigtimes_{i : \abs{S_{ij}} \leq \alpha} S_{ij}$}.  Similarly, $\soc(G / \soc(G)) = M_2 \times N_2$ where $M_2$ consists of the parts of $\soc(G / \soc(G))$ which correspond to large composition factors and $N_2$ corresponds to the parts which correspond to small composition factors.  In general, it is necessary to continue this process recursively.  However, in our case we shall assume that $\soc(G / \soc(G)) = G / \soc(G)$ since this suffices to illustrate the flaw.

We can construct a composition series $S$ for $G$ which starts with the composition factors for $M_2$, then $N_2$, then $M_1$ and finally $N_1$.  The problem with this is that we need to apply the generator-fixing trick to the composition factors which correspond to $M_2$ and $M_1$ without fixing the generators for the composition factors which correspond to $N_2$.  However, it isn't clear how to do this since we must fix generators for a continuous sequence of composition factors starting from the top of the composition series.  Fixing generators for only $M_1$ and $M_2$ but not $N_2$ runs into the problem described in \ssecref{main-flaw}.  Wagner's idea for fixing the flaw was to move $M_1$ to the top of the tree so that $M_1$ and $M_2$ come before $N_1$ and $N_2$.

Let $x_i$ be a complete set of coset representatives for $G / \soc(G)$.  Then any element of $g \in G$ can be written as $g = x_i y z$ for some $i$ where $y \in M_1$ and $z \in N_1$.  The revised construction is inspired by the fact that $g = y^{x_i} x_i z$ where $y^x = x y x^{-1}$ denotes conjugation.  We construct a composition series for $M_1$ and compute the corresponding tree $T$.  The leaf nodes of $T$ are labeled by the elements of $M_1$.  For each leaf $y \in M_1$ of $T$, we attach a copy $T_y$ of the tree which corresponds to a composition series for $G / \soc(G)$.  The nodes of each tree $T_y$ are labelled by the elements of $G / \soc(G)$.  To each leaf node $x_i \soc(G)$ of each $T_y$, we attach a copy $T_{y, x_i \soc(G)}$ of the tree for a composition series of $N_1$.  The leaves of each tree $T_{y, x_i \soc(G)}$ are labelled by the elements of $N_1$ so now we must relabel them so that the leaves of all the trees $T_{y, x_i \soc(G)}$ represent the elements of $G$.

Wagner does this by setting the label of the leaf $z \in N_1$ in the tree $T_{y, x_i \soc(G)}$ to $y x_i z$.  The problem is that there is a strong dependence on the representatives $x_i$.  Let us fix some $j$ and suppose that we replace $x_j$ with $\bar x_j \in x_j \soc(G)$.  Consider the element $g = y x_j z$ where $y \in M_1$ and $z \in N_1$.  Then $g$ corresponds to the node for $z$ in the tree $T_{y, x_j \soc(G)}$ when we use the original coset representatives $x_i$ for $G / \soc(G)$.  However, when $x_j$ is replaced by $\bar x_j$, then for some $\bar y \in M_1$ and $\bar z \in N_1$ we have

\begin{align*}
  g  & = y \bar x_j \bar y \bar z z \\
  {} & = y {\bar y}^{\bar x_j} \bar x_j \bar z z \\
  {} & = \dot y \bar x_j \dot z
\end{align*}
where $\dot y = y {\bar y}^{\bar x_j}$ and $\dot z = \bar z z$.  Thus, when $x_j$ is replaced by $\bar x_j$, the node which corresponds  $g$ is ripped out of the tree $T_{y, x_j \soc(G)}$ and associated with the node for $\dot z$ in the tree $T_{\dot y, \bar x_j \soc(G)} = T_{\dot y, x_j \soc(G)}$.  Moreover, $\dot y \in M_1$ and $\dot z \in N_1$ can be arbitrary depending on the choice of the representative $\bar x_j \in x_j \soc(G)$.  Thus, Wagner's revised construction for general groups does not respect the structure of $G$ unless the coset representatives $x_i$ are chosen in a consistent manner.

The above construction can be made to work if we guess a complete set of representatives for the factor group $G / \soc(G)$ rather than just for the subgroups $M_i$.  In fact, it is only necessary to guess representatives of a generating set of $G / \soc(G)$ since the representatives of a generating set determine a complete set of representatives.  However, the resulting algorithm for composition-series isomorphism no better than the generator-enumeration algorithm in the worst case.  This is in sharp contrast to the $n^{O(\log n / \log \log n)}$ runtime obtained for $p$-group composition-series isomorphism.

\subsection{Can Wagner's attempted fix be fixed?}
\label{ssec:fix-wagner-fix}
Given the conclusions of Subsections~\ref{ssec:natural-fix} and \ref{ssec:wagner-fix}, a natural question is whether the ideas in Wagner's attempt to fix the flaw~\cite{wagner2012a,wagner2012b} can be made to work with minor changes.  We will now present intuitive (but non-rigorous) arguments that the proposed fix~\cite{wagner2012a,wagner2012b} fundamentally breaks the structure of the group and therefore cannot be salvaged without major changes.  In our algorithm for solvable groups, we avoid this problem by using different techniques which exploit additional structure in solvable groups that manifests itself in the form of Sylow bases.

To illustrate our arguments, we shall restrict our attention to a semidirect product $\bbZ_p \rtimes \bbZ_q$ where $p$ is a large prime and $q$ is a small prime.  Isomorphism testing for this class of groups can be done efficiently using specialized methods~\cite{legall2008a}; however, one can imagine more complex groups for which the similar problems arise\footnote{For example, suppose that $\bbZ_p$ and $\bbZ_q$ are replaced by non-Abelian simple groups.}.  Intuitively, applying Wagner's fix to this group would imply that we could reverse the semidirect symbol so that $\bbZ_p \rtimes \bbZ_q$ is the same as some semidirect product $\bbZ_p \ltimes \bbZ_q$.  Of course, this is not possible for general semidirect products which suggests that Wagner's fix~\cite{wagner2012a,wagner2012b} cannot be salvaged without significant new ideas.

It is important to note that the arguments in this subsection are non-rigorous and it is not clear if they can be made more precise.  However, any algorithm which is based on composition-series isomorphism must avoid breaking the structure of the group in this way.  Our algorithm for the solvable groups surmounts this obstacle by utilizing special properties of solvable groups (namely Sylow bases) in the construction of the graph.

\subsection{Other flaws}
\label{ssec:other-flaws}
Wagner's algorithm~\cite{wagner2011b,wagner2012a,wagner2012b,wagner2012c} also contains a second less serious flaw that can be fixed.  In Wagner's algorithm~\cite{wagner2011b} for constructing a composition series, he implicitly assumes that the socle of $G$ is the direct product of \emph{all} of the minimal normal subgroups\footnote{The three most recent revisions of Wagner's paper~\cite{wagner2012a,wagner2012b,wagner2012c} correct this error in the algorithm but do not properly account for the cost of choosing a subset of the minimal normal subgroups in the runtime.} of $G$.  This is incorrect as one can see by considering the socle of the group $\bbZ_p^k$.  However, since Wagner represents the composition series using generators, this problem can be corrected without significantly affecting the runtime of his algorithm.  In our case, we represent composition series as the subsets of $G$ that correspond to the intermediate subgroups and the fact that the socle of $G$ is not the direct product of all of the minimal normal subgroups of $G$ has a significant impact on the runtime.  Indeed, if the socle of $G$ was the direct product of all minimal normal subgroups of $G$ we could improve our reduction from group isomorphism to composition-series isomorphism to run in $n^{O(\log \log n)}$ time.

  \newpage
  \section{Group theory proofs}
  \label{app:group-res}
  In this section, we provide more results from group theory and fill in some of the proofs omitted in \secref{group-back}.

\begin{proposition}
  \label{prop:ch-trans}
  If $K \ch H \ch G$ then $K \ch G$.
\end{proposition}

\begin{proof}
  Let $\phi \in \aut(G)$.  Then $\phi[H] = H$ so $\restr{\phi}{H} \in \aut(H)$.  Therefore, $\restr{\phi}{H}[K] = K$ so $K \ch G$.
\end{proof}

\begin{proposition}
  \label{prop:ch-norm-trans}
   If $H \ch N \trile G$ then $H \trile G$.
\end{proposition}

\begin{proof}
  Let $x \in H$ and $g \in G$.  Clearly, $\iota_g(x) = x^g \in N$ so $\restr{\iota_g}{N} \in \aut(N)$.  Therefore, $\restr{\iota_g}{N}[H] = H$ and $H \trile G$.
\end{proof}

The following proposition follows from \propref{ch-norm-trans}.

\mnscharsim*

\begin{proposition}
  \label{prop:norm-hom}
  Let $G$ and $H$ be groups and let $\phi : G \ra H$ be a homomorphism.
  
  \begin{enumerate}
  \item If $K \trile G$ then $\phi[K] \trile \phi[G]$.
  \item If $L \trile \phi[G]$ then $\phi^{-1}[L] \trile G$.
  \end{enumerate}
\end{proposition}

\begin{proof}
  To prove part (a), suppose that $K \trile G$ let $x \in K$ and $g \in G$.  Then $x^g \in K$ so $\phi(x)^{\phi(g)} = \phi(x^g) \in \phi[K]$ and $\phi[K] \trile \phi[G]$.  For part (b), suppose $L \trile \phi[G]$ and let $x \in \phi^{-1}[L]$ and $g \in G$.  Then $\phi(x^g) = \phi(x)^{\phi(g)} \in L$ so $x^g \in \phi^{-1}[L]$.
\end{proof}

\begin{corollary}
  \label{cor:norm-factor}
  Let $N$ be a normal subgroup of a group $G$ and $N \leq H \leq G$.  Then $H / N \trile G / N$ \ifft $H \trile G$.
\end{corollary}

\begin{proof}
  Let $\varphi : G \ra G / N$ be the canonical map.  By \propref{norm-hom}, if $H \trile G$ then $\varphi[H] = H / N \trile G / N$.  Similarly, if $H / N \trile G / N$ then $\varphi^{-1}[H / N] = H \trile G$ by \propref{norm-hom}.
\end{proof}

\begin{definition}
  Let $G$ and $H$ be groups and let $\phi : G \ra H$ be a homomorphism.  Then the kernel of $\phi$ is denoted by $\ker \phi = \setb{x \in G}{\phi(x) = 1}$.
\end{definition}

\begin{proposition}
  Let $G$ and $H$ be groups and let $\phi : G \ra H$ be a homomorphism.  Then $\ker \phi \trile G$.
\end{proposition}

\begin{theorem}[First Isomorphism Theorem, cf.~\cite{rotman1995a,robinson1996a,lang2002a,rose2009a,roman2011a}]
  Let $G$ and $H$ be groups and let $\phi : G \ra H$ be a homomorphism.  Then $G / \ker \phi \cong \phi[G]$.
\end{theorem}

The second and third isomorphism theorems are actually corollaries of the first isomorphism theorem.

\begin{theorem}[Second Isomorphism Theorem, cf.~\cite{rotman1995a,robinson1996a,lang2002a,rose2009a,roman2011a}]
  Let $H$ be a subgroup of a group $G$ and let $K$ be a normal subgroup of $G$.  Then $H \times K / K \cong H / H \cap K$.
\end{theorem}

\begin{theorem}[Third Isomorphism Theorem, cf.~\cite{rotman1995a,robinson1996a,lang2002a,rose2009a,roman2011a}]
  Let $H$ and $K$ be normal subgroups of a group $G$.  Then $\frac{G / H}{K / H} \cong G / K$.
\end{theorem}

\charsimprod*

\begin{proof}
  Suppose $G$ is characteristically simple.  Choose a minimal normal subgroup $N_0$ of $G$.  If $N_0 = G$ then $G$ is simple.  Consider a direct product $H = \bigtimes_i N_i$ of isomorphic minimal proper normal subgroups of $G$.  If $H = G$, then we note that each $N_i$ is characteristically simple.  We thus argue by induction that each $N_i$ is a direct product of isomorphic simple groups and we are done.  Otherwise, $H \tril G$ so there exists $\phi \in \aut(G)$ such that $\phi[H] \not\leq H$.  Then for some $i$, $N = \phi[N_i] \not\leq H$; because $N$ is a minimal normal subgroup of $G$, $N \cap H = 1$.  Since $H \trile G$, $H \times N$ is a direct product of minimal normal subgroups of $G$.  Thus, $G$ may be written as a direct product of minimal normal subgroups and the result follows by induction on the subgroups $N_i$.
\end{proof}

\socprod*

\begin{proof}
  Let $G$ be a group.  If $N_1$ and $N_2$ are minimal normal subgroups of $G$, then $N_1 \cap N_2 = 1$.  Also, the elements of $N_1$ and $N_2$ commute so they form a direct product that is a subgroup of the socle.  By continuing this process, we can write the socle as a direct product of minimal normal subgroups of $G$.
\end{proof}

\soccomp*

\begin{proof}
  We assume $G$ is nontrivial; if $G$ is trivial then the socle and minimal normal subgroups of $G$ are also trivial.  For each nontrivial $x \in G$, we test if $x$ the normal closure $\langle x \rangle^G$ of $x$ is a minimal normal subgroup.  Now $\langle x \rangle^G \trile G$ so we check if $\langle x \rangle^G$ is minimal.  This is done by considering each nontrivial $y \in \langle x \rangle^G$ and computing $\langle y \rangle^G$ which is the smallest normal subgroup of $G$ that contains $y$.  If $\langle x \rangle^G = \langle y \rangle^G$ for all $y \in \langle x \rangle^G$ then $\langle x \rangle^G$ is minimal; otherwise, it is not minimal.  Once we have computed the set $M$ of all minimal normal subgroups of $G$, we simply take the subgroup generated by the minimal normal subgroups of $G$ which by definition is the socle of $G$.
\end{proof}

\begin{theorem}[Jordan-H{\"{o}}lder, cf.~\cite{rotman1995a,robinson1996a,rose2009a}]
  The multiset of isomorphism classes of the composition factors of a group is an isomorphism invariant.
\end{theorem}

The following theorem can be proved using the class equation.

\begin{theorem}
  The center of every $p$-group is nontrivial.
\end{theorem}

\begin{corollary}
  \label{cor:p-sim}
  Every simple $p$-group is isomorphic to $\bbZ_p$.
\end{corollary}

\begin{proof}
  Let $G$ be a simple $p$-group.  Then its center $Z(G)$ is a nontrivial Abelian group.  It follows that there is an element $x \in Z(G)$ that has order $p$.  The group $H = \langle x \rangle$ is contained in the center and is therefore normal in $G$.  Since $H$ is nontrivial, $G = H \cong \bbZ_p$.
\end{proof}

\inshortorlong{\firstsylowenv}{\firstsylow*}

One can in fact prove a stronger version of the First Sylow Theorem~\cite{robinson1996a,lang2002a}.

\begin{theorem}
  \label{thm:first-syl-max}
  The Sylow $p$-subgroups of $G$ are the maximal $p$-subgroups of $G$.
\end{theorem}

\inshortorlong{\secsylowenv}{\secsylow*}

\begin{theorem}[Third Sylow Theorem, cf.~\cite{rotman1995a,robinson1996a,lang2002a,rose2009a,artin2010a,roman2011a}]
  The number of Sylow $p$-subgroups of $G$ is equal to $1$ mod $p$.
\end{theorem}

\begin{proposition}
  Every composition factor of a $p$-group has order $p$.
\end{proposition}

\begin{proof}
  This follows from \corref{p-sim}.
\end{proof}

\begin{proposition}
  \label{prop:sim-max}
  Let $G$ be a group and let $N$ be a normal subgroup.  Then $G / N$ is simple \ifft $N$ is a maximal normal subgroup of $G$.
\end{proposition}

\begin{proof}
  Suppose $N$ is a maximal normal subgroup of $G$ and $H / N \trile G / N$ for some $N \leq H \leq G$.  Then by \corref{norm-factor}, $H \trile G$ so $H \in \{N, G\}$.  It follows that $G / N$ is simple.  If $G / N$ is simple then consider $N \trile H \trile G$.  By \corref{norm-factor}, $H / N \trile G / N$ so $H / N \in \{\{N\}, G / N\}$.  Thus, $H \in \{N, G\}$ so $N$ is a maximal normal subgroup of $G$.
\end{proof}

\begin{proposition}[cf.~\cite{robinson1996a,rose2009a,roman2011a}]
  Let $S$ denote the subnormal series $G_0 = 1 \tril \cdots \tril G_m = G$.  Then $S$ is a composition series \ifft every factor group $G_{i + 1} / G_i$ is simple.
\end{proposition}

\begin{proof}
  By \propref{sim-max}, $G_{i + 1} / G_i$ is simple \ifft each $G_i$ is a maximal normal subgroup of $G_{i + 1}$; this is precisely the definition of a composition series.
\end{proof}

It is easy to show that $p$-group isomorphism is equivalent to nilpotent-group isomorphism.  We prove this with the following three propositions.

\inshortorlong{\compsylowenv}{\compsylow*}

\begin{proof}
  We start by computing the largest power $p^e$ of $p$ which divides $n$.  By definition, each Sylow $p$-subgroup of $G$ is of order $p^e$.  From \thmref{first-syl-max}, an equivalent definition is that a Sylow $p$-subgroup is a maximal $p$-subgroup of $G$.  We start by choosing a nontrivial element $g_1 \in G$ of order a power of $p$ and computing the subgroup $K_1 = \langle g_1 \rangle$ generated by $g_1$.  Then $K_1$ is contained in some Sylow $p$-subgroup of $G$.  If $\abs{K_1} = p^e$ then we have a Sylow $p$-subgroup of $G$; otherwise, we choose a nontrivial element $g_2 \in G \setminus K_1$ of order a power of $p$ which generates a subgroup $K_2 = \langle g_1, g_2 \rangle$ of order a power of $p$.  By continuing this process until we obtain $\abs{K_j} = p^e$, we can compute a Sylow $p$-subgroup in polynomial time.
\end{proof}

\begin{proposition}[cf.~\cite{rotman1995a,robinson1996a,holt2005a,rose2009a,roman2011a}]
  \label{prop:nilpotent}
  A group is nilpotent \ifft every Sylow subgroup is normal \ifft it is a direct product of its Sylow subgroups.
\end{proposition}

\begin{proposition}
  \label{prop:nil-red-p}
  Nilpotent-group isomorphism is Turing reducible to $p$-group isomorphism where each $p$ divides $n$.
\end{proposition}

\begin{proof}
  Let $G$ and $H$ be nilpotent groups.  Then the Sylow subgroups of $G$ and $H$ can be computed in time polynomial in $n$ by \propref{comp-sylow}.  If the sets of the orders of the Sylow subgroups are different in $G$ and $H$, we return $G \not\cong H$.  We note that a group is nilpotent \ifft it is the direct product of its Sylow subgroups by \propref{nilpotent}.  Then we have the direct products $G = \bigtimes_{i = 1}^{\ell} P_i$ and $H = \bigtimes_{i = 1}^{\ell} Q_i$ where $P_i$ and $Q_i$ are Sylow $p_i$-subgroups of $G$ and $H$.  Any isomorphism $\phi : G \ra H$ must satisfy $\phi[P_i] = Q_i$.  Conversely, if $\phi_i : P_i \ra Q_i$ is an isomorphism for each $i$ then we can easily define an isomorphism $\phi : G \ra H$ using the direct product structure of $G$ and $H$.  We have shown that $G \cong H$ \ifft $P_i \cong Q_i$ for all $i$.  Since the number of Sylow subgroups of $G$ is at most $\log n$, we can decide if $G \cong H$ by making at most $\log n$ calls to our algorithm for $p$-group isomorphism.
\end{proof}

  \newpage
  \section{Composition series proofs}
  \label{app:comp-series-proofs}
  We now provide the proofs omitted in \secref{comp-series}.

\indisowell*

\begin{proof}
  Let $G$ and $H$ be groups with normal subgroups $N$ and $N'$.  Let $\phi : G \ra H$ be an isomorphism such that $\phi[N] = N'$.  We see that $\indphi{G / N}$ is well-defined since if $g_1 N = g_2 N$ we have

  \begin{align*}
    \indphi{G / N}(g_1 N) & = \phi(g_1) N' \\
                       {} & = \phi(g_2 n) N' \text{ for some } n \in N \\
                       {} & = \phi(g_2) N' \\
                       {} & = \indphi{G / N}(g_2 N)
  \end{align*}
  
  It is clear that $\indphi{G / N}$ is injective since $\phi$ is injective; it is also easy to show that $\indphi{G / N}$ is surjective so $\indphi{G / N}$ is bijective.  Clearly, $\indphi{G / N}$ is a homomorphism so it is an isomorhpism.
\end{proof}

\respseries*

\begin{proof}
  Let $\phi : G \ra H$ be an isomorphism where $\phi[N] = N'$ and suppose $\indphi{G / N} : G / N \ra H / N'$ respects $T$ and $T'$.  We note that each $K_i = G_i / N$ and $L_i = H_i / N'$.  By assumption, $\indphi{G / N}[K_i] = L_i$ for all $i$ so

  \begin{align*}
    \phi[G_i] & = \bigcup_{x \in G_i} \phi[x N] \\
           {} & = \bigcup_{x \in G_i} \indphi{G / N}(x N) \\
           {} & = \bigcup_{y \in H_i} y N' \\
           {} & = H_i
  \end{align*}
  
  Thus, $\phi$ respects $S$ and $S'$.  For the converse, suppose $\phi$ respects $S$ and $S'$.  Let $x N \in K_i$; then $x \in G_i$ so $\indphi{G / N}(x N) = \phi(x) N' \in H_i / N' = L_i$ and $\indphi{G / N}[K_i] \subseteq L_i$.  Since $\indphi{G / N}$ is an isomorphism and $\abs{K_i} = \abs{L_i}$, it follows that $\indphi{G / N}[K_i] = L_i$.
\end{proof}
}

\inlong{\newpage}
\bibliographystyle{initials}
\bibliography{$HOME/LaTeX/computer-science-references,$HOME/LaTeX/math-references,$HOME/LaTeX/quantum-computing-references} 

\end{document}